\documentclass{siamltex}

\usepackage[breaklinks,colorlinks=true,linkcolor=blue,citecolor=red,backref=page]{hyperref}

\usepackage{amsfonts} 
\usepackage{amssymb,amsmath} 
\usepackage{graphicx,color}
\usepackage{fancybox}

\DeclareMathAlphabet{\itbf}{OML}{cmm}{b}{it}

\def\EE{\mathbb{E}}
\def\PP{\mathbb{P}}

\def\RR{\mathbb{R}}

\def\eps{\varepsilon}

\def\H{{ \tilde{\rm H}}}
\def\N{{ \tilde{N}}}

\newcommand{\nn}{\nonumber}

\newcommand{\ch}{\check{h}}

\newcommand{\ea}{\end{eqnarray}}  
\newcommand{\ba}{\begin{eqnarray}}  
\newcommand{\ee}{\end{equation}}  
\newcommand{\be}{\begin{equation}}  
\newcommand{\ean}{\end{eqnarray*}}  
\newcommand{\ban}{\begin{eqnarray*}}

\begin{document}  
\bibliographystyle{plainnat}
\title{\
Optimal hedging under fast-varying stochastic volatility
  }

\author{Josselin Garnier\footnotemark[1]
 and Knut S\O lna\footnotemark[2]}   

\maketitle

\renewcommand{\thefootnote}{\fnsymbol{footnote}}

\footnotetext[1]{Centre de Math\'ematiques Appliqu\'ees,
Ecole Polytechnique, 91128 Palaiseau Cedex, France
{\tt josselin.garnier@polytechnique.edu}}

\footnotetext[2]{Department of Mathematics, 
University of California, Irvine CA 92697
{\tt ksolna@math.uci.edu}}

\renewcommand{\thefootnote}{\arabic{footnote}}

\maketitle

\begin{abstract}
 {  In a market with a rough or Markovian mean-reverting stochastic
volatility there is no perfect hedge.  Here it is shown how various 
delta-type hedging  strategies perform and can be evaluated  
in such markets in the case of European options.  
A precise characterization of the hedging cost,
the replication cost caused
by the volatility fluctuations,  is presented in an asymptotic regime
of rapid mean reversion for the volatility fluctuations.  
The optimal dynamic asset based hedging 
 strategy in the considered regime is identified as the so-called 
 ``practitioners'' delta hedging scheme. 
It is moreover shown that the  performances of the delta-type hedging schemes are
essentially independent of the regularity of the volatility paths in the 
considered regime and that the hedging costs are related to a  Vega 
risk martingale whose magnitude is proportional to a new market risk parameter.
It is also shown via numerical simulations that the proposed hedging schemes 
which derive from option price approximations in the  regime of rapid mean reversion, 
 are robust:
the  ``practitioners'' delta hedging scheme that is identified as being optimal by our asymptotic analysis when the mean reversion time is small
 seems to be optimal  with arbitrary mean reversion times.}  
%that is identified and 
%which must be calibrated in a hedging context,  this follows since this risk parameter  
%it is not a part of the pricing theory associated with such models.
\end{abstract}



\section{Introduction}

We consider an incomplete market with stochastic volatility model for the underlying. 
Our main objective is to characterize the performance of option hedging schemes
in such markets.  
The rather general class of stochastic volatility models that we consider 
incorporates standard Markovian volatility models and 
also rough volatility models that have received a lot of attention recently, see
 \cite{alos2,viens3,gatheral1,jacquier,bayer,F15} and the literature reviews  in
 \cite{sv2,sv3}. 
 In the context of portfolio optimization Markovian models 
 have been considered for instance in \cite{fsz},  while recently
 the non-Markovian case was considered in \cite{fH,fH2,fH3}.
   
Here  we model the volatility as a smooth function of a 
volatility factor that is a stationary Volterra type Gaussian process.
In the standard volatility model the volatility factor is a mean-reverting Markov 
process such as an Ornstein-Uhlenbeck process.
In the  rough volatility  model the correlation function of 
the volatility factor decays rapidly at the origin, faster than the decay
associated with a Markov process, producing  rough paths.  
The decay rate is characterized by 
the Hurst exponent $H$.  
The Gaussian volatility factor may be chosen for instance as a fractional 
Ornstein-Uhlenbeck process with Hurst exponent $H < 1/2$. 
 The main asymptotic context that we consider  is a rapidly 
 mean-reverting volatility situation.  The results presented here build  on
 and extend  those presented in \cite{sv3} regarding 
 option pricing for such models. 
 Here we extend this framework to a more general
 class of volatility models and analyze the performance of 
 a large class of hedging strategies for European options
 that we call dynamic asset (DA) based hedging schemes.
 A DA scheme is based on a replicating portfolio 
made of some number of underlyings and some amount in the bank account.
In particular, this class contains the  ``delta'',
 $\delta$, hedging strategies, in which the number 
 of underlyings in the portfolio is  the $\delta$ of the price,  that is,  the partial derivative
 of the option price with respect to the underlying price.   
  For the classic Black-Scholes model 
 with a constant volatility this strategy makes it possible to trade in a self-financing manner
 in the underlying and the bank account to perfectly replicate the 
 payoff of the option.  
 In the situation when the volatility is stochastic such a scheme 
 accumulates extra cost during the lifetime of the option due to the fluctuations
 in the volatility. We consider  here two main market situations: (I)
 the option trades at the Black-Scholes option price at the ``effective
 volatility'' or a {\it Black-Scholes market}, this is discussed in Section \ref{sec:delta1};  
 (II)  the market incorporates the effects of rapid volatility fluctuations
 and trades at a corrected price or a {\it corrected market}, 
 this is discussed in Section \ref{sec:delta2}.
 Here (I) the effective volatility refers to  the root mean square of the volatility
 process averaged with respect to the invariant distribution of the
 volatility factor and (II) the corrected price refers to the Black-Scholes
 price at the effective volatility with a correction which follows
 from an asymptotic analysis of the rapidly mean-reverting situation,
 see Proposition \ref{prop:main}.
 We  assume that the mean reversion time of the volatility
 factor is small relative to the diffusion  time of the underlying price. 
 We remark that the distinction between the market situations (I) and (II) is
 important in the case  of early exercise. 
  Note  moreover that we consider several canonical ways of computing 
  the effective $\delta$ of the replication strategy.  These are described
  in more detail below. 
  In the case that ``vol-of-vol'' is zero,
  that is in the limit of small volatility fluctuations, 
   these $\delta$s  become the standard
  Black-Scholes $\delta$ and the hedging strategies become the standard 
  self-financing replicating strategy.   
  In  the case of a fluctuating volatility 
  we present here a novel and precise characterization
   of the  extra  {\it hedging cost} that accumulates
   due to the fluctuations.             
 For  the strategy (I) this extra cost is semimartingale with 
 in general a non-zero mean and variance that we quantify, while
 for  the strategy
 (II) the extra cost is a {\it true martingale} and we compute its variance. 
  We   compute the costs for the DA hedging strategies and we 
  identify the optimal  hedging strategy within the DA class that minimizes 
  the variance of the hedging cost in our  regime.
%  In Section \ref{sec:sum} we summarize the main results.
%   This cost is parameterized by three effective
%  market parameters and  we comment on their calibration in Section \ref{sec:est}. 
  We allow for early exercise when
  evaluating the cost and we show how the cost depends on the relative 
  exercise time.
  It is important to note that our results are universal
 in that they hold for both  rough ($H<1/2$) and classic Markovian  
 stochastic volatility factors in the regime of rapid mean reversion.  
However, in a regime of slow, rather than fast mean  reversion, or when 
$H > 1/2$,  this picture  changes qualitatively and results regarding these regimes
will be presented elsewhere.   Note,  moreover,  that  we here consider the case with
``leverage'', which means that the volatility factor is correlated with the Brownian motion driving
the underlying price.  In fact,
 in the situation with zero correlation all the hedging approaches coincide 
 and the cost is characterized fully by the  Vega risk  martingale.

The role of stochastic volatility for delta 
hedging schemes in the uncorrelated  case has been discussed in \cite{touzi}.
Underhedged and overhedged situations are discussed there
and  we  revisit such a characterization
here in the correlated  case.    
Superheging schemes   provide an upper bound 
for the replication cost  \cite{PH,touzi2}.
Here we present  a statistical characterization of the hedging cost 
 which can be used for  a ``value at risk'' type  characterization of the hedging cost.
 When stochastic volatility is mixing and
rapidly mean-reverting 
the hedging cost was discussed in \cite{sircar00} in the case without leverage and in 
\cite{fouque00} in the case with leverage. 
We extend here this discussion to get explicit expressions for the 
hedging cost and consider more general DA hedging schemes. 
While we here consider hedging schemes with a view toward minimizing
replication cost, portfolio construction from the point of view of utility optimization
is discussed in \cite{fsz} in the context of stochastic volatility in various
asymptotic regimes.
Our objective is indeed  to characterize analytically the performance 
of classic (including delta) hedging schemes which plays an important
role in practical risk mitigation schemes \cite{P1}. 
In \cite{HWnew} the importance of the leverage in determining  risk in
 hedging schemes is emphasized and  explored from an
empirical perspective. Here we give an analytic description
of hedging risk (mean and variance of the hedging cost) 
in particular for the delta hedging  schemes discussed in \cite{HWnew} in the context of leverage and 
rapid mean reversion.    

 {\it Outline of paper:}
 First, in Section \ref{sec:sum} we summarize the main result of the paper.
Then, in Section \ref{sec:svmodel},  we discuss the details of the 
modeling of the market with a fast 
mean-reverting stochastic volatility and in Section \ref{sec:price} 
we give the leading order
stochastic volatility  
price  correction for  a European option in this model. 
 { Note that when we refer to  { leading order}  below we refer
to terms of order $\sqrt{\eps/T}$ or larger with $\eps$ being the mean reversion
time of the volatility factor and $T$ the time to maturity. 
Then we present the main result of the paper in Section
\ref{sec:hedging}  
on the characterization of the hedging costs
for the various  hedging schemes  that we consider. 
 The hedging strategies are  computed relative to  leading order price approximations,
which closely approximates the price  in the asymptotic regime we consider 
with $\eps/T \ll 1$.   
 We discuss in more detail the  main effective parameters 
 that are necessary to implement the strategies and those that characterize the hedging costs 
 in Sections \ref{sec:est} and  \ref{sec:expou}.
We specialize to   the  case  of a call option in Section \ref{sec:call}
 and we present numerical illustrations of the asymptotic results. In Section \ref{sec:sim}
 we present some  Monte Carlo  simulations where we compute the actual hedging
 costs for the various hedging schemes    in the case of call options. 
 We find that the hedging schemes we have set forth based on the asymptotic theory
 in the regime of rapid mean reversion perform well also when the mean reversion time 
 is of the same order as the time to maturity.
   We finally provide some concluding remarks in Section \ref{sec:concl}.   }
 
\section{Summary of Main Results}
\label{sec:sum}

We consider in this section hedging of a European option with payoff $h(X_T)$ with 
$T$ the maturity and $X_t$ the underlying. The underlying is assumed 
to follow a diffusion process with a stochastic volatility as described in Section 
\ref{sec:svmodel},  Eq.~(\ref{eq:Xdef}). In this paper we do not consider
short rate effects,  corresponding to assuming as numeraire  the zero 
coupon bond with maturity~$T$. Moreover, we do not consider effects associated
with dividends, transaction cost or  
 market price of volatility risk. 
An important assumption is, however, that we assume a non-zero ``leverage'',
which means that the volatility factor is driven by a Brownian motion that is correlated with the 
Brownian motion driving the underlying, see Eq. (\ref{eq:corr}) below.
Our main objective is to identify analytically the {\it hedging cost}.
We assume a regime where the mean reversion time of the 
volatility factor is  small relative to the diffusion time of the underlying which is
on the scale of the maturity $T$, that is, we consider a {\it rapidly mean-reverting stationary}
volatility.  We present  asymptotic results
in the regime of rapid mean reversion and below we make precise the sense
of the approximation.
Our class of volatility models  incorporates standard Markovian volatility models and 
rough volatility models.

 {Let the root mean square or ``historical'' volatility be denoted by
$\bar\sigma$ (see Eq.~(\ref{def:barsigma})  below for the  definition). } 
Moreover, let $Q^{(0)}(t,x;\sigma)$ be the standard Black-Scholes 
(European option) price at 
volatility level $\sigma$ evaluated at time $t$ and current value 
$x$ for the underlying.  Then the price  that incorporates 
the leading order correction due to the rapidly mean-reverting  stochastic
volatility is:
\begin{equation}
\label{eq:Pcorrected}
   P(t,x )  
    =  Q^{(0)}(t,x; \bar\sigma) + 
   D(T-t) \big( x \partial_x  (x^2 \partial_x^2 ) \big) Q^{(0)}(t,x;\bar\sigma) ,
%   \\
%    &\equiv&  Q^{(0)}(t,x; \bar\sigma) + 
%   D Q^{(1)}(t,x;\bar\sigma)    .
\end{equation}
see Section \ref{sec:price}.
Here $D$ is an {\it effective pricing parameter} that can be calibrated 
from observations of the implied volatility skew, see 
  Section \ref{sec:est}.

%We consider then two main market contexts:
%
%\medskip
%
%(I)   Construct a hedging portfolio whose value is $Q^{(0)}(t,x;\bar\sigma)$
%and assume the underlying option can be traded at this  Black-Scholes 
%price evaluated at the {\it historical volatility}. 
%
%(II)   Construct a hedging portfolio whose value is
% $P(t,x;\bar\sigma)$ and 
%assume that the underlying option can be traded at this  {\it corrected price}. 
% 
%Note that in both cases we replicate the payoff at maturity $T$.  
%We then construct hedging portfolios that replicate this price. 
   
\medskip

We construct a replicating portfolio so that $a_t$ is the number of
underlyings at time $t$ and $b_t$ is the amount in the bank account.  
The value of the portfolio is then
\begin{equation}
\label{eq:generalformportfolio}
   V_t = a_t X_t + b_t .
\end{equation}
The portfolio is required to replicate the price of the option 
so it replicates the payoff at maturity $V_T = h(X_T) $.
The net payment stream provided by the market over the time interval $(0,T)$ 
due to changes in the price of the underlying  is
$$
   \int_0^T  a_s dX_s .
$$
The change in the portfolio value that is not  ``financed'' by the market
has to be paid by the portfolio holder and we call this  the cost function:
$$
  E_T  = h(X_T) - \int_0^T a_s dX_s .
$$
This hedging scheme is called a DA scheme if  $a_t $ is a function of $t$ and $X_t$.
The general class of DA hedging schemes contains the delta hedging strategies, that is to say, 
the strategies in which  the number $a_t = \delta(t,X_t)$ of
underlyings in the portfolio at time $t$ is the derivative of the price of the option
with respect to the value of the underlying.
We consider  first two main delta hedging strategies characterized by 
the chosen ``delta'':

(HW):   The delta of the corrected price: 
\ba\label{eq:dMV} 
 \delta^{\rm HW}(t,x)=\partial_x
 P(t,x ) ,
\ea
with $P$ given by (\ref{eq:Pcorrected}).

(BS):    The delta of the Black-Scholes price at the implied volatility: 
\ba
\label{eq:dBS}
  \delta^{\rm BS}(t,x)=\partial_x 
  Q^{(0)}(t,x;\sigma)|_{\sigma=\sigma(t,x)},
  \ea
 with the implied volatility $\sigma(t,x)$ solving 
\begin{equation}
\label{def:impliedvol}
P(t,x ) =  Q^{(0)}(t,x;\sigma(t,x))  . 
\end{equation}

\medskip

 {
Note  that we  here  define the  implied volatility relative to the corrected price 
$P$.  }

In the case that the volatility is constant and equal to $\bar\sigma$, corresponding
to the standard  Black-Scholes model,  these    approaches coincide
and the portfolios are self-financing. 
In the case that the volatility is fluctuating,  the model is incomplete   
and we accumulate additional hedging cost 
during the lifetime  of the option.  
We remark that with no leverage effect (which means that
the volatility factor is independent of the Brownian motion driving
the underlying price), then $D=0$ and the two approaches 
 coincide and give the same hedging cost.

By (\ref{def:impliedvol}), the delta of hedging scheme (HW) corresponds to 
\ban
   \delta^{\rm HW} (t,x)= \partial_x Q^{(0)}(t,x;\sigma(x,t)) 
   + \partial_\sigma Q^{(0)}(t,x;\sigma(x,t))  
   \times \partial_x \sigma(x,t)   .
\ean
 This scheme is referred to as the {\it minimum variance delta} 
  in the recent paper \cite{HWnew}  by Hull and White.
   They find by empirical comparison of a few strategies that 
   this hedging approach is the one associated with minimum 
   hedging risk or cost variance.
 In \cite{HWnew} the minimum variance delta and enhanced performance
 is motivated by the presence  of leverage.
 Here we  quantify the means and variances of the hedging costs analytically
 and correspondingly identify analytically the hedging approach with minimum  hedging cost variance 
in our setting, which is not the (HW) scheme.

 The costs of the hedging strategies are characterized by the three market parameters
\[
 \bar\sigma, \quad D, \quad \Gamma ,
\]
see 
  Section \ref{sec:est}.
The first and second are sufficient to characterize the price as we have remarked above,
the third is a {\it hedging risk} parameter. 
 Consider the situation when we construct a hedging
 portfolio of value $P(t,X_t)$ and write the total hedging  cost 
 at maturity  $T$ by
\ba\label{eq:cost1}
 E_T^{\rm C}=   P(0,X_0) + Y_T^{\rm C} , \quad {\rm C=HW,BS},
\ea  
for the  two choices of hedging delta. 
Here  $X_0$ is the underlying value at initiation time $t=0$
and $P(0,X_0)$ the initiation cost of the portfolio.
Then in a sense made precise  below   the 
random  part of the cost  at maturity $Y_T^{\rm HW}$ is 

\ban
       Y_T^{\rm HW}    =  \Gamma 
       \int_0^T \big( x^2 \partial_x^2 \big)^2 Q^{(0)}(s,X_s;
       \bar{\sigma}) dB_s =
      {\Gamma} 
       \int_0^T \frac{  \partial_\sigma  Q^{(0)}(s,X_s;\bar{\sigma})  }{T-s} dB_s
        ,
\ean 
for $B$ a standard Brownian motion. 
If the price sensitivity to volatility changes, the  Vega, 
is small, then the  Vega risk  is  small as well. 
The sensitivity to  Vega in the cost accumulation becomes  larger 
as one approaches maturity.
The cost does not depend on the market pricing parameter $D$,  and hence
it does not depend on the leverage correlation parameter $\rho$ either
($\rho$ is the correlation between the 
volatility factor and the Brownian motion driving
the underlying price, see Eq.~(\ref{eq:corr}) below). 
However, it is proportional to the hedging risk parameter $\Gamma$ which does not depend on $\rho$ and which is
the central new parameter.   
 Thus, the hedging  approach is {\it leverage compensating}  in that it  ``immunizes'' 
the portfolio with respect to ``leverage risk''.
In the particular case of a European call option with strike $K$, i.e. $h(x)=(x-K)^+$, we have 
$\EE[ Y_T^{\rm HW} \mid {\cal F}_0 ]=0$ and 
\begin{equation}
\label{eq:varMVT}
     {\rm Var}\left(  Y_{T}^{\rm HW}  \mid {\cal F}_0 \right)  =   \Big( \frac{K  \Gamma}{\bar\sigma}  \Big)^2
       \frac{1}{2\pi}
   \int_0^{1} \exp\Big(-\frac{d_{_-}^2}{1+s} \Big) \frac{ 1}{\sqrt{1-s^2}}  ds  ,         
\end{equation}
with the standard Black-Scholes parameter 
\begin{equation}
\label{eq:bsparameter0}
  d_{\pm} = \frac{\log(X_0/K)}{\sqrt{\tau}}  \pm \frac{\sqrt{\tau}}{2}  , 
  \quad \quad \tau = \bar\sigma^2T .
\end{equation}
Here the expectation and variance are taken conditionally on the information at time zero.
We show this hedging cost variance at maturity in Figure \ref{fig_2} 
 as a function of relative time to maturity,  $\tau=\bar{\sigma}^2 T$, and moneyness,  $m=X_0/K$. 
\begin{figure}
\begin{center}
\begin{tabular}{c}
\includegraphics[width=8.4cm]{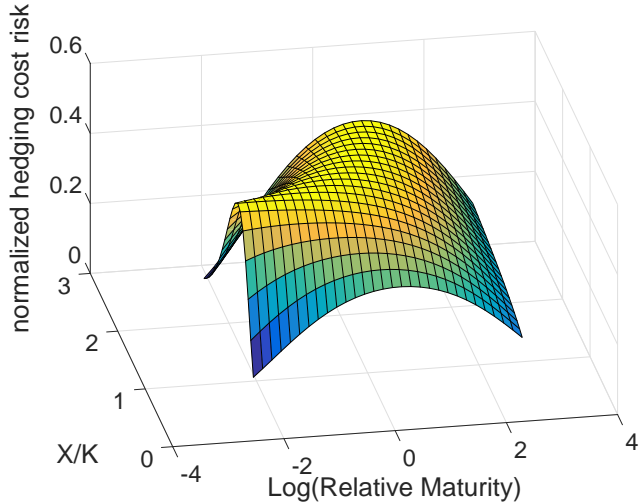}
\end{tabular}
\end{center}
\caption{     The figure shows the normalized  hedging cost standard deviation
$ St.Dev( Y_{T}^{\rm HW} ) \bar\sigma/(K\Gamma)$ as function of relative time to maturity $\tau=\bar{\sigma}^2 T$ and
moneyness $m=X_0/K$.
\label{fig_2} }
\end{figure}

We next state the important result that leverage 
makes the ``practitioners'' hedging approach superior.
We have   explicitly $\EE [  Y_{T}^{\rm BS}  \mid {\cal F}_0]=0$ and
\ba\label{eq:main}
  {\rm Var}\left(  Y_{T}^{\rm BS}   \mid {\cal F}_0\right)    =  
  {\rm Var}\left(  Y_{T}^{\rm HW}  \mid {\cal F}_0 \right) \left(  1 -  \Big(\frac{D}{\bar\sigma \Gamma}\Big)^2  \right)   ,
  \mbox{ with } \left|\frac{D}{\bar\sigma \Gamma}\right| \leq |\rho| \leq 1  ,
\ea
which implies $ {\rm Var}\left(  Y_{T}^{\rm BS}  \mid {\cal F}_0 \right)  \leq  
  {\rm Var}\left(  Y_{T}^{\rm HW} \mid {\cal F}_0  \right)$.
  The main result of this paper is then set forth in Section \ref{sec:opt}, Proposition \ref{prop:optim}:
 {\it the (BS) hedging scheme minimizes the hedging cost variance among all DA
  hedging  schemes}, thus is the true minimum variance hedging scheme
  in the regime discussed here!  
  {This result is proved in the regime of fast mean reversion and confirmed by the numerical
  simulations reported in Section \ref{sec:sim}.
  These simulations also reveal that the result is robust with respect to the scaling regime:  the hedging cost variance  
  of the (BS) strategy is always smaller than (or equal to) the one of the  (HW) strategy.}

  In Section \ref{sec:expou} we discuss the explicit 
  expressions of the effective market parameters when the volatility model is the exponential of a standard or fractional 
  (with Hurst exponent $H<1/2$) Ornstein-Uhlenbeck  process.  In this case we have 
\begin{equation}
     \frac{D}{\bar\sigma \Gamma}  \approx \rho .
\end{equation}
 
 Note that the implementation of the delta hedging schemes
  (HW) and (BS) requires the knowledge of the two effective market parameters
  $\bar\sigma$ and $D$.    
  Below we will also discuss the case when
  we choose  a ``homogenized'' or ``historical''  delta:
  
  (H):    The delta of the Black-Scholes price at the historical volatility: 
  \begin{equation}\label{eq:dH}
\delta^{\rm H} (t,x)=  \partial_x   Q^{(0)}(t,x;\bar{\sigma}) . 
\end{equation}
%and with the portfolio value being $Q^{(0)}(t,x)$.  
This  hedging scheme (H) can be implemented with 
{\it only the knowledge of $\bar\sigma$} and does not 
require calibration based on  pricing data.
However, in all cases implementing
the hedging scheme and  simultaneously characterizing
the hedging cost mean and variance requires the knowledge of all three market 
parameters $(\bar\sigma, D, \Gamma)$.   For the scheme (H)
we can write
as in Eq. (\ref{eq:cost1}) for the hedging cost at maturity:
\ba\label{eq:cost2}
  E_T^{\rm H}=  P(0,X_0) + Y_T^{\rm H}  ,
\ea  
and it follows from Proposition \ref{prop:optim} that  
${\rm Var}\left(  Y_{T}^{\rm H} \mid{\cal F}_0 \right) \geq {\rm Var}\left(  Y_{T}^{\rm BS}  \mid{\cal F}_0 \right)$.  
In particular for a European call with strike $K$ we have $\EE \left[   Y_{T}^{\rm H}  \mid {\cal F}_0\right] = 0$ and 
\ba
  {\rm Var}\left(  Y_{T}^{\rm H}  \mid {\cal F}_0 \right)    =  
  {\rm Var}\left(  Y_{T}^{\rm HW}  \mid {\cal F}_0 \right)  +   
  \Big(\frac{K D}{\bar\sigma^2}\Big)^2  \hat{w}^{\rm H}( d_{_{-}})  ,
\ea
with $d_{_{-}}$ given by (\ref{eq:bsparameter0}) and
\begin{eqnarray}
\nonumber
 \hat{w}^{\rm H}(d )   &=& 
    \frac{ 2 }{ \pi } 
 \int_0^1    \exp\left(- \frac{d^2}{1+s}\right)  
   \frac{1}{\sqrt{1-s^2}}   \Big[
d^4 \frac{(1-s)^2}{(1+s)^4}  +
 d^2 \frac{6s(1-s)}{(1+s)^3}   +
  \frac{3s^2 }{(1+s)^2} 
  - \frac{1}{2}
      \Big] ds\\
      &&
           -    \frac{   d^2  \exp(-d^2 ) }{ {2\pi} }    .
\end{eqnarray}
%and we will characterize explicitly the function $w^{\rm H}$  in Section~\ref{sec:call}.
  {
We present numerical simulations in the case of European call options in  Section \ref{sec:sim}.
We find that the (BS) hedging scheme performs well even beyond the regime
of rapid mean reversion which is the asymptotic regime from which it derives. 
We summarize the  form of the deltas introduced in the case 
of European call options:
   \ban
 \delta^{\rm H}(t,x)  & = &     {\cal N}(d_+) ,  \\
    \delta^{\rm BS}(t,x)    & = &  \delta^{\rm H}(t,x)   +  {\cal D}  \frac{  d_{_-}^2  \exp(-d_{_-}^2/2) }{x \sqrt{\tau} }  ,   \\
   \delta^{\rm HW}(t,x)  & = &    \delta^{\rm H}(t,x)   +  {\cal D}  \frac{ (d_{_-}^2-1)   \exp(-d_{_-}^2/2) }{x \sqrt{\tau} }   ,
 \ean 
 with  ${\cal N}$ the cumulative normal distribution and
 $d_\pm, \tau$ given by (\ref{eq:bsparameter0}).
  Here   ${\cal D}$ is   a canonical hedging parameter. This parameter  can in fact be calibrated from the implied volatility skew, while 
the calibration approach we promote here is to calibrate this  from historical price paths so  as to minimize the hedging
cost with respect to this  parameter.       }

Below we will also present the results for the hedging costs 
in the case with  early exercise $t<T$.     
 Before we  present such hedging risk characterizations
in  cases  with  general payoffs and exercise times  in Section \ref{sec:hedging}
we discuss the modeling of the stochastic volatility in Section 
\ref{sec:svmodel} and the asymptotic pricing  formula in Section~\ref{sec:price}.   
     
%Note that  in \cite{HWnew} the authors set forth the relation:
%\ba
% \label{eq:HW}
% \partial_x \sigma(t,x)   \approx     
%   \frac{  H\left( \partial_x Q^{(0)}(t,x;\sigma)\mid_{\sigma=\sigma(t,x)} \right)  }{X_t \sqrt{T-t} } , 
% \ea
% with $H$ being a quadratic function and thus containing
% three parameters. This ansatz is motivated by  a scale invariance  argument and is
% tested empirically by  considering the case with the underlying being the S\&P500.
%The fit of such a  characterization  will indeed depend on the model posed
%in the analytic case, moreover,
%likely  the particular market  considered  in the empirical case.
%In Section \ref{sec:call}  we also compare the model  in Eq. 
%\ref{eq:HW}  with the one that derives from the analysis of the class of stochastic
%volatility models set forth  this paper.  
    
\section{A Class of  Fast Mean Reverting Rough Volatility Models}  
\label{sec:svmodel}

Consider  the price of the risky asset which follows,  under the historical measure,  the stochastic differential equation:   {
\begin{equation}
\label{eq:Xdef}
dX_t = X_t    \left(  d\mu_t    + \sigma_t^\eps   dW^*_t \right) ,
\end{equation}
 with $W^*$ a standard Brownian motion.    
In this paper we assume that  the short term interest rate $r=0$
and that the drift is negligible, so we set $\mu=0$. 
The stochastic volatility is a stationary process of the form
\begin{equation}
\label{def:stochmodel}
\sigma_t^\eps = F(Z_t^\eps) .
\end{equation}
 The stochastic volatility is not a Gaussian process but it is
a function of the volatility factor 
$Z_t^\eps$ that is a  scaled stationary Gaussian process: \begin{equation}
\label{eq:Zgen}
Z^\eps_t = \sigma_{\rm z} \int_{-\infty}^t {\cal K}^\eps(t-s) dW_s,
\end{equation}
where  $W_t$ is a standard Brownian motion under the historical measure and
\ba
\label{def:Keps}
{\cal K}^\eps(t) &=&  \frac{1}{\sqrt{\eps}} {\cal K}\Big(\frac{t}{\eps}\Big). 
 \ea
We have introduced the mean reversion  time scale $\eps$
 which is the small time scale in our problem. 
 It means in particular that we consider contracts whose time to maturity is long compared
 to the natural time scale of the volatility factor. 
 Thus, we refer to the volatility factor
 and associated volatility process  as {\it rapidly mean-reverting}.   }

We make the following assumption regarding the volatility model:
 \begin{enumerate}
\item[(i)]
${\cal K} \in L^2(0,\infty)$ with $\int_0^\infty {\cal K}^2(u) du = 1$
and ${\cal K} \in L^1(0,\infty)$.
  \item[(ii)] There is a $d>1$ so that: 
\begin{equation}
\label{eq:defd}
|{\cal K} (t)| = O(t^{-d}) \quad \hbox{as}  \quad t \to \infty. 
\end{equation}
 \item[(iii)]   $F$ is smooth increasing  and bounded from below (away from zero) and from above. 
\end{enumerate}

Under  these conditions $Z_t^\eps$ has mean zero and variance $\sigma_{{\rm z}}^2$.
We assume that 
 $W^*_t$ is a Brownian motion that is correlated to the stochastic volatility through
\begin{equation}\label{eq:corr} 
W^*_t = \rho W_t + \sqrt{1-\rho^2} W'_t ,
\end{equation}
where the Brownian motion $W'_t$ is independent of $W_t$.
The function $F$ is assumed to be one-to-one, positive-valued,
smooth, bounded and with bounded derivatives.
Accordingly, the filtration ${\cal F}_t $ generated
by $(W_t',W_t)$ is also the one generated by $X_t$.
Indeed, it is equivalent to the one generated by $(W^*_t,W_t)$, or $(W^*_t,Z^\eps_t)$.
Since $F$ is one-to-one, it is equivalent to the one generated by $(W^*_t,\sigma_t^\eps)$.
Since $F$ is positive-valued,
it is equivalent to the one generated by $(W^*_t,(\sigma_t^\eps)^2)$, or $X_t$. 

The volatility may thus be a mixing process or a  
rough process with rapid decay of correlations at the origin.  In the latter  case
the volatility is neither a martingale nor a Markov process.
We discuss next some particular volatility models.

\subsection{Standard Ornstein-Uhlenbeck Model}
\label{subsec:standardOU}
Here we discuss the standard model where 
$Z_t^\eps$ is the scaled Ornstein Uhlenbeck (OU) process.
It has the form (\ref{eq:Zgen}-\ref{def:Keps})  with
${\cal K}(t) = \sqrt{2} \exp(-t) $.
 The OU process $Z_t^\eps$ is a centered  Gaussian process 
with covariance of the form
\begin{equation}
\label{eq:expresscovZ}
\EE [ Z^\eps_t Z^\eps_{t+s}  ] = \sigma^2_{{\rm z}} {\cal C}_Z \Big(\frac{s}{\eps}\Big) ,
\end{equation}
with $ {\cal C}_Z(s)  =  \exp(-|s|)$.
It solves a Langevin equation driven by standard Brownian motion.
It is a martingale and a Markov process, which allows for the use of stochastic calculus \cite{fouque00}.

\subsection{Rough Volatility Models}
\label{sec:rough}
We discuss here the model where
$Z_t^\eps$ is the scaled fractional Ornstein Uhlenbeck (fOU) process with Hurst exponent $H \in (0,1/2)$. 
This process is described in more detail in  Appendix \ref{app:fOU},
it has the form (\ref{eq:Zgen}-\ref{def:Keps})  with
\ba
\label{def:Keps0}
{\cal K}(t) &=& \frac{\sqrt{2 \sin(\pi H)}}{\Gamma(H+\frac{1}{2})} 
 \Big[ t^{H - \frac{1}{2}} - \int_0^t (t-s)^{H - \frac{1}{2}} e^{-s} ds \Big] .
\ea
 The fOU process $Z_t^\eps$ is a centered  Gaussian process 
with covariance of the form (\ref{eq:expresscovZ})
with ${\cal C}_Z(0)=1$,
 see Eq. (\ref{eq:sou2}).  
 Compared to the standard OU process addressed in the previous subsection,
we allow here for more general
 volatility factors to capture the situations discussed in a number of  recent empirical findings that the volatility
 process is rough corresponding to rapid decay of  ${\cal C}_Z$ at the origin  \cite{gatheral1}.
 We arrive at such a situation by assuming that the OU process is driven by
 a fractional Brownian motion with Hurst exponent $H \in (0,1/2)$ rather than a standard 
 Brownian motion \cite{cheridito03}.
 As described  in Appendix  \ref{app:fOU}  this gives a volatility factor that is {\it rough}.
 We have  specifically  
now  that the  covariance function ${\cal C}_Z$ is rough at zero in the sense:
\begin{equation}
\label{eq:OUa1}
 {\cal C}_Z(s) =  1 - 
 \frac{1}{  \Gamma(2H+1)}  s^{2H}
+ o\big( s^{2H}\big) 
, \quad \quad s \ll 1,
\end{equation}
while it is integrable and it decays as $s^{2H-2}$ at infinity:
\begin{equation}
\label{eq:OUa2}
 {\cal C}_Z(s) =   
 \frac{1}{  \Gamma(2H-1)}  s^{2H-2}
+ o\big( s^{2H-2}\big) 
, \quad \quad s \gg 1,
\end{equation}
see Figure \ref{fig_OUcorr}.
This behavior of the covariance function is inherited by the volatility process $\sigma_t^\eps$ 
itself, see Eqs.~(\ref{eq:corrY12}) and (\ref{eq:corrY12b}).
 \begin{figure}
\begin{center}
\begin{tabular}{c}
\includegraphics[width=8.4cm]{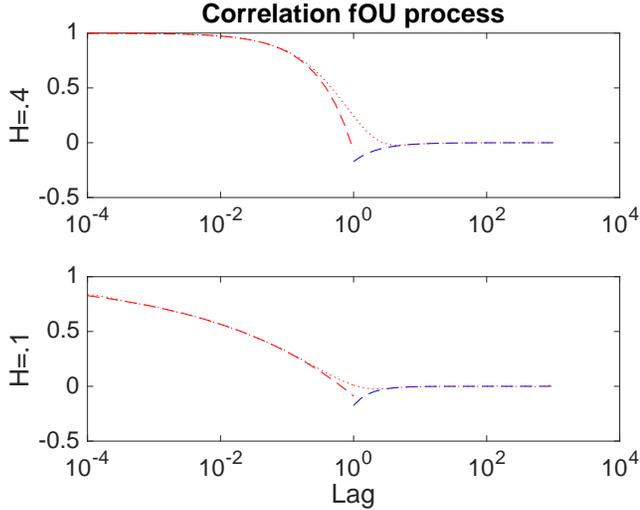}
\end{tabular}
\end{center}
\caption{     The dotted lines show the correlation functions  ${\cal C}_Z$  for
$H=.4$ (top plot) and $H=.1$ (bottom plot). Note that for large lags the correlation
function is slightly negative. 
Note also the rapid decay of the correlations at the origin, increasingly so
with smaller $H$. 
The dashed lines are the approximations in Eqs.
(\ref{eq:OUa1}) and (\ref{eq:OUa2}).   
    \label{fig_OUcorr}
}
\end{figure}
  For more details regarding this model we refer to \cite{sv3}. 
%  In Section \ref{sec:price} we give a result on the price  of European options
%  written on the underlying $X_t$. This result is  derived in \cite{sv3} and  
%   form the foundation for our  analysis of the hedging problem.  

\section{Prices of European Options}
\label{sec:price}
We are interested in  computing the option price defined as the martingale
\begin{equation}\label{eq:p1}
M_t =\EE^\star\big[ h(X_T) \mid {\cal F}_t \big]  ,
\end{equation}
where $h$ is a smooth function, $0 \leq t \leq T$.
%and the price follows the equation (\ref{eq:Xdef}) with the volatility model
%described in Section \ref{sec:svmodel}. 
 {In fact weaker assumptions are possible for $h$,
as we only need to control the function $Q^{(0)}_t(x)$ defined below rather than $h$,
as is discussed in  \cite[Section 4]{sv1} where the extension to more general $h$ such as $h(x)=(x-K)^+$ is addressed.
The expectation in Eq. (\ref{eq:p1}) is computed with respect to the pricing measure $\PP^\star$.
 {Recall that we assume that the short term interest  rate $r=0$, moreover that the drift $\mu=0$ under the historical 
measure.  We make here one more assumption in that the market price of volatility risk is assumed to be zero 
so that  $\PP=\PP^\star$ and indeed the models for $X_t$   coincide under the pricing measure $\PP^\star$ and the historical measure $\PP$.}
We remark first that the case with non-zero  interest rate and drift 
(under the historical measure)  could  have been analyzed in the framework presented below,
however,  for simplicity of expression we do not include this generality here.   
We remark second that the case with a non-zero market price of risk gives slightly more 
involved option price correction formulas than those presented below, see \cite{fouque00,fouque11}.  
This distinction is relevant in the case that we want to use information from the observed prices and  the
associated  implied volatility skew  for calibration of the hedging strategy, see the discussion
in Section \ref{sec:est}.  
In this paper we compute the  hedging cost statistics under the historical measure. 
}
   
We introduce the standard Black-Scholes operator at zero interest rate and
constant volatility $\sigma$:
\begin{equation}
{\cal L}_{\rm BS} (\sigma) = \partial_t +\frac{1}{2} \sigma^2 x^2 \partial_x^2 .
\end{equation}
We exploit the fact that the price process is a
martingale to obtain an approximation, via constructing an explicit
function  $P(t,x)$ so that $P(T,x)=h(x)$ and so that $P(t,X_t)$ 
is a martingale up to first order corrective terms in $\eps$. 
Then, indeed $P(t,X_t)$ gives the approximation for $M_t$ up to first  order in $\eps$.
The leading order price is the price at the homogenized or constant 
parameters. 
The following proposition gives the first-order correction to the 
expression for  the martingale $M_t$ in the regime of $\eps$  small.

\begin{proposition}
\label{prop:main}
We have
\begin{equation}
\lim_{\eps \to 0}  \eps^{-1/2} \sup_{t\in [0,T]}   \EE\left[ | M_t  - P(t,X_t) |^2 \right]^{1/2}  = 0  ,
\end{equation} 
 where
\begin{equation}
\label{def:Qt}
P(t,x)  = Q_t^{(0)}(x) +\sqrt{\eps}   \rho Q_t^{(1)}(x)   ,
\end{equation}
$Q_t^{(0)}(x) $ is deterministic and given by the Black-Scholes formula with constant 
volatility~$\bar{\sigma}$,
\begin{equation}
\label{eq:bs0}
{\cal L}_{\rm BS} (\bar\sigma) Q_t^{(0)}(x) =0,   \quad \quad Q_T^{(0)}(x) = h(x),
\end{equation}
with
\begin{equation} 
\label{def:barsigma}
\bar{\sigma}^2 = \left< F^2\right> = \int_\RR F(  \sigma_{\rm z}  z )^2 p(z) dz,
\end{equation}
$p(z)$ the pdf of the standard normal distribution,
$Q_t^{(1)}(x) $ is the deterministic correction solving
\begin{equation}
\label{eq:bs1}
{\cal L}_{\rm BS} (\bar\sigma) Q_t^{(1)}(x) =
-\overline{D}  \big( x \partial_x  (x^2 \partial_x^2 ) \big) Q_t^{(0)} (x),   \quad \quad Q_T^{(1)}(x) = 0.
\end{equation}
 The deterministic correction is
 \begin{equation}
\label{def:Q1t}
Q_t^{(1)}(x) =(T-t)\overline{D}  \big( x \partial_x  (x^2 \partial_x^2 ) \big) Q_t^{(0)} (x)   ,
\end{equation}
where the coefficient $\overline{D}$ is defined by 
\begin{equation}
\overline{D} =   \sigma_{\rm z}  \int_0^\infty
\Big[\iint_{\RR^2} F(\sigma_{\rm z}  z )  
(FF')(\sigma_{\rm z}  z') p_{{\cal C}_{\cal K}(s,0)} (z,z') d z dz'\Big]
 {\cal K}(s) ds,
\label{def:DtT}
\end{equation}
with $p_C(z,z')$ the pdf of the bivariate normal distribution with mean zero and covariance matrix 
%$\begin{pmatrix}
%1 & C\\ C & 1
%\end{pmatrix}
%$
\begin{equation}
\label{eq:Cmat}
\begin{pmatrix}
1 & C\\ C & 1
\end{pmatrix}   ,
\end{equation}
and 
\ba
\label{Cdef}
{\cal C}_{\cal K}(s,s')  =  \int_0^{\infty}   {\cal K}(s+v)  {\cal K}(s'+v)  dv .
\ea
 \end{proposition}
%
%From Eq.~(\ref{def:Qt}) we immediately get the following corollary, in which $Q^{(j)}(t,x;\sigma )$ 
%stands for $Q_t^{(j)}(x)$  with the constant volatility $\sigma$ instead of $\bar{\sigma}$.
%\begin{corollary}\label{corr:delta}
%The leading order implied volatility $\sigma(t,x)$ solves
%\ban
%   P(t,x)=  Q^{(0)}(t,x;\sigma(t,x)) ,
%%    = Q^{(0)}(t,x;\bar\sigma)+  \sqrt{\eps}   \rho Q^{(1)}(t,x;\bar\sigma)  . 
%     \ean
%with $P(t,x)$ given by (\ref{def:Qt}).    
%The delta of the corrected price has the leading order form
%\ban
%    \partial_x P(t,x) = \partial_x Q^{(0)}(t,x;\sigma(t,x))  + 
%     \partial_\sigma Q^{(0)}(t,x;\sigma(t,x)) \partial_x \sigma(t,x) . 
%\ean
%\end{corollary}
%
%

 {The mixing (Markov) case is proved in \cite{fouque00,fouque11} 
and the rough case is derived in \cite{sv3}. More precisely, the above statement concerns 
a  generalization of the  volatility model addressed in \cite{sv3} 
and can be derived via a straightforward modification of the proof presented there.}
Thus,  we see that the effect of the
volatility fluctuations gives a price modification that is of the order of $\eps^{1/2}$ 
and which is determined by the effective parameter $\overline{D}$ only.  
 The main result of this paper  is a precise statistical 
 characterization of  hedging  cost  in the context of fast 
 mean-reverting stochastic volatility. 
 Our novel analysis  uses  the analytic framework  set forth in  \cite{sv3}.
As for the case of  option prices the hedging cost results are  for the general  volatility model
 (\ref{def:stochmodel}).   Therefore they apply in particular  in a uniform
way to the cases of Markov and rough volatility. 
 
 {We remark  that the rough volatility case $H<1/2$ and the mixing case are qualitatively similar.
As a matter of fact, the parameter $\overline{D}$ for the standard OU process of Subsection \ref{subsec:standardOU}
is the limit as $H \nearrow 1/2$ of the  parameter $\overline{D}$ of the fOU process of Subsection \ref{sec:rough} 
(this can be shown  by using the dominated convergence theorem and the convergence of (\ref{def:Keps0}) to $\sqrt{2}\exp(-t)$).
However,  the ``long-memory''  case     
  addressed in \cite{sv2}, corresponding to $H>1/2$,  is different.
   In this case
  the volatility ``history'' plays a crucial role  and gives  a qualitatively
  different picture from the point of view of pricing and hedging.
  This is also the case for small volatility fluctuations as presented 
  in \cite{sv1} which in fact is quite similar  in its analysis to 
  a slow volatility factor.   In the case  of a slow volatility factor, slow relative to the
  maturity horizon,  the volatility will in fact appear as non-stationary on the time scale of
  the maturity.  These  other cases will be discussed  elsewhere.}

\section{Hedging Cost Accumulation}  
\label{sec:hedging}  

In the following sections we derive  the results for  the costs associated
with the hedging schemes introduced above in the context of European 
options.  We summarize 
in the next proposition these results. 
We introduced the hedging schemes (H),  (HW),  (BS) in Section 
\ref{sec:sum}.  In Section \ref{sec:mH} we introduce the modified  scheme 
($\H$) where the delta is chosen to be $\delta^{\rm H}$ 
as in the  (H) scheme, however,
the value of the portfolio is chosen to be $P(t,x)$ rather than $Q^{(0)}_t(x)$ as in the 
(H) scheme. 
The following proposition follows directly from Propositions 
\ref{prop:main2}, \ref{prop:main2b2}, \ref{prop:c}, 
\ref{prop:main2b2b} and Section \ref{sec:mH}. 
It gives the leading-order 
expressions of the expectations and the variances of the hedging costs 
(the leading order is $\sqrt{\eps}$ for the expectation and $\eps$ for the variance).

\begin{proposition}
\label{prop:sum}
If we write the hedging cost in the form 
 \ba
 \label{eq:costN}
  E_t^{\rm C}=   P(0,X_0) + Y_t^{\rm C} , \quad \mbox{ for } {\rm C=H, HW,  BS, \H},
\ea
then we have
\ba
\label{eq:meanYH}
&&
\lim_{\eps \to 0}\EE\left[\left(\eps^{-1/2}  \EE\big[   Y_t^{\rm H}  \mid {\cal F}_0 \big]   -  \frac{(t-T)}{T}    \frac{\rho \overline{D}}{\bar\sigma^2}   g (X_0,T)
\right)^2\right]^{1/2}
= 0 ,
   \\
&&\lim_{\eps \to 0}\EE\left[\left( \eps^{-1/2}  \EE\big[   Y_t^{\rm C}  \mid {\cal F}_0 \big] \right)^2\right]^{1/2}    =  0 , \quad   \mbox{ for }    {\rm C=HW,  BS, \H}   ,
\ea
and 
 \ba
&&  \lim_{\eps \to 0}\EE\left[\left|  \eps^{-1} {\rm Var}\big(   Y_t^{\rm HW}   \mid {\cal F}_0\big) -  \frac{\overline{\Gamma}^2}{\bar\sigma^2}    v(X_0,t,T)
   \right|\right] = 0
       ,\\
&& \lim_{\eps \to 0}\EE\left[\left| 
\eps^{-1} {\rm Var}\big(   Y_t^{\rm C}  \mid {\cal F}_0 \big)  -  \frac{\overline{\Gamma}^2}{\bar\sigma^2}    v(X_0,t,T)- 
  \frac{\rho^2 \overline{D}^2}{\bar\sigma^4}  w^{\rm C}(X_0,t,T)    \right|\right] = 0   ,
 \ea
 for  $ {\rm C=H,  BS, \H} $, 
 where $\overline{D}$ and $\overline{\Gamma}$ are the parameters given by (\ref{def:DtT}) and (\ref{def:barGamma})
 and $g,v,w^{\rm C}$ are cost mean and variance functions that depend on the
 payoff function $h$.
 \end{proposition}

The explicit forms  of $g,v,w^{\rm C}$ are given in Section \ref{sec:call}, Proposition \ref{prop:sum2},
in the case of European call options $h(x)=(x-K)^+$.  

{\it Remark.}
In the following we show that,  up to terms  of order $o(\sqrt{\eps})$:
\begin{align*}
 {E}^{\rm HW}_t - P(0,X_0) =& \,  N^{(1)}_t    ,  \\
 {E}^{\rm BS}_t - P(0,X_0) =&\, N^{(1)}_t  + \sqrt{\eps} \rho N^{(2)}_t    ,  \\
  {E}^{\rm \H}_t - P(0,X_0) =&\, N^{(1)}_t  + \sqrt{\eps} \rho \tilde{N}^{(2)}_t   ,
\end{align*}
where  ${N}^{(1)}$, resp. ${N}^{(2)}$,   $\tilde{N}^{(2)}$, are the martingales 
defined in Eq.~(\ref{eq:Ndef1}), resp. Eq.~(\ref{def:M}), Eq.~(\ref{eq:Ntdef}).
It is in fact the negative correlation between $N^{(1)}$ and $N^{(2)}$
that makes the (BS)  scheme  superior, see Section \ref{sec:opt}. 
In the case of the scheme (H) the hedging cost is characterized by
$$
 {E}_t^{\rm H} - Q_0^{(0)}(X_0)  =
             N^{(1)}_t +      
          \sqrt{\eps} \rho  \overline{D} \int_0^t 
           \big( x \partial_x  (x^2 \partial_x^2 ) \big) Q_s^{(0)} (X_s) ds ,
$$
with 
\ban
    \EE\left[ 
           \sqrt{\eps} \rho  \overline{D} \int_0^t 
           \big( x \partial_x  (x^2 \partial_x^2 ) \big) Q_s^{(0)} (X_s) ds \mid {\cal F}_0
      \right]
          =   
          \left( \frac{t}{T} \right)  
             \left(  P(0,X_0)   -   Q_0^{(0)} (X_0)  \right) .   
\ean
Here and below $\big( x \partial_x  (x^2 \partial_x^2 ) \big) Q_s^{(0)} (X_s)$ stands for $\big( x \partial_x  (x^2 \partial_x^2 ) \big) Q_s^{(0)} (x)$
evaluated at $x=X_s$. We next derive these results.

\subsection{Hedging Cost Process with (H) Hedging Strategy}
\label{sec:delta1}

Consider the (H) hedging scheme.
%Thus, we construct a replicating portfolio so that $a_t$ is the number of
%underlyings at time $t$ and $b_t$ the amount in the bank account.  
%The value of the portfolio is then
%\begin{equation}
%   V^{\rm H}_t = a_t X_t + b_t .
%\end{equation}
We assume that the {\it effective volatility $\bar \sigma$  is known}  and 
choose here the number of underlyings in the replicating portfolio  as the ``$\delta$'' of the Black-Scholes 
price evaluated at the {\it effective volatility} and the current price for the underlying.
Thus,  we consider here  the situation with ``homogenized''  or ``historical'' delta:  
\begin{equation}
\label{eq:D1}
a_t^{\rm H}= \delta^{\rm H} (t,X_t) ,\quad \quad  \delta^{\rm H}(t,x) = \partial_x   Q_t^{(0)}(x) ,
\end{equation}
as in Eq. (\ref{eq:dH}).
Moreover, in this section we choose the value  of the portfolio $V_t^{\rm H}$  to replicate
the Black-Scholes price $Q^{(0)}_t(X_t)$ evaluated at the  effective volatility:
\ba\label{eq:V0}
V^{\rm H}_t=Q_t^{(0)}(X_t) , \quad  0 \leq t \leq T ,
\ea
  and  
$b_t^{\rm H} = Q_t^{(0)}(X_t)  - a_t^{\rm H} X_t $.
As mentioned this hedging scheme can then be implemented knowing only
$\bar\sigma$. As we will show though in order to characterize the hedging
cost mean and variance we need to know also the effective market parameters $(D,\Gamma)$.
%Since we  require the portfolio to  replicate  the option price 
% at the effective volatility
%we have for  
%the value function 
%\ba\label{eq:V0}
%V^{\rm H}_t=Q_t^{(0)}(X_t) , \quad  0 \leq t \leq T ,
%\ea
%  and  
%$b_t^{\rm H} = Q_t^{(0)}(X_t)  - a_t^{\rm H} X_t $.
The portfolio replicates the payoff at maturity
$V_T^{\rm H} = Q^{(0)}_T(X_T) = h(X_T)$.
The cost function is:
\begin{equation}\label{eq:bc}
  E^{\rm H}_t  = V_t^{\rm H} - \int_0^t  a_s^{\rm H} dX_s ,
\end{equation}
with in particular
$  E^{\rm H}_0  =  Q_0^{(0)}(X_0) $.
We aim  to understand how this cost can be characterized.  

Using the fact that $Q^{(0)}$ solves the Black-Scholes equation we find
\ba\nonumber
     dE^{\rm H}_t &=& dV_t^{\rm H} -a_t^{\rm H} dX_t 
     = \Big( \partial_t +\frac{1}{2} (\sigma_t^\eps)^2 \big(x^2 \partial_x^2\big) \Big)
        Q^{(0)}_t(X_t)  dt + \partial_x Q^{(0)}_t(X_t)  dX_t - a_t^{\rm H} dX_t 
        \\  & = & \frac{1}{2} \left( (\sigma_t^\eps)^2  - \bar\sigma^2 \right) 
        \big(x^2 \partial_x^2\big) Q^{(0)}_t(X_t)  dt  . \label{eq:Edef}
\ea
We remark that we can write
\ban
dE^{\rm H}_t     =
        \frac{1}{2} \left( (\sigma_t^\eps)^2  - \bar\sigma^2 \right)  
        \frac{\nu_t(X_t) }{\bar\sigma (T-t)}  dt ,
\ean
where we introduced the  `` Vega'':
\begin{equation}
\nu_t(x) = \partial_{\bar\sigma} Q^{(0)}_t(x)  =  \bar\sigma (T-t) \big(x^2 \partial^2_x \big) Q^{(0)}_t(x) .
\end{equation}
Note that in the special case of constant volatility 
we have  $\sigma_t^\eps \equiv \bar\sigma$ and thus $dE^{\rm H}_t = 0$, which means that the cost 
is deterministic and given by the Black-Scholes price:
$$
\EE \big[ E^{\rm H}_t \mid {\cal F}_0\big] = Q_0^{(0)}(X_0) ,\quad 
{\rm Var}\big( E^{\rm H}_t \mid {\cal F}_0\big) = 0, \quad 0 \leq t \leq T.
$$
In the rapid  stochastic volatility case (\ref{def:stochmodel}), we can identify the leading-order terms of the cost.
Two equivalent expressions can be determined as shown in Lemma \ref{thm1}.
They will be useful to compute the mean and variance of the cost in the next propositions.
 
\begin{lemma}
\label{thm1}%
The hedging cost satisfies
\begin{equation}
   \label{def:hE1a}
\lim_{\eps \to 0}  \eps^{-1/2} \sup_{t\in [0,T]}   \EE\left[ | E^{\rm H}_t  - \hat{E}^{\rm H}_t   |^2   \right]^{1/2}  = 0  ,
\end{equation} 
 where
\begin{equation} 
\hat{E}_t^{\rm H} 
   \label{def:hE1}
=   Q_0^{(0)}(X_0)  
+ 
\eps^{1/2} \rho  \big(    Q_0^{(1)}(X_0)  
   -   Q_t^{(1)}(X_t)         \big)
   + N^{(1)}_t
+   
 \eps^{1/2}    \rho N^{(2)}_t       ,
\end{equation}
$N^{(1)}_t$ and $N^{(2)}_t $ are the martingales starting at zero
\begin{align}
\label{eq:Ndef1}
N^{(1)}_t &= \int_0^t \big( x^2 \partial_x^2 \big) Q_s^{(0)}(X_s) d\psi_s^\eps  ,\\
\label{def:M}
N^{(2)}_t   &= \int_0^t   \big( x \partial_x\big) Q_s^{(1)}  (X_s) \sigma^\eps_s dW_s^*   ,
\end{align}
with 
\begin{equation}
\label{def:Kt}
\psi_t^\eps = 
\EE \Big[  \frac{1}{2} \int_0^T \big( (\sigma_s^\eps)^2 -\overline{\sigma}^2 \big) ds \mid {\cal F}_t\Big] .
\end{equation}
We also have 
\begin{equation}
   \label{def:hE2a}
\lim_{\eps \to 0}  \eps^{-1/2} \sup_{t\in [0,T]}   \EE\left[ | E^{\rm H}_t  - \check{E}_t^{\rm H}  |^2   \right]^{1/2}  = 0  ,
\end{equation} 
 where
\ba
\check{E}_t^{\rm H}&=&
    Q_0^{(0)}(X_0)     +  
          \eps^{1/2} \rho  \overline{D} \int_0^t  \big( x \partial_x  (x^2 \partial_x^2 ) \big) Q_s^{(0)} (X_s) ds
            + N^{(1)}_t \label{def:hE2} % \\
%        &=&      Q_0^{(0)}(X_0)     +  
%          \eps^{1/2} \rho  \overline{D} \int_0^t (t-s) \big( (x \partial_x)^2  (x^2 \partial_x^2 ) \big) 
%          Q_s^{(0)} (X_s)   \sigma^\eps_s  dW^*_s
%            + N^{(1)}_t  \nn
     . 
  \ea
 \end{lemma}
Note that   the difference in Eq (\ref{def:hE1}) can be interpreted as the cost
of trading  the correction over the interval $(0,t)$ and $N^{(2)}$ is
(minus)  the martingale part of   this cost which gives Eq.~(\ref{def:hE2})
in view of the problems solved by $Q^{(0)}$ and  $Q^{(1)}$ as stated in Proposition
\ref{prop:main}.
Moreover, we can write from (\ref{def:Q1t}):
\ban
   \lim_{\Delta t \downarrow 0} 
    \frac{ \EE\left[  \check{E}_{t+\Delta t}^{\rm H} - \check{E}_t^{\rm H} \mid {\cal F}_t  \right]   }
       { \Delta t }  
    =   \frac{  \eps^{1/2} \rho Q_t^{(1)} (X_t)}{T-t} ,
\ean
so that the current  ``coherent cost flux'' corresponds to the accumulation of the cost of the correction
over the interval remaining until maturity.

\begin{proof}
\label{proof1}%  
Let $\phi_t^\eps$ be defined as the expected accumulated square volatility deviation in between
the present and maturity:  
\begin{equation}
\label{def:phit}
\phi_t^\eps= \EE\Big[ \frac{1}{2} \int_t^T \big( (\sigma_s^\eps)^2 -\overline{\sigma}^2 \big) ds \mid {\cal F}_t\Big] .
\end{equation}
Then we have
$$
\phi_t^\eps = \psi_t^\eps - \frac{1}{2} \int_0^t \big( (\sigma_s^\eps)^2 -\overline{\sigma}^2 \big) ds ,
$$
where the martingale $\psi_t^\eps$ is defined by (\ref{def:Kt}).
%We next denote
%\begin{equation}
%\label{def:G}
%G(z) = \frac{1}{2} \big( F(z)^2 - \overline{\sigma}^2\big) ,
%\end{equation}
%then
%the martingale $\psi^\eps_t$ defined by (\ref{def:Kt}) has the form
%\begin{equation}
%\psi_t^\eps = 
%\EE \Big[   \int_0^T G(Z_s^\eps)  ds \mid{\cal F}_t\Big] .
%\end{equation}
%Moreover, with ${\cal K}^\eps$ defined in Eq.~(\ref{def:Keps0}),
$(\psi_t^\eps)_{t\in [0,T]}$ is a square-integrable martingale that satisfies the following properties:
\begin{itemize}
\item
The quadratic covariation of $\psi^\eps$ and $W$ is
\begin{equation}
\label{def:varthetaeps}
d \left< \psi^\eps, W\right>_t =  \vartheta^\eps_{t} dt ,
\quad \quad \vartheta^\eps_{t} = \sigma_{\rm z}
 \int_t^T \EE \big[ FF'(Z_s^\eps) \mid {\cal F}_t \big]{\cal K}^\eps(s-t) ds   ,
\end{equation}
with ${\cal K}^\eps$ of the form~(\ref{def:Keps}). %(\ref{def:Keps0}).
\item
   There  exists a constant $K_T$ such that we have almost surely
\begin{equation}
\label{eq:estimvarthetaeps}
\sup_{t\in [0,T]} \big|   \vartheta^\eps_{t}  \big| \leq K_T \eps^{1/2} .
\end{equation}
\end{itemize}
The first part  was proved in \cite[Lemma B.1]{sv2}.
The second part follows from
the fact that ${\cal K}^\eps(t)={\cal K}(t/\eps)/\sqrt{\eps}$, ${\cal K} \in L^1(0,\infty)$.

We define the martingales starting from zero at time zero:
\ba\label{eq:Ndef0}
dN_t^{(0)} &= & ( x \partial_x) Q_t^{(0)} (X_t) \sigma_t^\eps dW_t^* , \\ 
%dN^{(1)}_t &=&  \big( x^2 \partial_x^2 \big) Q_t^{(0)}(X_t) d\psi_t^\eps  ,\\
dN^{(3)}_t &=& 
\big( x\partial_x (  x^2 \partial_x^2 )\big) Q_t^{(0)}(X_t) \sigma_t^\eps \phi_t^\eps dW_t^*
   .
%dN^{(2)}_t &=&   \big( x \partial_x Q_t^{(1)}\big) (X_t) \sigma_t^\eps dW_t^*  .
\ea
Then Eqs. 
(31) and (36)
%(4.9) and (4.14) 
in \cite{sv3} read:
\ba\label{eq:1}
&& \frac{1}{2}\big( (\sigma_t^\eps)^2-\overline{\sigma}^2\big) \big( x^2 \partial_x^2 \big) Q_t^{(0)}(X_t) dt
 = 
dQ_t^{(0)}(X_t)
 -  dN^{(0)}_t ,  \\  && 
 \nonumber
 dQ_t^{(0)}(X_t) = 
- d \big[   \phi_t^\eps \big( x^2 \partial_x^2 \big) Q_t^{(0)}(X_t)
+\eps^{1/2} \rho  Q_t^{(1)}(X_t)  
\big]
\\
\nonumber
 &&
\hspace*{0.8in} +
   \frac{1}{2} 
\big( x^2\partial_x^2(  x^2 \partial_x^2)\big) Q_t^{(0)}(X_t) \big( (\sigma_t^\eps)^2 - \overline{\sigma}^2\big)\phi_t^\eps dt
\\  
&&
\nonumber
\hspace*{0.8in}
+\frac{\eps^{1/2}}{2} \rho  \big( x^2 \partial_x^2 \big) Q_t^{(1)}(X_t) \big( (\sigma_t^\eps)^2 -\overline{\sigma}^2 \big) dt 
\\  
&&
\nonumber
\hspace*{0.8in}
  +  \rho   \big( x\partial_x (  x^2 \partial_x^2  ) \big) Q_t^{(0)}(X_t)  
  \big( \sigma_t^\eps  {\vartheta}_{t}^{\eps} - \eps^{1/2} \overline{D}\big) dt \\
  &&\hspace*{0.8in}
  + d N^{(0)}_t + dN^{(1)}_t +\eps^{1/2} \rho dN^{(2)}_t+ dN^{(3)}_t  .
  \label{eq:2}
\ea
In \cite{sv3} it is shown that the third, fourth, and fifth terms of the right-hand side of (\ref{eq:2}) are smaller than $\eps^{1/2}$. 
That is, if 
we introduce for any $t \in [0,T]$:
\begin{eqnarray}
R^{(1)}_{t,T} &=& \int_t^T  \frac{1}{2} \big( x^2\partial_x^2(  x^2 \partial_x^2 )\big) Q_s^{(0)}(X_s)  \big( (\sigma_s^\eps)^2 - \overline{\sigma}^2\big) \phi_s^\eps
 ds , \\
 \label{def:R2}
R^{(2)}_{t,T} &=& \int_t^T \frac{\eps^{1/2}}{2} \rho  \big( x^2 \partial_x^2 \big) Q_s^{(1)}(X_s)  \big( (\sigma_s^\eps)^2 -\overline{\sigma}^2 \big)ds , \\
R^{(3)}_{t,T} &=& \int_t^T  \rho   \big( x\partial_x(  x^2 \partial_x^2 ) \big) Q_s^{(0)}(X_s) 
 \big(  \sigma_s^\eps\vartheta_{s}^\eps  -\eps^{1/2}  \overline{D} )ds  ,
\end{eqnarray}
 we have  for $j=1,2,3$,
\begin{equation}
\label{eq:estimeRj}
\displaystyle \lim_{\eps \to 0} \eps^{-1/2} \sup_{t \in [0,T]} \EE \big[ (R^{(j)}_{t,T})^2 \big]^{1/2} =0 .
\end{equation}
  From  Proposition \ref{prop:main} we have that 
 \ba
 \nonumber
    -d Q_t^{(1)}(X_t)   &=&
     \overline{D}  \big( x \partial_x  (x^2 \partial_x^2 ) \big) Q_t^{(0)} (X_t) dt
     -   dN_t^{(2)} \\
     && -   \frac{1}{2}   \big( x^2 \partial_x^2 \big) Q_t^{(1)}(X_t)  \big( (\sigma_t^\eps)^2 -\overline{\sigma}^2 \big)dt . 
     \label{eq:dQ1a}
 \ea  
 It then follows from (\ref{eq:Edef})-(\ref{eq:1})-(\ref{eq:2})-(\ref{eq:dQ1a}) that
 \ba
 \nonumber
    dE^{\rm H}_t &=&   \eps^{1/2} \rho  
     \overline{D}  \big( x \partial_x  (x^2 \partial_x^2 ) \big) Q_t^{(0)} (X_t) dt
         + d R_{t,T}^{(1)}+ d R_{t,T}^{(3)}  \\ 
    &&     \hbox{}  -d\big[ \phi_t^\eps \big( x^2 \partial_x^2 \big) Q_t^{(0)}(X_t) \big]   +   
       dN_t^{(1)} + dN_t^{(3)}   . 
       \label{eq:Kb}
   \ea 
 It follows from Lemma \ref{lem:A2}  that the first term in the second line
 of Eq. (\ref{eq:Kb})  is small: 
\begin{equation}
\lim_{\eps \to 0}  \eps^{-1/2} \sup_{t\in [0,T]}   
\EE\left[ \left|
      \int_0^t -d\big[ \phi_s^\eps \big( x^2 \partial_x^2 \big) Q_s^{(0)}(X_s) \big]    
 \right|^2 \right]^{1/2}  = 0  ,
\end{equation} 
and the third term, i.e. the martingale $N^{(3)}_t$, is small as well:
$$
\lim_{\eps \to 0}  \eps^{-1/2} \sup_{t\in [0,T]}   
\EE\left[ \big|N^{(3)}_t \big|^2 \right]^{1/2}=0
.
$$
We then get (\ref{def:hE2a}-\ref{def:hE2}).
 {Finally, by substracting (\ref{def:hE1}) from (\ref{def:hE2}), we obtain
$$
\check{E}^{\rm H}_t - \hat{E}^{\rm H}_t = \eps^{1/2} \rho \Big[
  \overline{D} \int_0^t  \big( x \partial_x  (x^2 \partial_x^2 ) \big) Q_s^{(0)} (X_s) ds - N^{(2)}_t \Big]
-
\eps^{1/2} \rho  \big(    Q_0^{(1)}(X_0)  
   -   Q_t^{(1)}(X_t)         \big) ,
$$
which gives with (\ref{def:R2})  and (\ref{eq:dQ1a})  that}
$$
\check{E}^{\rm H}_t = \hat{E}^{\rm H}_t + R^{(2)}_{0,T}-R^{(2)}_{t,T},
$$
so that (\ref{eq:estimeRj}) gives (\ref{def:hE1a}-\ref{def:hE1}).
%{\it Remarks about the proof.}\\
%1) We remark that 
%since $E_t$ is a finite variation process the martingale part of the second line in 
%Eq.~(\ref{eq:Kb}) must vanish identically, and indeed we have
%\ban
% && -d\big[ \phi_t^\eps \big( x^2 \partial_x^2 \big) Q_t^{(0)}(X_t) \big]   +      dN_t^{(1)}  +      dN_t^{(3)}  \\
% &&  = -  \sigma_t^\eps \big( x \partial_x (x^2 \partial_x^2) \big) Q_t^{(0)}(X_t)   d\left< \phi^\eps, W^* \right>_t
%      +   \frac{1}{2} \big( (\sigma_t^\eps)^2-\overline{\sigma}^2\big) 
%       \big( x^2 \partial_x^2 \big) Q_t^{(0)}(X_t)  dt  +dR_{t,T}^{(1)} \\
%&& = 
%       \Big(    -  \sigma_t^\eps  \vartheta^\eps_{t} \big( x \partial_x  ( x^2 \partial_x^2) \big) Q_t^{(0)}(X_t)
%             + \frac{1}{2} \big( (\sigma_t^\eps)^2-\overline{\sigma}^2\big)  
%              \big( x^2 \partial_x^2 \big) Q_t^{(0)}(X_t)  \\
%&&\quad              +  \frac{1}{2} \big( x^2\partial_x^2(  x^2 \partial_x^2 )\big) Q_t^{(0)}(X_t)  \big( (\sigma_t^\eps)^2 - \overline{\sigma}^2\big) \phi_t^\eps
%       \Big)  dt  . 
% \ean
%2)  We also remark that bounding the first term in the second line of Eq.~(\ref{eq:Kb})
% uses the estimate in Lemma \ref{lem:A2} and the fact that $H<1/2$. The case with $H \geq 1/2$
% is different and will be treated elsewhere.   
\end{proof}
  
 We next consider the expected hedging cost. 
 We find that, if we exercise at some time $0 \leq t\leq T$, 
 %the expected hedging cost at the initiation time  is a linear function in the price correction  (to leading order $\sqrt{\eps}$).
%  That is,  
  the  extra 
 hedging cost  beyond the Black-Scholes price at the effective volatility
  is the fraction $t/T$
  of the price correction at the initiation time: 
  
  \begin{proposition}
  \label{prop:main2}
The mean hedging cost satisfies
\begin{equation}
\lim_{\eps \to 0}  \EE\left[\left(  \eps^{-1/2}
  \EE\big[ E_t^{\rm H} - E^{\rm H}_0   \mid {\cal F}_0 \big]  -
    \frac{t}{T}    \rho Q_0^{(1)}(X_0)  \right)^2\right]^{1/2}=0  ,
    \label{eq:prop:main2}
\end{equation} 
with $E_0^{\rm H}=Q_0^{(0)}(X_0)$.  
 \end{proposition}  
 
 Therefore, we have
 \begin{equation}
 \lim_{\eps \to 0}  \EE\left[\left(  \eps^{-1/2} \big(
  \EE\big[ E_t^{\rm H}     \mid {\cal F}_0 \big]  -P(0,X_0) \big) -
    \frac{t-T}{T}    \rho Q_0^{(1)}(X_0)  \right)^2\right]^{1/2}=0 ,
 \end{equation}
 which gives (\ref{eq:meanYH}).

%  with $g(X_0,T)=  (\bar{\sigma}^2/\overline{D}) Q_0^{(1)}(X_0)=\bar{\sigma}^2 T \big( x \partial_x  (x^2 \partial_x^2 ) \big) Q_0^{(0)} (X_0)$.

 \begin{proof}
 From (\ref{def:hE2}) we have
 \begin{align}
 \eps^{-1/2} \EE\big[ \check{E}^{\rm H}_t - Q_0^{(0)}(X_0)    \mid {\cal F}_0 \big] 
  =       \rho  \overline{D} \int_0^t 
            \EE\big[ \big( x \partial_x  (x^2 \partial_x^2 ) \big) Q_s^{(0)} ({X}_s) 
           \mid {\cal F}_0  \big]
            ds  .
            \label{eq:prop:main2b}
 \end{align}
Using (\ref{def:hE2a}), Lemma  \ref{lem:EMn} (Eq.~(\ref{eq:lemEMn1})), and dominated convergence theorem, it follows that 
 \begin{align}
\lim_{\eps \to 0}  \EE\left[\left(
  \eps^{-1/2} 
  \EE\big[ E^{\rm H}_t - E^{\rm H}_0   \mid {\cal F}_0 \big] 
- \rho  \overline{D} \int_0^t 
            \EE\big[ \big( x \partial_x  (x^2 \partial_x^2 ) \big) Q_s^{(0)} (\tilde{X}_s) 
           \mid {\cal F}_0 \big] ds\right)^2 \right] =0 ,
  \label{eq:proofprop:main2a}           
 \end{align}
 with
 \begin{equation}
 \label{eq:tXdef0}
d\tilde{X}_t = \bar\sigma  \tilde{X}_t dW^*_t , \quad \quad  \tilde{X}_0 =  {X}_0  .
\end{equation}
On the one hand, from (\ref{eq:bs1})  we get
 \begin{align*}
\rho  \overline{D} \int_0^t 
            \EE\big[ \big( x \partial_x  (x^2 \partial_x^2 ) \big) Q_s^{(0)} (\tilde{X}_s) 
           \mid {\cal F}_0 \big] ds
    &=    -\rho  \EE\left[\int_0^t  {\cal L}_{\rm BS}(\bar{\sigma}) Q_s^{(1)}(\tilde{X}_s) ds   \mid {\cal F}_0 \right] 
\\&  
= -\rho \EE \big[   Q_t^{(1)} (\tilde{X}_t) - Q_0^{(1)} (\tilde{X}_0) 
           \mid {\cal F}_0  \big],
 \end{align*}
which is equal to $0$ at $t=0$ and equal to $\rho  Q_0^{(1)} (X_0)$ at $t=T$. \\
On the other hand, we have by It\^o's formula and  (\ref{eq:bs0}) that
\ban
    \EE\left[ \big( x \partial_x (x^2 \partial_x^2) \big) Q_s^{(0)} (\tilde{X}_s) \mid {\cal F}_0  \right]  =
      \big( x \partial_x ( x^2 \partial_x^2)\big)  Q_0^{(0)} ({X}_0) ,
\ean
which shows that the integral term in~(\ref{eq:proofprop:main2a}) is a linear function in $t$. Therefore it is equal to
 $(t/T) \rho Q_0^{(1)}(X_0)$, which completes the proof of (\ref{eq:prop:main2}).
\end{proof} 

We are also interested in the risk or uncertainty in the hedging cost
if we exercise at or before expiry.  We find that the magnitude of the
cost fluctuations is of order  $\sqrt{\eps}$.
We have an explicit integral
expression for the variance of the hedging cost fluctuations (to leading order $\eps$)
as explained in the following proposition:
  \begin{proposition}
  \label{prop:main2b2}
  The asymptotic variance of the cost fluctuations satisfies
\ba
\label{eq:cost0}
\lim_{\eps \to 0} \EE\left[\left| \eps^{-1}  {\rm Var} \big(
E^{\rm H}_t-E^{\rm H}_0   \mid {\cal F}_0\big)
-{\cal V}^{(1)}_t(X_0)-2{\cal V}^{(2)}_t(X_0)-{\cal V}^{(3)}_t(X_0) \right| \right]      =  0,
  \ea
with
\ba
%\label{eq:int2}
%&&
%V_j
%  = {\cal V}^{(j)}_t(X_0),  \quad j=1,2,3,    \\
%&& {\cal V}^{(1)}_t(x_0) =   \rho^2 \overline{D}^2 \bar\sigma^2 \int_\RR dz  p(z)  \int_0^t ds(t-s)^2
%  \left( (x \partial_x)^2    x^2 \partial_x^2   
%  Q_s^{(0)} (x_0 e^{\bar{\sigma} \sqrt{s} z -\bar{\sigma}^2 s/2} ) \right)^2  \nn   \\
%  && \hspace*{0.65in}
%- \Big(\frac{t}{T}\rho Q_0^{(1)}(x_0)\Big)^2 ,
% \\
&& {\cal V}^{(1)}_t(x_0) =  2  \rho^2 \overline{D}^2   \int_\RR dz  p(z)  \int_0^t ds(t-s) 
  \left(  \big( x \partial_x     (x^2 \partial_x^2   )\big)
  Q_s^{(0)} (x_0 e^{\bar{\sigma} \sqrt{s} z -\bar{\sigma}^2 s/2} ) \right)^2 
   \nn \label{eq:V1}  \\
  && \hspace*{0.65in}
- \Big(\frac{t}{T}\rho Q_0^{(1)}(x_0)\Big)^2 , \\
&& {\cal V}^{(2)}_t(x_0) =    \rho^2 \overline{D}^2 
 \int_\RR dz  p(z)  \int_0^t ds (t-s)
  \left( \big( (x \partial_x)^2  (x^2 \partial_x^2)  \big)
  Q_s^{(0)} (x_0 e^{\bar{\sigma} \sqrt{s} z -\bar{\sigma}^2 s/2}\right)   \nn \\
  &&  \hspace*{0.65in}   \times
   \left(  (  x^2 \partial_x^2  )
  Q_s^{(0)} (x_0 e^{\bar{\sigma} \sqrt{s} z -\bar{\sigma}^2 s/2}\right)  ,
   \label{eq:V2} \\
   && {\cal V}^{(3)}_t(x_0) =     
 {\overline{\Gamma}^2} \int_\RR dz  p(z)  \int_0^t ds 
  \left(  ( x^2 \partial_x^2  )
  Q_s^{(0)} (x_0 e^{\bar{\sigma} \sqrt{s} z -\bar{\sigma}^2 s/2}\right)^2   
   \label{eq:V3}  .
\ea
Here $p(z)$ is the pdf of the standard normal distribution, $\overline{\Gamma}$ is the parameter
 \begin{equation}
   \overline{\Gamma}^2  = 2   \sigma^2_{\rm z} 
\int_0^{\infty}  \int_{s}^{\infty}  
  \left[ 
   \iint_{\RR^2}  FF'( \sigma_{\rm z} z )  
        FF'(\sigma_{\rm z} z' ) 
p_{{\cal C}_{\cal K}(s,s')}(z,z') dzdz' \right]  {\cal K}(s) {\cal K}(s')  ds'ds , 
\label{def:barGamma}
\end{equation}
and $p_C$ is the pdf of the bivariate normal distribution with covariance matrix (\ref{eq:Cmat}) and
 ${\cal C}_{\cal K}(s,s')$ is defined by (\ref{Cdef}). 
 \end{proposition}
   \begin{proof}
From (\ref{def:hE1a}) and (\ref{def:hE2}), we can write
\ba 
&&  \eps^{-1}  {\rm Var} \big(
E^{\rm H}_t-E^{\rm H}_0   \mid {\cal F}_0\big)  =      V_1^\eps+2V_2^\eps + V_3^\eps +o(1) , \\
 && V_1^\eps=
 {\rm Var}\left( 
    \rho  \overline{D} \int_0^t  \big( x \partial_x  (x^2 \partial_x^2 ) \big) Q_s^{(0)} ({X}_s) ds    \mid {\cal F}_0    \right),\\
&&V_2^\eps =   \eps^{-1/2}  {\rm Cov}\left(
    \rho  \overline{D} \int_0^t  \big( x \partial_x  (x^2 \partial_x^2 ) \big) Q_s^{(0)} ({X}_s) ds,
     N^{(1)}_t  \mid {\cal F}_0 \right) ,
  \\
   &&
  V_3^\eps=  \eps^{-1} {\rm Var}\big( N^{(1)}_t \mid {\cal F}_0  \big)  .
   \label{def:V3}
\ea
Note that we have
$  x^2 \partial_x^2   =  \left(x \partial_x\right)^2 -x \partial_x$. It follows that 
\ban
 {\cal L}_{\rm BS}(\bar\sigma)   
 \big( \left(x \partial_x\right)^j  (x^2 \partial_x^2 ) \big) Q_t^{(0)} (x) = 0, \quad 
 j=0,1,\ldots . 
\ean
Then one can show that $ V_1^\eps$ converges in $L^1$ to 
$ {\cal V}^{(1)}_t(X_0)$ (given by Eq.~(\ref{eq:V1})) by Lemma  \ref{lem:EMn2}-Eq.~(\ref{eq:EM1}) and Proposition
\ref{prop:main2}. Similarly, using the expression (\ref{eq:Ndef1}) of $N^{(1)}$,  
one can show that $ V_2^\eps$ and $V_3^\eps$  converge in $L^1$
to $ {\cal V}^{(2)}_t(X_0)$ and $ {\cal V}^{(3)}_t(X_0)$ (given by  Eqs.~(\ref{eq:V2}) and (\ref{eq:V3})) by
Lemma \ref{lem:A2} and 
by Lemma \ref{lem:EMn2}-Eqs.~(\ref{eq:EM2}-\ref{eq:EM3})
respectively.  
  \end{proof}

We illustrate  the above result 
in the case  of a European call option in Section~\ref{sec:call}.
  
\subsection{Hedging Cost Process  using (HW) Hedging Strategy}
\label{sec:delta2}

In this section we analyze  the  hedging scheme
(HW)  described by Eq.~(\ref{eq:dMV}) where we use a ``corrected delta''  
 to  construct the portfolio. 
That is, we now use  the corrected Black-Scholes price in
Proposition \ref{prop:main}  and associated 
  delta  and value function.
 
Thus,  we construct a replicating portfolio so that $a^{\rm HW}_t$ is the number of
underlyings at time $t$ and $b^{\rm HW}_t$ is the amount in the bank account
according to the corrected strategy.
The value of the portfolio is now
\begin{equation}
   V^{\rm HW}_t = a^{\rm HW}_t X_t + b^{\rm HW}_t ,
\end{equation}
and we choose
\begin{equation}
\label{eq:D2}
a^{\rm HW}_t  = \delta^{\rm HW}(t,X_t),\quad\quad 
\delta^{\rm HW}(t,x) = \partial_x  P(t,x) = 
 \partial_x  \big( Q_t^{(0)} +\eps^{1/2}   \rho Q_t^{(1)}   \big) (x).
\end{equation}
We moreover require the portfolio to replicate  the corrected option price so that 
the value of the portfolio is
\ba\label{eq:V0c}
V^{\rm HW}_t=P(t,X_t), \quad  0 \leq t \leq T ,
\ea
  and  $
   b^{\rm HW}_t = P(t,X_t)   - a^{\rm HW}_t X_t$.
Again  the portfolio replicates the payoff at maturity
$    V_T^{\rm HW} = P(T,X_T) = h(X_T) $.
 The financing cost of the portfolio is
 \begin{equation}\label{eq:Ec}
  E^{\rm HW}_t  = V^{\rm HW}_t - \int_0^t a^{\rm HW}_s dX_s ,
\end{equation}
with in particular $
  E^{\rm HW}_0  =  P(0,X_0) $.
We aim  to understand how the cost  is affected by using the corrected 
strategy.  
The following lemma shows that, by using the corrected hedging strategy, we have 
in the incomplete market 
{\it restored the situation
with existence of a self-financing replicating portfolio
 to the order  of the approximation in the mean}.
 Moreover the hedging cost is characterized by the martingale 
 $N^{(1)}$ defined by (\ref{eq:Ndef1}).  
% This will have dramatic consequences on the mean and variace of the hedging cost as we will see
% in the forthcoming lemmas and propositions.

%We finally have
 \begin{lemma}
 \label{lem:Ec}
 The cost of the corrected hedging strategy satisfies
\begin{equation}
\label{eq:estimedEt}
\displaystyle \lim_{\eps \to 0} \eps^{-1/2} \sup_{t \in [0,T]} 
\EE \Big[ \big( E^{\rm HW}_t   - P(0,X_0)  - N_t^{(1)}  \big)^2 \Big]^{1/2}  = 0
    ,
\end{equation} 
where
$N^{(1)}_t$ is  the martingale defined in Lemma  \ref{thm1}, Eq.~(\ref{eq:Ndef1}).
%where $Q^\eps_0(x)  =  Q_0^{(0)} (x)  +\eps^{1/2}\rho Q_0^{(1)} (x)$ is defined by (\ref{def:Qt}).
\end{lemma}

\begin{proof}
In view of Eqs.  (\ref{eq:bs0})  and  (\ref{eq:bs1}) we find
\ban
     dE^{\rm HW}_t &=& dV^{\rm HW}_t -a^{\rm HW}_t dX_t  \\
     &=& \Big( \partial_t +\frac{1}{2} (\sigma_t^\eps)^2 x^2 \partial_x^2 \Big)
%        Q^\eps_t(X_t)
        P(t,X_t)
        dt + \partial_x 
        %Q^\eps_t(X_t)  
        P(t,X_t)
        dX_t - a^{\rm HW}_t dX_t 
        \\  & = & \frac{1}{2} \left( (\sigma_t^\eps)^2  - \bar\sigma^2 \right) 
        \big( x^2 \partial_x^2\big) 
        %Q^\eps_t(X_t)  
        P(t,X_t)
        dt 
         - \eps^{1/2} \rho \overline{D}  \big( x \partial_x  (x^2 \partial_x^2 ) \big) Q_t^{(0)} (X_t) dt .
\ean
We define $\tilde{E}^{\rm HW}$ by 
\begin{equation}
\label{eq:tEdef}
     d\tilde{E}^{\rm HW}_t  
= \frac{1}{2} \left( (\sigma_t^\eps)^2  - \bar\sigma^2 \right) 
        \big(x^2 \partial_x^2\big) Q^{(0)}_t(X_t)  dt   
                 - \eps^{1/2} \rho \overline{D}  \big( x \partial_x  (x^2 \partial_x^2 ) \big) Q_t^{(0)} (X_t) dt  ,
\end{equation}
starting from $  \tilde{E}^{\rm HW}_0 =  
%Q^\eps_0(X_0) 
P(0,X_0)$. 
Therefore
$$
{E}^{\rm HW}_t   - \tilde{E}^{\rm HW}_t  =
 \rho \eps^{1/2} \int_0^t  \frac{1}{2} \left( (\sigma_s^\eps)^2  - \bar\sigma^2 \right) 
        \big(x^2 \partial_x^2\big) Q^{(1)}_s(X_s)  ds  ,
$$
and we get from   Lemma \ref{lem:EMn0}:
\begin{equation}
\label{eq:estimedE}
\displaystyle \lim_{\eps \to 0} \eps^{-1/2} \sup_{t \in [0,T]} 
\EE \big[ ( {E}^{\rm HW}_t   - \tilde{E}^{\rm HW}_t    )^2 \big]^{1/2}  = 0 .
\end{equation} 
We have from (\ref{eq:Edef}) and  (\ref{eq:tEdef}):
$$
d \tilde{E}^{\rm HW}_t  -dE^{\rm H}_t =  
- \eps^{1/2} \rho \overline{D}  \big( x \partial_x  (x^2 \partial_x^2 ) \big) Q_t^{(0)} (X_t) dt .
$$
Using (\ref{def:hE2}) we get
\ban
 \tilde{E}^{\rm HW}_t  - E^{\rm H}_t + \check{E}^{\rm H}_t  
 &=&  \tilde{E}^{\rm HW}_0  - E_0^{\rm H} + \check{E}_0^{\rm H} +  N^{(1)}_t \\
 &=&   P(0,X_0) - Q^{(0)}_0(X_0) +Q^{(0)}_0(X_0) + N^{(1)}_t =
   P(0,X_0)+N^{(1)}_t.
\ean
Using  (\ref{def:hE2a}) we find that 
\ban
  && \lim_{\eps \to 0} \eps^{-1/2} \sup_{t \in [0,T]} 
\EE \left[ \left( {\tilde{E}}^{\rm HW}_t   -    P(0,X_0) - N^{(1)}_t    
\right)^2 \right]^{1/2}    
\\
&&
   =
\lim_{\eps \to 0} \eps^{-1/2} \sup_{t \in [0,T]} 
\EE \big[ (  E_t^{\rm H} -  \check{E}_t^{\rm H}   )^2 \big]^{1/2}   
 = 0 ,
\ean
which gives the desired result with Eq. (\ref{eq:estimedE}).
 \end{proof}
 
 This lemma allows us to characterize the mean and variance of the cost of the corrected hedging strategy.
 
\begin{proposition}
\label{prop:c}
The mean extra hedging cost beyond the corrected price is zero:
\begin{equation}
\lim_{\eps \to 0}  \EE\left[\left( \eps^{-1/2}
  \EE\big[ E_t^{\rm HW}- E_0^{\rm HW}   \mid {\cal F}_0 \big] \right)^2\right]^{1/2} = 
0  ,
\end{equation} 
with $E_0^{\rm HW}=P(0,X_0)$.  
The variance of the cost fluctuations satisfies
\ba 
\lim_{\eps \to 0}  \EE\left[\left| \eps^{-1}  {\rm Var} \big(
E_t^{\rm HW}-E_0^{\rm HW}   \mid {\cal F}_0
\big)   - {\cal V}^{(3)}_t(X_0) \right|\right]=0,
  \ea
  where ${\cal V}^{(3)}_t$ is given by  (\ref{eq:V3}).
 \end{proposition}  
\begin{proof}
The result on the mean follows from Lemma \ref{lem:Ec} and the fact that $N^{(1)}_t$ is a zero-mean martingale.
The result on the variance follows from Lemma \ref{lem:Ec} and the formula for the asymptotic variance of $N^{(1)}_t$  obtained in Proposition \ref{prop:main2b2}.
\end{proof}

%\begin{proposition}
%\label{prop:v3}
%The variance of the cost fluctuations is
%\ba 
%\lim_{\eps \to 0} \eps^{-1}  {\rm Var} \left[
%E_t^{\rm MV}-E_0^{\rm MV}   \mid {\cal F}_0
%  \right]      =  V_3 ,
%  \ea
%  where $V_3$ is given by  (\ref{eq:int2})-(\ref{eq:V3}).
%\end{proposition}
%\begin{proof}
%This follows from Lemma \ref{lem:Ec} and the formula for the asymptotic variance of $N^{(1)}_t$  obtained in Proposition \ref{prop:main2b2}.
%\end{proof}
  
\subsection{Hedging Cost with (BS) Hedging Strategy} 
\label{sec:BS}     
 {We consider here the hedging scheme (BS) described in Section 
\ref{sec:sum}, that is using the delta of the BS price at the implied volatility $\delta^{\rm BS}$ defined by (\ref{eq:dBS}) and (\ref{def:impliedvol}). }
 Here $Q^{(j)}(t,x;\sigma ), ~j=0,1$ 
stands for $Q_t^{(j)}(x)$  with the constant volatility $\sigma$ instead of $\bar{\sigma}$.
 {Since we here evaluate the BS hedging scheme which is based
on computing the implied  volatility we  assume that the  Black Scholes  Vega, $\partial_\sigma Q^{(0)}$, is positive  in the domain of interest.
The problem of identifying  the implied volatility in the case of a small correction is then well posed, see below.
We remark that  for the European put and call options that we discuss below the  Vega is
positive, as shown by the following lemma proved in Appendix \ref{app:pos}.}
\begin{lemma}
 { The Black-Scholes  Vega,  $\partial_\sigma Q^{(0)}(t,x)$,  is  well defined and positive for $x> 0, t>0$   if  the payoff function $h: [0,\infty) \rightarrow  \RR$  is convex and not affine and of at most polynomial growth.}
\end{lemma} 

Using a similar technique as in the derivation 
of Lemma \ref{thm1}  and  Proposition \ref{prop:main2b2} we then find the following result.

\begin{lemma}
\label{thm1b}%
The cost for  the hedging scheme (BS), $E^{\rm BS}$,  satisfies  
\begin{equation}
   \label{def:hE1ab}
\lim_{\eps \to 0}  \eps^{-1/2} \sup_{t\in [0,T]}  
 \EE\left[ \big( E^{\rm BS}_t  - \hat{E}_t^{\rm BS}  \big)^2 \right]^{1/2}  = 0  ,
\end{equation} 
 where
\begin{equation}
\hat{E}_t^{\rm BS}
   \label{def:hE1b}
=   P(0,X_0)  
    + N^{(1)}_t
+   
 \eps^{1/2}  \rho  \N_t       ,
\end{equation}
$N^{(1)}_t$ is  the martingale defined  by (\ref{eq:Ndef1}), and 
$\N_t  $ is the martingale defined by 
\ba\label{eq:Ntdef}
\N_t   &= & \overline{D}  \int_0^t  \tilde{\cal H}_s(X_s)   \sigma^\eps_s dW_s^*  ,\\
    \tilde{\cal H}_s(x)   &=&  \frac{1}{ \overline{D}} \left(  \big( x \partial_x\big)   
    - \Big(  \frac{x \partial_x \partial_\sigma Q^{(0)}(s,x;\bar\sigma)}
  {\partial_\sigma Q^{(0)}(s,x;\bar\sigma)} \Big)   
 \right)   Q^{(1)}(s,x;\bar\sigma)  .
 \label{def:tildeH}
\ea
   \end{lemma}
 \begin{proof}
 The implied volatility $\sigma(t,x)$ is such that
\ban
 Q^{(0)}(t,x;\sigma(t,x))  =  P(t,x) = Q^{(0)}(t,x;\bar\sigma) +
\sqrt{\eps}   \rho Q^{(1)}(t,x;\bar\sigma)   .
\ean

 {Note that with $\partial_\sigma Q^{(0)}$, the Black-Scholes  Vega, being continuous and not zero at 
$\bar{\sigma}$,
we have by the implicit function theorem:}
  \ban
  \sigma(t,x) - \bar\sigma     &=&  
   \frac{ \sqrt{\eps}   \rho Q^{(1)}(t,x;\bar\sigma)}{\partial_\sigma Q^{(0)}(t,x;\bar\sigma)}  +o(\sqrt{\eps})  .
 \ean
 The (BS) delta is: %,  up to terms of order $\sqrt{\eps}$:
\ban
 \delta^{\rm BS}(t,x) &=& 
  \left. \left( \partial_x  Q^{(0)}(t,x;\sigma) \right)\right|_{ \sigma=\sigma(t,x)} \\
   &=&
    \left. \left( \partial_x  \big(  Q^{(0)}(t,x;\bar\sigma)  +  
  \partial_\sigma Q^{(0)}(t,x;\bar\sigma)
   (\sigma-\bar\sigma)\big) \right)\right|_{\sigma=\sigma(t,x)}   +o(\sqrt{\eps})   , 
    \ean
so that we can write: %,  up to terms of order $\sqrt{\eps}$:
 \ban
 \delta^{\rm BS}(t,x)  = 
     \delta^{\rm H}(t,x)  +    \sqrt{\eps}   \rho Q^{(1)}(t,x;\bar\sigma)
  \left(  \frac{\partial_x \partial_\sigma Q^{(0)}(t,x;\bar\sigma)}
  {\partial_\sigma Q^{(0)}(t,x;\bar\sigma)} \right)  +o(\sqrt{\eps}) .
     \ean
  Then it follows from Eqs. (\ref{eq:bc}) and (\ref{def:hE1})  that the  cost is
\ban\label{eq:EBS}
  {E}_t^{\rm BS} 
&=&   P(t,X_t) - \int_0^t \delta^{\rm BS}(s,X_s) d X_s \\
&=& E^{\rm H}_t +\sqrt{\eps} \rho Q^{(1)}_t(X_t) - \sqrt{\eps}   \rho  \int_0^t   
  \left(  \frac{x \partial_x \partial_\sigma Q^{(0)}(s,X_s;\bar\sigma)}
  {\partial_\sigma Q^{(0)}(s,X_s;\bar\sigma)} \right)  
  Q^{(1)}(s,X_s;\bar\sigma)
    \sigma_s^\eps  dW^*_s    +o(\sqrt{\eps})  \\
&=&
P(0,X_0) 
+     N^{(1)}_t
+    \sqrt{\eps}    \rho N^{(2)}_t     \\ && \hbox{}
  -   \sqrt{\eps}   \rho  \int_0^t   
  \left(  \frac{x \partial_x \partial_\sigma Q^{(0)}(s,X_s;\bar\sigma)}
  {\partial_\sigma Q^{(0)}(s,X_s;\bar\sigma)} \right)  
  Q^{(1)}(s,X_s;\bar\sigma)
    \sigma_s^\eps  dW^*_s   +o(\sqrt{\eps})  \\ 
    \nn   &= &P(0,X_0)  +     N^{(1)}_t  +  \sqrt{\eps}    \rho  \N_t +o(\sqrt{\eps})  ,
\ean
with $\N_t  $ defined by (\ref{eq:Ntdef}).
 \end{proof}
 
  This lemma allows us to characterize the mean and variance of the cost of the (BS) hedging scheme.

 \begin{proposition}
  \label{prop:main2b2b}
  The mean and variance of the cost fluctuations satisfy
\ba
 \label{eq:cost0b}
&& \hspace*{-0.25in}
 \lim_{\eps \to 0} \EE\left[ \left( \eps^{-1/2}  \EE \big[
E^{\rm BS}_t-E^{\rm BS}_0   \mid {\cal F}_0
  \big]  \right)^2\right]^{1/2}=  0,   \\
&& \hspace*{-0.25in}
\lim_{\eps \to 0} 
\EE
\left[ \left| 
\eps^{-1}  {\rm Var} \big(
E^{\rm BS}_t-E^{\rm BS}_0   \mid {\cal F}_0
 \big)   - \tilde{\cal V}^{(1)}_t(X_0)-2 \tilde{\cal V}^{(2)}_t(X_0)- \tilde{\cal V}^{(3)}_t(X_0) 
\right|\right]=0, 
  \ea
with $E_0^{\rm BS}=P(0,X_0)$, 
\ba
\label{eq:int2b}
&& \tilde{\cal V}^{(1)}_t(x_0) =   \rho^2 \overline{D}^2 \bar{\sigma}^2 \int_\RR dz  p(z)  \int_0^t ds 
     \left(  
  \tilde{\cal H}_s(x_0 e^{\bar{\sigma} \sqrt{s} z -\bar{\sigma}^2 s/2} ) \right)^2   , 
   \label{eq:V1b} \\
&& \tilde{\cal V}^{(2)}_t(x_0) =    \rho^2 \overline{D}^2 
 \int_\RR dz  p(z)  \int_0^t ds  \, 
        \tilde{\cal H}_s(x_0 e^{\bar{\sigma} \sqrt{s} z -\bar{\sigma}^2 s/2} )
      \nn \\
  &&  \hspace*{0.65in}   \times
   \left(    x^2 \partial_x^2  
  Q_s^{(0)} (x_0 e^{\bar{\sigma} \sqrt{s} z -\bar{\sigma}^2 s/2}\right)  ,
  \label{eq:V2b} \\
   && \tilde{\cal V}^{(3)}_t(x_0) =     
 {\overline{\Gamma}^2} \int_\RR dz  p(z)  \int_0^t ds 
  \left(   x^2 \partial_x^2  
  Q_s^{(0)} (x_0 e^{\bar{\sigma} \sqrt{s} z -\bar{\sigma}^2 s/2}\right)^2   
 \label{eq:V3b}  ,
\ea
where  $\overline{\Gamma}$ is defined by (\ref{def:barGamma}) and $\tilde{\cal H}_s(x)$ is defined by (\ref{def:tildeH}).
\end{proposition}

\subsection{Hedging Cost with a Modified (H) Hedging Strategy} 
\label{sec:mH}     

To facilitate comparison of  the schemes  at early exercise times 
we here consider the hedging scheme (H) using the delta at the Black-Scholes price 
at the effective volatility, $\delta^{\rm H}$,  however,  
modified in that the portfolio value is chosen to be the corrected price
$P(t,x)$ rather than the price $Q^{(0)}_t(x)$ at the effective volatility.   
We label this scheme ($\H$). 

Note that using Eq. (\ref{def:hE1})  we can write that the accumulated asymptotic 
hedging cost until time $t$  has the form:
\begin{equation}\label{eq:EH}
%\label{def:hE}
{E}_t^{\H} 
=   P(t,X_t)   - \int_0^t \delta^{\rm H}(s,X_s) dX_s  
=  P(0,X_0) 
    + N^{(1)}_t  +    
 \eps^{1/2}    \rho N^{(2)}_t    +o( \eps^{1/2})   .
\end{equation}
We then find that the hedging cost is characterized by
Lemma \ref{thm1b} and Proposition \ref{prop:main2b2b} upon the replacements:
$\N  \mapsto  N^{(2)} $ and $\overline{D} \tilde{\cal H}_t(x)  \mapsto     \big( x \partial_x\big) Q_t^{(1)}  (x) $.

\section{On Estimation of Effective Market Parameters}
\label{sec:est}

For the above results to be useful we must be able to 
estimate the three  market  parameters  discussed in Section 
\ref{sec:sum}
\ba\label{eq:pareff}
\bar\sigma, \quad \quad D=\sqrt{\eps}\rho \overline{D}, \quad\quad  \Gamma =  \sqrt{\eps} \, \overline{\Gamma} .
\ea 

We refer to $D=\sqrt{\eps}\rho \overline{D}$ as an  {\it effective pricing parameter}  with the price correction being scaled by this parameter. 
 The effective pricing parameter can together with the effective or historical
 volatility, $\bar\sigma$,  be calibrated from observation of vanilla option prices
 and the associated implied volatility skew.
    
   The parameter $\Gamma=\sqrt{\eps} \overline{\Gamma}$  is a  {\it hedging risk parameter}
 and the magnitude of   Vega risk martingale $N^{(1)}$ scales with this parameter. 
The hedging cost parameter can be calibrated from historical data.  Indeed,  by constructing the (HW) 
hedge for instance and 
recording  the accumulated cost over times $t_i, i=0,\ldots,n$ say, we will have an estimate
of the martingale  $N^{(1)}$ at these times from which the parameter 
$\sqrt{\eps} \overline{\Gamma}$ can be estimated via a least squares procedure that fits the empirical variance of the martingale  $N^{(1)}$
with the formula (\ref{def:V3})-(\ref{eq:V3}) in which only $ {\eps} \overline{\Gamma}^2$ is   unknown.
Then this ``historical'' hedging risk  parameter estimate can  
be used to project future hedging cost (mean and variance), thus,  the theory provides a bridge from 
historical to future hedging cost.  

 {In more complex market situations and modeling,  incorporating
for instance (random)  market  price of volatility risk and interest rate,  there 
will be additional parameters to estimate. 
The parameter $D$ can, however, be calibrated from 
the observed  volatility skew, even with a non-zero
market price of risk, see \cite{fouque00,fouque11} where 
  calibration based on the implied volatility skew is  discussed in detail.
The  historical volatility $\bar\sigma$  can be calibrated from historical 
    observations of the underlying price, while  a corrected effective volatility $\sigma^\star$
    can be calibrated from the implied volatility skew, and then the  difference of these volatility measures  leads to an estimate 
    of the market price of volatility risk, see \cite[Chapter 5]{fouque11}  for  details.   
    In \cite{25}  a data  calibration is carried out  and there
    a fast volatility factor on the scale of a few days was identified and effective parameters were estimated. 
    We stress  that the asymptotic regime  we consider here is one where the time  to maturity
    is large compared to the time scale of the volatility factor. Thus, we do not consider in this paper short time 
    to maturity asymptotics where the limit of small time to maturity is considered while other parameters are kept 
    fixed. One important consequence of our modeling and regime is that 
    the form of the price corrections and the hedging approach do not depend on the Hurst parameter
    in the rough case with $H\leq1/2$, the expressions  are in fact  ``universal'' as a consequence
    of the assumption of a fast mean reverting volatility factor.  The relevance of this regime
    and corroboration of the asymptotic results can be found in \cite{funa} which reports such a 
    universality based on numerical simulations.
   In \cite{funa} the authors find that the ``Hurst index under fractional volatility has a
crucial impact on option prices when the maturity is short and speed of mean reversion
is slow. On the contrary, the impact of the Hurst index on option prices reduces for
long-dated options'', and indeed it is the regime of long maturity horizons  that is considered here. 
   In the  numerical simulations in Section  \ref{sec:sim} we explore further the robustness of the  results with respect 
      to the assumption of fast mean reversion and indeed find that the (BS) hedging scheme presented here
      is robust with respect to the assumption  of fast mean reversion.   }

\section{Effective Market Parameters  Deriving from ExpfOU}
\label{sec:expou}

We discuss here the exponential fractional Ornstein-Uhlenbeck process
or ExpfOU model.
We then define the volatility by  $\sigma_t^\eps = F(Z_t^\eps)$ with 
\begin{equation}
   F(z) = \bar\sigma \exp \Big(  \frac{\omega z}{\sigma_{\rm z}}  - \omega^2\Big) ,
\end{equation}
which is such that $\left< F^2\right>  =\bar\sigma^2$.
Here, $\omega >0$ is a fluctuation parameter that measures the typical amplitude of the relative fluctuations
of the volatility:
\ban
    \frac{\left< F^4\right> -   \left< F^2\right>^2}{ \left< F^2\right>^2}   =  e^{4\omega^2} -1 .
 \ean
We introduce two parameters that summarize the information 
contained in ${\cal K}$ as defined in  (\ref{def:Keps}) (and the function ${\cal C}_{\cal K}$ defined in terms of ${\cal K}$ by (\ref{Cdef})):
\begin{eqnarray}
   \alpha &=&  \frac{\overline{D}}{\bar\sigma^3}  =
   \omega  e^{-  \frac{\omega^2}{2}} \int_0^\infty e^{2 \omega^2 {\cal C}_{\cal K}(s,0)} {\cal K}(s) ds 
  ,
\\
   \beta &=&    \frac{\overline{\Gamma}}{\bar\sigma^2}    =   
   \left(
     \omega^2   \int_0^\infty \int_0^\infty e^{4 \omega^2 {\cal C}_{\cal K}(s,s')}
      {\cal K}(s) {\cal K}(s') ds ds'  \right)^{1/2}
   .
 \end{eqnarray}
 These two parameters (with $\bar{\sigma}$) are necessary and sufficient to compute the 
 corrected price and hedging cost.
 In the  case of  a  ``classic'' ExpOU model with 
   ${\mathcal K}(t) = \sqrt{2}\exp(-t)$ they are given explicitly by:
 \ban
  \alpha =   e^{-\omega^2/2} \frac{  e^{2\omega^2 }  -   1 }{\sqrt{2} \omega }  , \quad \quad
 \beta  =  
\sqrt{\frac{1}{2} {E}_1(4\omega^2) -\frac{\gamma}{2} - \ln (2\omega) } ,
% \left( 2 \omega^2 \sum_{n=0}^\infty \frac{ \left( 4 \omega^2 \right)^n}{n! (n+1)^2}   \right)^{1/2}   .
 \ean
 with $E_1(z)=\int_z^\infty \frac{e^{-t}}{t} dt$ the exponential integral function and $\gamma \simeq 0.577$ the Euler constant.
 %{\bf Formula for $\beta$?  Not well defined for large $\omega$?}NB
 We plot $\alpha$  and $\beta$ as function of  $\omega$ in the ExpOU case  in  Figure  \ref{fig_ab}.   
  Note that $\alpha/\beta  \leq 1 $ is nearly independent of $\omega$ and approximately
 equal to $1$ for $\omega \leq 1$.
   
 \begin{figure}
\begin{center}
\begin{tabular}{c}
\includegraphics[width=6.8cm]{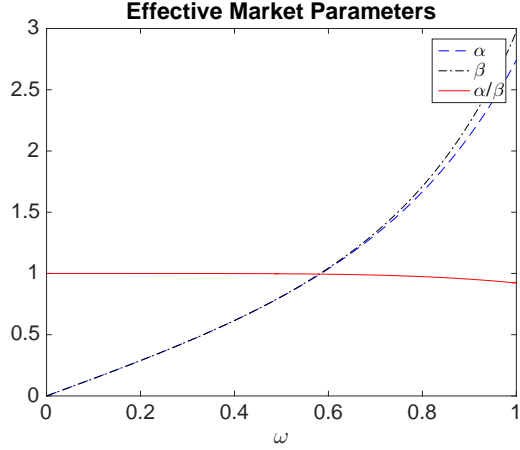}
\end{tabular}
\end{center}
\caption{      \label{fig_ab}
{Market parameters $\alpha$, $\beta$ in ExpOU case.}
}
\end{figure}

\section{Hedging Cost Statistics for European Call Options}\label{sec:call}
 
\subsection{The Call Price and its Delta and their Corrected Versions} 

 {
In Figure \ref{fig_relQ1} we show the normalized 
call price correction   $ \bar{\sigma}^2 Q^{(1)}_0/( K \overline{D})$  
and in Figure \ref{fig_bs} we show the Black-Scholes price relative to strike $Q^{(0)}_0/K$  for 
comparison. Note that for small maturities and moneyness the mean correction
is more important. 
Figure \ref{fig_relQ1b} corresponds to Figure \ref{fig_relQ1} only that we plot  the call price correction in terms
of a normalized implied volatility correction. 
 In Figure \ref{fig_deltabs} we show the  delta for the Black-Scholes 
price and in Figure \ref{fig_deltaQ1} we show the delta for the
normalized price correction.  If we assume a negative leverage parameter $\rho$
then  for short maturities and around the money the Black-Scholes delta  at the effective volatility  gives an {\it underhedged } situation in that the 
delta 
associated with the price correction is positive. We also see that for short maturities
and moneyness the Black-Scholes delta  gives an {\it overhedged } situation. 
  }
%as the stochastic volatility enhanced to ``liability'' associated with providing  the option.  

\begin{figure}
\begin{center}
\begin{tabular}{c}
\includegraphics[width=8.4cm]{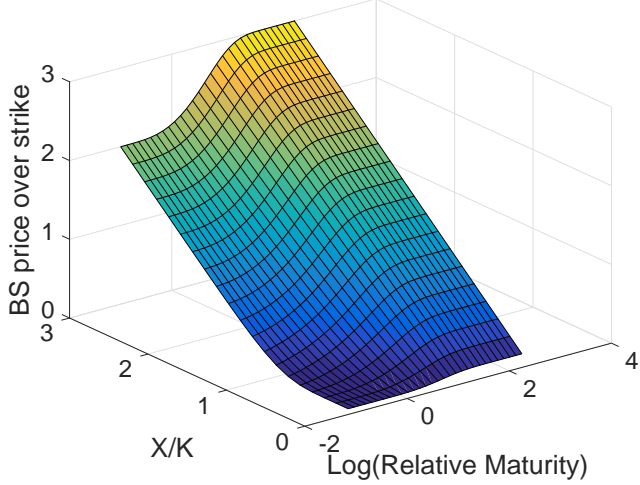}
\end{tabular}
\end{center}
\caption{
The figure shows the European call option price relative to strike: $Q_0^{(0)}/K$.
It is plotted as a function of {\rm Log relative maturity}, 
$\log_{10}(\tau) = \log_{10}(T\bar{\sigma}^2)$,  and 
moneyness,  $m=X_0/K$.  For short maturities we see the call payoff while there is a transition
regime to the large maturity limit,  the identity,   for relative maturity roughly around unity. 
    \label{fig_bs}
}
\end{figure}
 \begin{figure}
\begin{center}
\begin{tabular}{c}
\includegraphics[width=8.4cm]{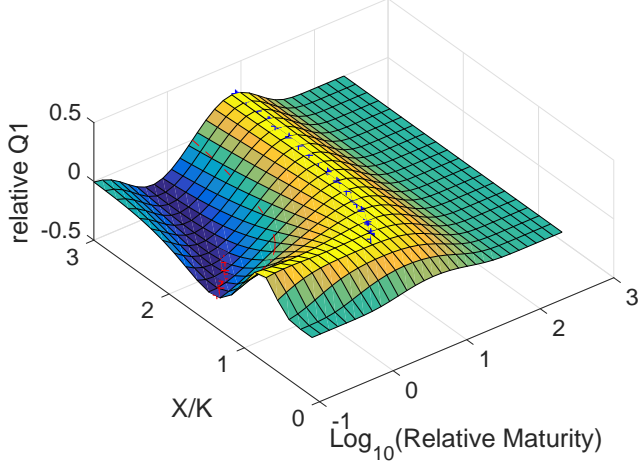}
\end{tabular}
\end{center}
\caption{ {
The figure shows the  normalized call price correction for the European call option: 
$\bar{\sigma}^2 {Q}^{(1)}_0/(K \overline{D})=-d_{_{-}}\exp(-d_{_{-}}^2/2)/\sqrt{2\pi}$.
It is plotted as a function of {\rm Log relative maturity}, 
$\log_{10}(\tau)= \log_{10}(T\bar{\sigma}^2)$,  and 
moneyness,  $X_0/K$ relative to the same domain as in Figure \ref{fig_bs}.  
We see that the correction is  large in the price transition zone
and that its maximal value is rather insensitive to the moneyness. 
We see moreover that when the time to maturity $T$ is large  relative to the diffusion time
$\bar\sigma^{-2}$ then the correction plays a minor role. 
 The red dashed line corresponds to $d_{_{-}}=0$, or $\tau =  2 \ln(m)$,  so that 
${Q}^{(1)}_0=0$ (with $m=X_0/K$). The blue and red crosses are asymptotic approximations, in $\ln(m)$,  for the partial derivative 
of $Q^{(1)}_0$ with respect to  maturity being zero. The blue crosses in the figure are 
$\tau=4+ 4\ln(m)$,
the red crosses are $\tau=\ln^2(m)$. }
    \label{fig_relQ1}
}
\end{figure}
\begin{figure}
\begin{center}
\begin{tabular}{c}
\includegraphics[width=8.4cm]{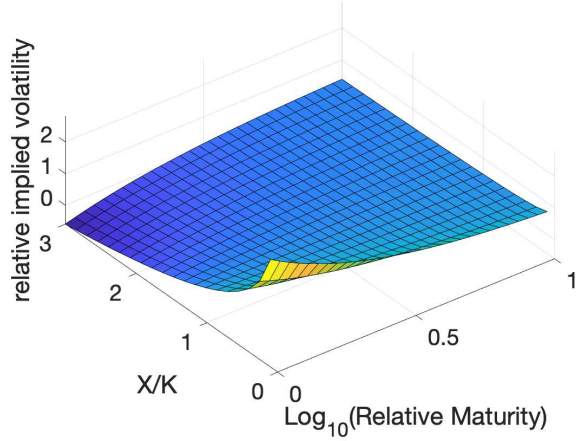}
\end{tabular}
\end{center}
\caption{  {
The figure plots  the call price  correction  as in the previous  figure however measured in terms of a relative implied volatility correction.
That is,  let $\bar{\sigma}+\Delta \sigma$ be the implied volatility associated with the price  correction, then the figure
plots $( \Delta \sigma  /  \bar{\sigma}    )( \bar{\sigma}^2/ (\sqrt{\eps} \rho \bar{D}  ) )=-d_{_-}/\sqrt{\tau}$.    }   \label{fig_relQ1b}
}
\end{figure}
   \begin{figure}
\begin{center}
\begin{tabular}{c}
\includegraphics[width=8.4cm]{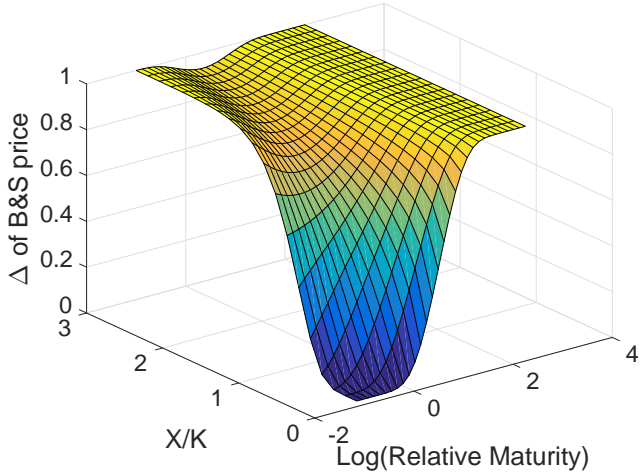}
\end{tabular}
\end{center}
\caption{   { The figure shows  the Black-Scholes delta at the effective volatility, that is 
$\partial_x Q_0^{(0)}$. 
It is plotted as a function of Log relative maturity, $\log_{10}(T\bar{\sigma}^2)$,  and 
moneyness,  $X_0/K$.   
For large moneyness  or maturity  
 this quantity  is close to unity corresponding to holding approximately a unit of the underlying
in the replicating portfolio, while for  small moneyness and maturity  this quantity  is close to zero corresponding to holding only cash.
By comparing with  Figure \ref{fig_relQ1}  it is seen that the price correction is small when approximately a unit of the underlying is held  in  the portfolio.  }
    \label{fig_deltabs}
}
\end{figure}
 \begin{figure}
\begin{center}
\begin{tabular}{c}
\includegraphics[width=8.4cm]{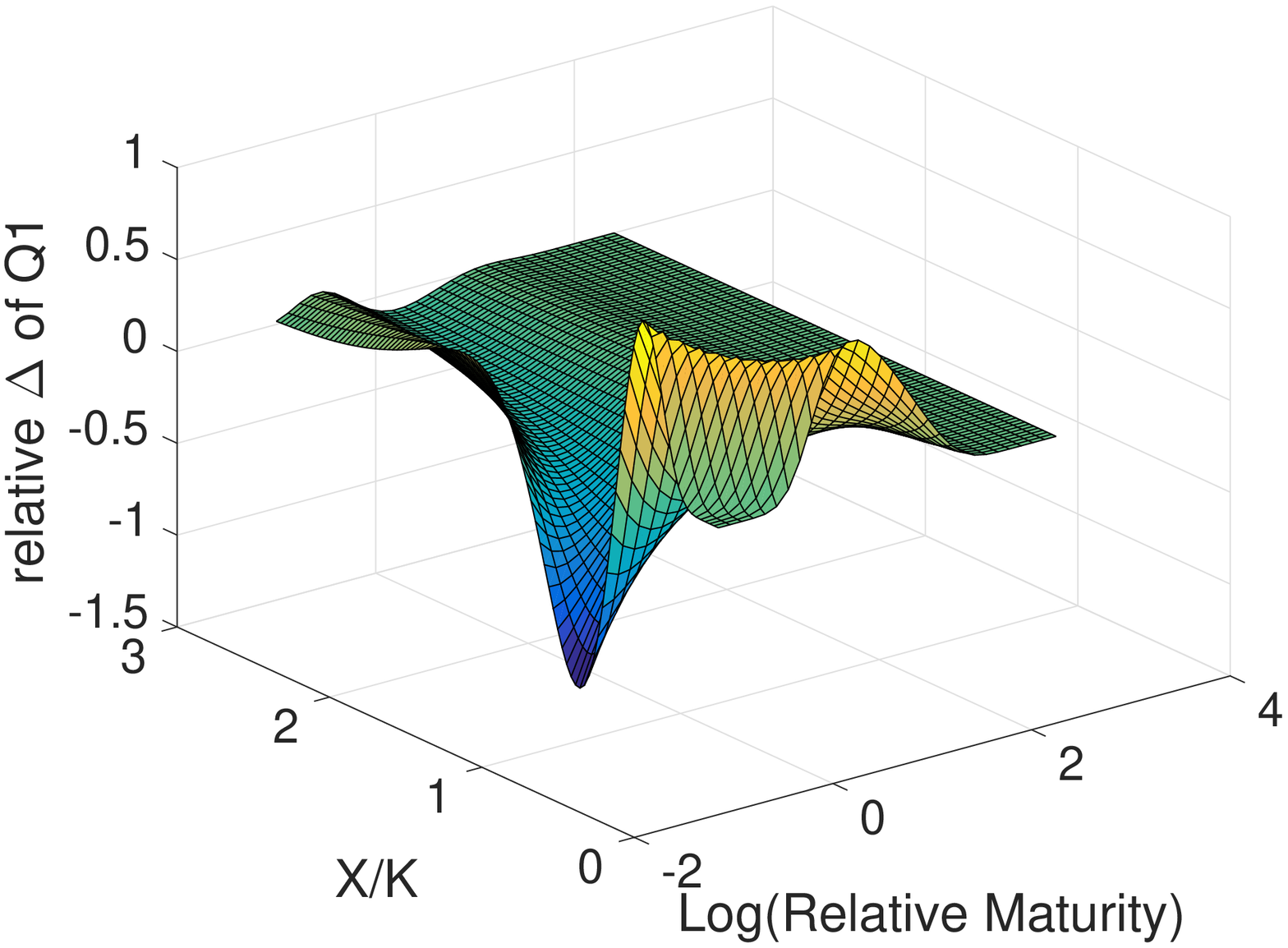}
\end{tabular}
\end{center}
\caption{    { The figure shows  the  delta  of the correction, that is 
$\bar{\sigma}^2  \partial_x Q^{(1)}_0 / \overline{D} $. 
It is plotted as a function of {Log relative maturity}, $\log_{10}(T\bar{\sigma}^2)$,  and 
moneyness,  $X_0/K$.    
We remark that far out of the money and for small maturities, moreover,  with $\rho<0$,  the correction to the price
gives a negative correction to the number of underlyings held in the portfiolio in the (HW) case.   
These  corrections   to the  portfolio weights will be of order $O(\sqrt{\eps})$ in our regime.  \newline \newline }
%
%It shows that for out of the money and for small maturities  the Black-Scholes $\Delta$ at the effective volatility 
%gives an
%``underhedged'' portfolio, moreover it shows that around the money the Black-Scholes $\Delta$
%at the effective volatility is in fact  ``overhedged''. 
     \label{fig_deltaQ1}
}
\end{figure}

\subsection{Call Hedging Risk}

%Denote 
%\ban
%%\tau_{\rm rm} &=&   {\eps}{\bar\sigma^2} ,  \quad \quad 
%%\tau_t & = &    {(T-t)}{\bar\sigma^2}  ,  \quad \quad 
%%m_t   =    \frac{X_t}{K}   ,    \\  %\quad \quad 
%%\theta =   \frac{t}{T}  \in (0,1) ,\\
%d   &=& 
% \frac{\ln(m)}{\sqrt{ \bar\sigma^2 T }}  - \frac{\sqrt{\bar\sigma^2 T }}{2}   .
%\ean

In Proposition \ref{prop:sum}  we gave the expressions of the means and variances of 
the hedging costs in the case of a general payoff. 
The explicit expressions for the normalized functions $g,v, w^{\rm C}$,
for ${\rm C=H, BS, \H}$, follow from the
propositions in Section \ref{sec:hedging}.
Here we consider the situation with a European call.
Then we can use the results in Appendix \ref{app:mom}, Eqs.~(\ref{eq:greek1}-\ref{eq:4th}),  to get explicit expressions 
for the  normalized  functions $g,v, w^{\rm C}$. 
 
 \begin{proposition}
 \label{prop:sum2}
In  the case of a European call option $h(x)=(x-K)^+$  and using  the  notation  in Proposition \ref{prop:sum},
the  normalized  functions  $g$, $v$, $w^{\rm C}$ depend on $d_{_-}$ and $\theta=t/T$ only:
\ban
  &&  \frac{g(d_{_{-}})}{K}    =    -  \frac{   d_{_{-}}  \exp(-d_{_{-}}^2/2) }{\sqrt{2\pi} }  , \\   
  &&  \frac{v(\theta;d_{_{-}})}{K^2} =     \frac{1}{2\pi}
   \int_0^\theta  \exp\Big( {-\frac{d_{_{-}}^2}{1+s}} \Big)   \frac{ds}{\sqrt{1-s^2}}   ,  \\
  &&  \frac{ w^{\rm H}(\theta;d_{_{-}})}{K^2}   = 
    \frac{ 1 }{ \pi } 
 \int_0^\theta    \exp\Big(- \frac{d_{_{-}}^2}{1+s}\Big)  
   \frac{(\theta-s)}{(1-s)^2}   \big[
        2 f_4(s,d_{_{-}})  -  f_0(s)
      \big] ds
           -     \theta^2   \frac{   d_{_{-}}^2  \exp(-d_{_{-}}^2) }{{2\pi} }   , \\   
           &&   w^{\rm BS}(\theta;d_{_{-}})   = -  v(\theta; d_{_-})  , \\
   && \frac{ w^{\rm \H}(\theta;d_{_{-}})}{K^2}   =    \frac{1}{2 \pi } 
 \int_0^\theta \exp\Big(- \frac{d_{_{-}}^2}{1+s}\Big)  \frac{1}{ (1-s) }   
\big[ f_4(s,d_{_{-}})   - f_0(s)\big] ds
     ,
\ean
with $f_j, j=0,2,4$ defined in  Proposition \ref{lem:gauss} and $d_{_-},\tau$ defined by (\ref{eq:bsparameter0}).
%\ban
%  d_{_{-}} = \frac{\log(X_0/K)}{\sqrt{\tau}} - \frac{\sqrt{\tau}}{2}   ,\quad \tau=\bar{\sigma}^2 T .
%\ean
\end{proposition}

It then  follows that, as $\eps \to 0$,
\ban
{\rm Var} \big( Y^{\rm BS}_t \mid {\cal F}_0 \big) =  \eps \left( \frac{\overline{\Gamma}^2}{\bar\sigma^2} 
 -  \frac{\rho^2 \overline{D}^2}{\bar\sigma^4} \right) v(\theta;d_{_-})  =  \left( \frac{\Gamma^2}{\bar\sigma^2} 
 -  \frac{D^2}{\bar\sigma^4} \right) v(\theta;d_{_-})  
 , 
\ean
which gives Eq. (\ref{eq:main}). 
In Figure \ref{fig_AMV} we plot $v$ as a function of normalized maturity and moneyness.
We see that $v$  is  large for  large exercise times 
 and small values of~$d_{_{-}}$. 
%We remark that $d_{_{-}}$ may be seen as a measure of a  distance from  the strike and 
%becomes large in the regime of  large times to maturity.   
In Figures \ref{fig_AH} and \ref{fig_AtH} we show respectively 
$w^{\rm H}$ and $w^{\rm \H}$.   In the regime  of large exercise times  and 
small values of $d_{_{-}}$ these schemes offer a slight advantage relative to the
(HW) scheme in terms of cost variance.  
Note that at maturity the two schemes (H) and ($\H$) have the same cost.
Recall, however, that for the scheme (H) it is assumed
that the option can be traded at the price $Q^{(0)}$ so the schemes cannot be 
compared directly other than at maturity when  $Q^{(1)}_T=0$. 
In Figure \ref{fig_Amean} we show the function $g/K$ which describes the 
coherent cost correction as a function of $d_{_{-}}$, we see that this correction 
is maximal for $d_{_{-}}$ around unity.  
  
\begin{figure}
\begin{center}
\begin{tabular}{c}
\includegraphics[width=8.4cm]{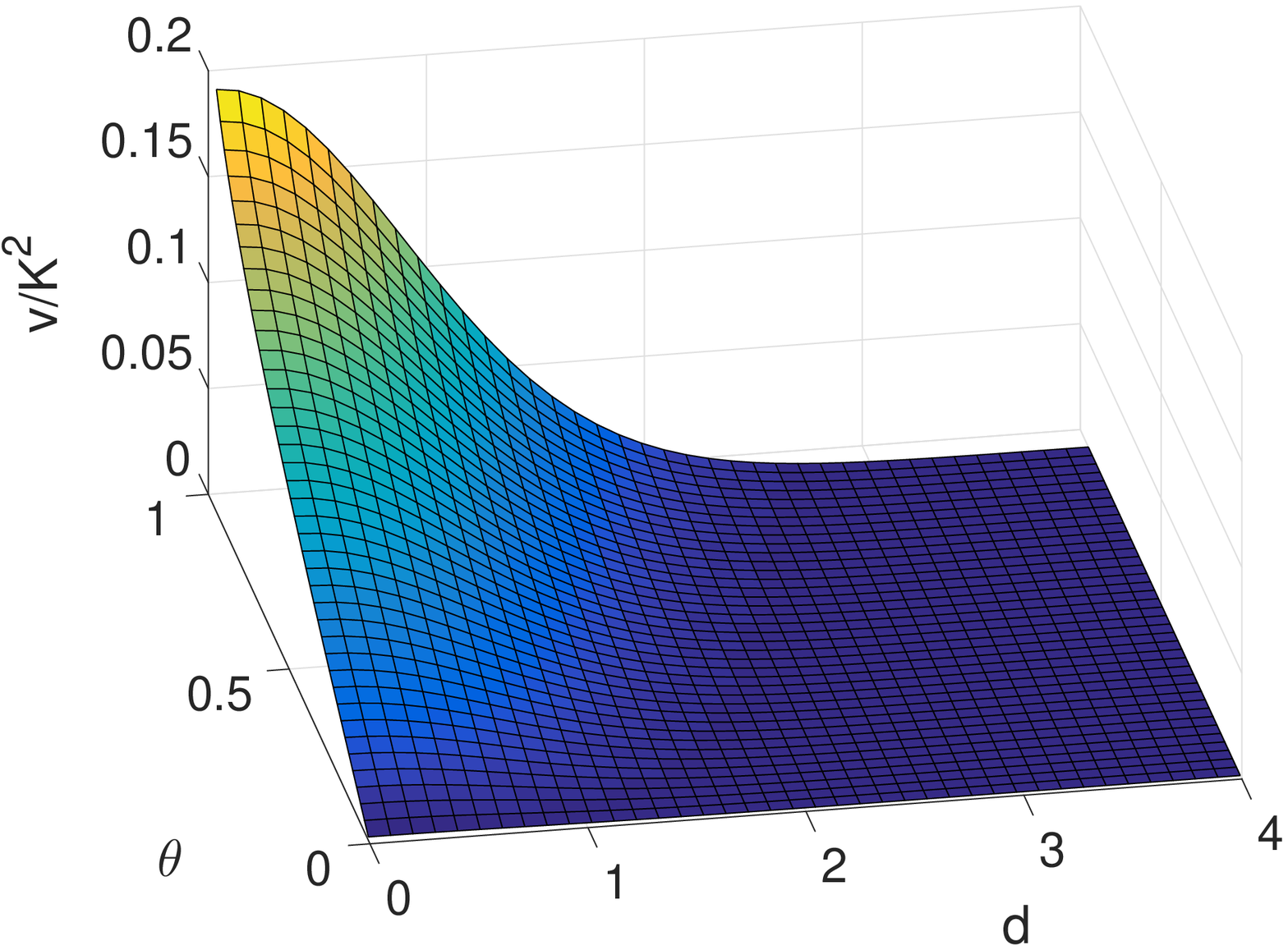}
\end{tabular}
\end{center}
\caption{The figure shows  the hedging cost variance function $v(\theta;d_{_{-}})/K^2$.   
    \label{fig_AMV}
}
\end{figure}

\begin{figure}
\begin{center}
\begin{tabular}{c}
\includegraphics[width=8.4cm]{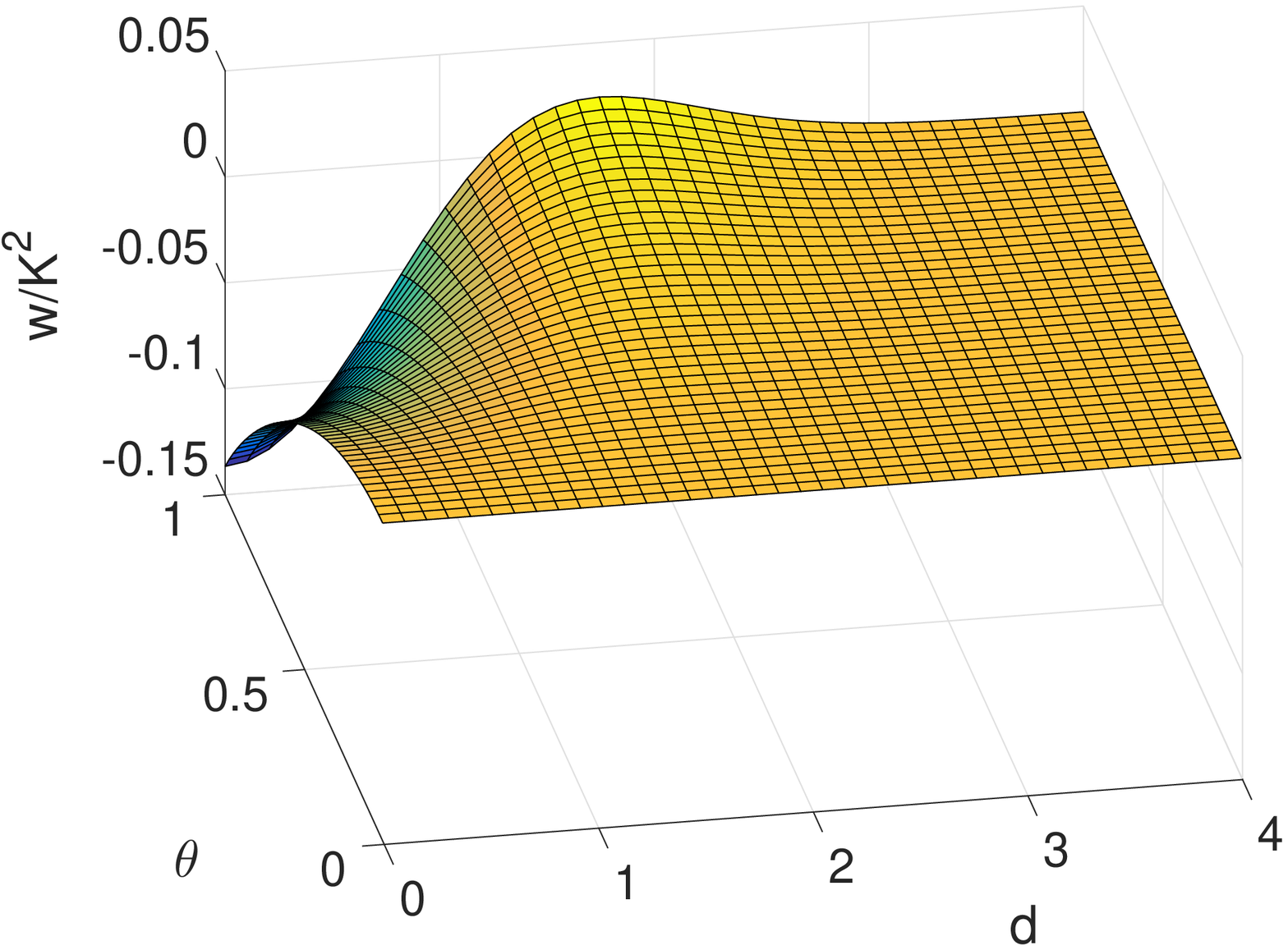}
\end{tabular}
\end{center}
\caption{The figure shows  the hedging cost variance function $w^{\rm H}(\theta;d_{_{-}})/K^2$.
    \label{fig_AH}
}
\end{figure}

\begin{figure}
\begin{center}
\begin{tabular}{c}
\includegraphics[width=8.4cm]{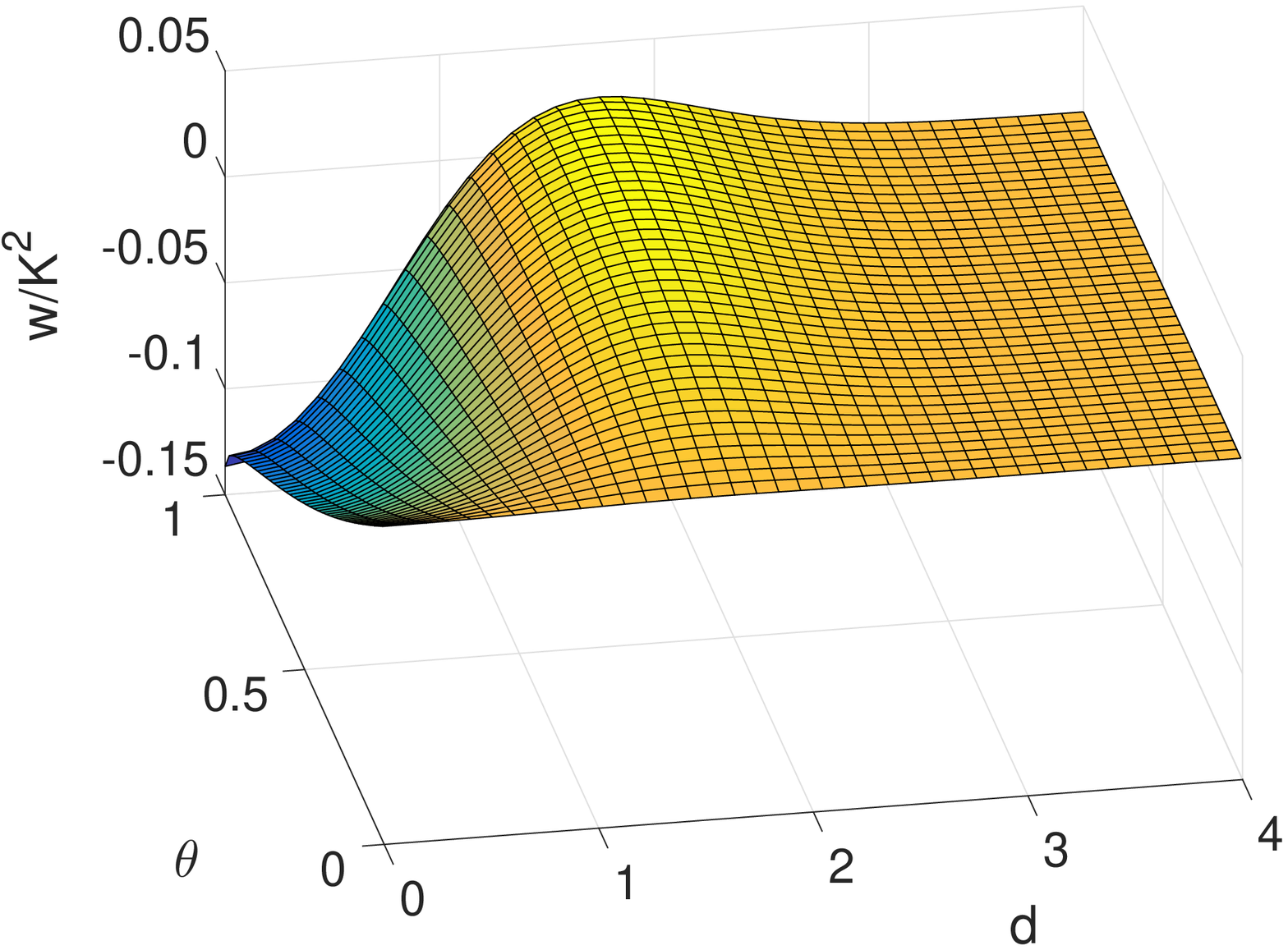}
\end{tabular}
\end{center}
\caption{The figure shows  the hedging cost variance function $w^{\rm \H}(\theta;d_{_{-}})/K^2$.   
    \label{fig_AtH}
}
\end{figure}

\begin{figure}
\begin{center}
\begin{tabular}{c}
\includegraphics[width=7.8cm]{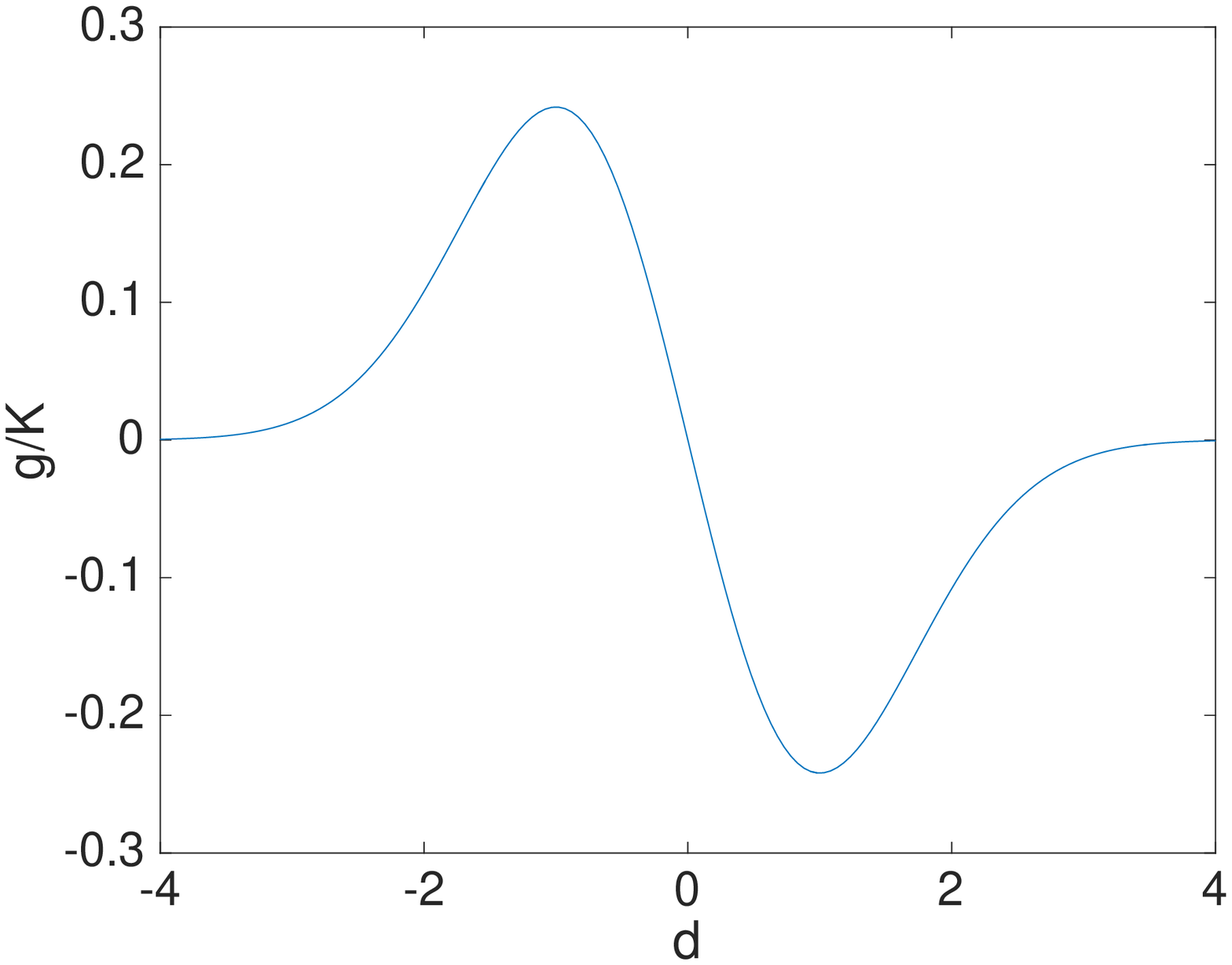}
\end{tabular}
\end{center}
\caption{  The figure shows  the coherent cost correction function $g(d_{_{-}})/K$.   
    \label{fig_Amean}
}
\end{figure}

\subsection{Optimality of Practitioners Scheme}
\label{sec:opt}   {
In the context of our modeling the practitioners scheme (BS) has the lowest 
risk (i.e. cost variance) among the schemes that we have considered (${\rm H, HW,  BS, \H}$). 
Here, we show that in fact the practitioners approach is the
optimal scheme amongst all DA  hedging strategies in the context of a call and for sufficiently 
small $\eps$.  } 

\begin{definition}
A DA hedging scheme is based on a replication  
portfolio of value $P$ of the form (\ref{eq:generalformportfolio}) with 
the number of underlyings $a_t$ being a smooth function of $t$ and $X_t$.
\end{definition}

\begin{proposition}
\label{prop:optim}
Let ${\cal A}(t,x)$ be a smooth and bounded function. 
Let $a_t ={\cal A}(t,X_t)$ be the number of underlyings in a replication 
portfolio of value $P(t,X_t)$. 
 Let 
\begin{equation}
     E^*_t   =  P(t,X_t)  - \int_0^t a_s dX_s
\end{equation}
be the cost associated to the hedging strategy $a_t$.
Then we have  {up to terms of order $o(\eps)$:}
\begin{equation}
\EE[E^*_t \mid {\cal F}_0 ] =  P(0,X_0), \quad  {\rm Var}(E^*_t \mid{\cal F}_0  )    \geq   {\rm Var}(E^{\rm BS}_t  \mid {\cal F}_0) , 
\quad  t \in [ 0,T]  .
\end{equation}
\end{proposition}
 {This proposition shows that  there is one scheme, the (BS) scheme, 
that  is the asymptotic optimal DA scheme  for any exercise time $t \leq T$.}
\begin{proof}
We write the cost as
$$
E^*_t  = P(t,X_t)  -   \int_0^t \delta^{\rm HW}(s,X_s) dX_s
+  \int_0^t \left(  \delta^{\rm HW}(s,X_s)   - a_s  \right) dX_s  .
$$
We first address the most interesting case 
consistent with the regime addressed here, that is, the case when
 ${\cal A} (t,x) - \partial_x Q^{(0)}_t(x)$ is of order $\sqrt{\eps}$:
$$
 {\cal A} (t,x) = \partial_x Q^{(0)}_t(x) +\sqrt{\eps}  {\cal A}_1 (t,x)  .
 $$
Then
$$
E^*_t  =
 P(0,X_0) + 
 N^{(1)}_t  +  \check{N}_t +o(\sqrt{\eps}) ,
 $$
with  (using Eq.~(\ref{eq:D2}))
$$
 \check{N}_t  = \sqrt{\eps}
  \int_0^t  \check{\cal A}_s(X_s) \sigma_s^\eps dW^*_s  ,\quad \quad \check{\cal A}_s(x) =
  \big( \rho Q^{(1)}_s(x)- {\cal A}_1(s,x)  \big) x .
%  \eps^{-1/2} \left(  \delta^{\rm HW}(s,x)   - {\cal A}(s,x)  \right) x .
$$
The two martingales $N^{(1)}$ and $\check{N}$ have amplitudes of order $\sqrt{\eps}$.
Using Eq.~(\ref{a32}) 
we get
\begin{align*}
\EE \big[ (N^{(1)}_t)^2 \mid {\cal F}_0 \big] 
& = \EE\Big[ \int_0^t (x^2\partial_x^2 Q^{(0)}_s )(X_s)^2 (\vartheta^\eps_s)^2 ds \mid {\cal F}_0\Big]\\
& = \eps \overline{\Gamma}^2 \EE\Big[ \int_0^t (x^2\partial_x^2 Q^{(0)}_s )(X_s)  ^2 ds \mid {\cal F}_0\Big] + o(\eps) ,
\end{align*}
in the sense that
$$
\lim_{\eps \to 0}
\EE\left[\left| \eps^{-1} \EE \big[ (N^{(1)}_t)^2 \mid {\cal F}_0 \big] 
-
 \overline{\Gamma}^2 \EE\Big[ \int_0^t (x^2\partial_x^2 Q^{(0)}_s )(X_s)  ^2 ds \mid {\cal F}_0\Big]
 \right|\right]=0.
$$
Similarly, using Eqs.~(\ref{a31}) and (\ref{eq:comparesigmasigmaeps}), 
\begin{align*}
\EE \big[ (\check{N}_t)^2 \mid {\cal F}_0 \big]  
& = \eps   \EE\Big[ \int_0^t \check{\cal A}_s(X_s)^2 (\sigma^\eps_s)^2  ds \mid {\cal F}_0\Big] \\
& = \eps \bar{\sigma}^2 \EE\Big[ \int_0^t \check{\cal A}_s(X_s)^2 ds \mid {\cal F}_0\Big] + o(\eps) ,\\
\EE \big[ N^{(1)}_t \check{N}_t \mid {\cal F}_0 \big] 
& = \sqrt{\eps}  \rho \EE\Big[ \int_0^t \big( (x^2\partial_x^2 Q^{(0)}_s) \check{\cal A}_s \big)(X_s)  \sigma^\eps_s \vartheta^\eps_s ds \mid {\cal F}_0\Big] \\
& = \eps \rho \overline{D} \EE\Big[ \int_0^t \big( (x^2\partial_x^2 Q^{(0)}_s) \check{\cal A}_s \big)(X_s)  ds \mid {\cal F}_0\Big] + o(\eps) .
\end{align*}
Therefore, we find  
to leading order
\ban
\check{\rho}_t  &=&  {\rm Corr}\big(  N^{(1)}_t,  \check{N}_t  \mid{\cal F}_0 \big) 
=  
 \frac{\rho \overline{D} \EE\Big[   \int_0^t  \big( (x^2\partial_x^2 Q^{(0)}_s) \check{\cal A}_s \big) (X_s)ds   \mid {\cal F}_0 \Big] }
  { \bar\sigma \overline{\Gamma}  \sqrt{
  \EE\Big[   \int_0^t   ( x^2\partial_x^2 Q^{(0)}_s)(X_s)^2   ds \mid {\cal F}_0 \Big]
 \EE\Big[  \int_0^t   \check{\cal A}_s(X_s) ^2  ds  \mid {\cal F}_0  \Big]  } } ,
 \ean 
 so that by Cauchy-Schwarz inequality
$   |\check{\rho}_t|    \leq  \overline{\rho} $,
where
\begin{equation}
 \overline{\rho}  =  \frac{\rho \overline{D}  }
  { \bar\sigma \overline{\Gamma} } =  \frac{D}{\bar\sigma \Gamma}   .
\end{equation}
Thus, using Proposition  \ref{prop:sum2}   and denoting
\ban
  \alpha_t = \sqrt{ \frac{ {\rm Var}(\check{N}_t \mid {\cal F}_0 ) }{  {\rm Var}(N^{(1)}_t \mid {\cal F}_0 )   } } ,
\ean
  we have 
\ban
\nonumber
    {\rm Var}\big( E^*_t  \mid {\cal F}_0 \big) &=&   {\rm Var} \big(  N^{(1)}_t  \mid {\cal F}_0 \big) 
    \big( 1 + 2\check{\rho}_t \alpha_t + \alpha_t^2 \big) 
    \geq {\rm Var}\big(  N^{(1)}_t \mid {\cal F}_0  \big) 
    \big( 1 -  2\overline{\rho}  \alpha_t + \alpha_t^2 \big)   \\
    &\geq&      {\rm Var}\big(  N^{(1)}_t  \mid {\cal F}_0 \big) 
    \big( 1 -  \overline{\rho}^2 \big) = {\rm Var}\big( E^{\rm BS}_t \mid {\cal F}_0 \big) ,
\ean
which proves the desired result.

If we assume that ${\cal A} (t,x) - \partial_x Q^{(0)}_t(x)$ is smaller than $\sqrt{\eps}$,
then we easily find that 
$
E^*_t  =
 P(0,X_0) + 
 N^{(1)}_t  +o(\sqrt{\eps}) ,
 $
and therefore
$ {\rm Var}\big( E^*_t  \mid {\cal F}_0 \big) 
=
 {\rm Var}\big(  N^{(1)}_t  \mid {\cal F}_0 \big) $
 up to terms of order $o({\eps})$.
 
If we assume that ${\cal A} (t,x) - \partial_x Q^{(0)}_t(x)$ is  larger than $\sqrt{\eps}$:
$$
 {\cal A} (t,x) = \partial_x Q^{(0)}_t(x) + {\eps}^p  {\cal A}_1 (t,x) ,
 $$
 with $p\in [0,1/2)$, 
then
$$
E^*_t  =
 P(0,X_0) +  \hat{N}_t +o( {\eps}^p) ,
 $$
with 
$$
 \hat{N}_t  = \eps^p 
  \int_0^t  \hat{\cal A}_s(X_s) \sigma_s^\eps dW^*_s  ,\quad \quad  \hat{\cal A}_s(x) =
- {\cal A}_1(s,x) x .
$$
We then have
\begin{align*}
\EE \big[ (N^{(1)}_t)^2 \mid {\cal F}_0 \big]   =O( \eps) , \quad \quad
\EE \big[ (\hat{N}_t)^2 \mid {\cal F}_0 \big]  
=  \bar{\sigma}^2 \eps^{2p} \EE\Big[ \int_0^t \hat{\cal A}_s(X_s)^2 ds \mid {\cal F}_0\Big] (1 + o(1)) ,
\end{align*}
which shows that
$$
    {\rm Var}\big( E^*_t  \mid {\cal F}_0 \big) =
     {\rm Var}\big(  \hat{N}_t  \mid {\cal F}_0 \big)  (1+o(1)) \gg 
     {\rm Var}\big(  N^{(1)}_t  \mid {\cal F}_0 \big)  \geq {\rm Var}\big( E^{\rm BS}_t \mid {\cal F}_0 \big) .
$$

For completeness (and to prove the last inequality in (\ref{eq:main})), 
we also remark that, by using Eq.~(\ref{eq:pareff}) and Lemma \ref{lem:K2}, we have
   \ban
    |\overline{\rho}|    = \frac{|\rho|}{\bar\sigma} \ 
     \lim_{\eps \to 0}    \frac{\EE\left[ \vartheta^\eps_t \sigma_t^\eps \right] }
     {\sqrt{ \EE\left[  (\vartheta^\eps_t )^2\right]}}    \leq
        |\rho|  \frac{1}{\bar\sigma}   \lim_{\eps \to 0} 
          \sqrt{ \EE\left[ ( \sigma^\eps_t  )^2\right]}  = |\rho| .
\ean
\end{proof}

\section{ {Numerical Illustration and Robustness}} \label{sec:sim} 

We illustrate  the performance of the different hedging schemes numerically.
We consider the case of a European call. 
Recall that we here define the implied volatility by 
 $\sigma(t,x)$ solving 
\begin{equation}
\label{def:impliedvolII}
P(t,x ) =  Q^{(0)}(t,x;\sigma(t,x))  \,
\end{equation}
where   $P(t,x)$  is  the corrected price: 
\begin{equation}
\label{eq:PcorrectedII}
   P(t,x )  
    =  Q^{(0)}(t,x; \bar\sigma) + 
   D(T-t) \big( x \partial_x  (x^2 \partial_x^2 ) \big) Q^{(0)}(t,x;\bar\sigma) ,
%   \\
%    &\equiv&  Q^{(0)}(t,x; \bar\sigma) + 
%   D Q^{(1)}(t,x;\bar\sigma)    .
\end{equation}
and where  $\bar\sigma$ is the  historical volatility and $Q^{(0)}$  is the  standard Black-Scholes price.  

In the call case  the hedging deltas are explicitly given by

--The ``historical'' (H) delta:  
\ba\label{eq:da} 
 \delta^{\rm H}(t,x)  & = & \partial_x   Q^{(0)}(t,x;  \bar\sigma )     =   {\cal N}(d_+) ,
 \ea 
 for  ${\cal N}$ the cumulative normal distribution.

 -- The  Black-Scholes (BS) or practitioners delta: 
 \ba\label{eq:dc} 
  \delta^{\rm BS}(t,x)    & = & \partial_x 
  Q^{(0)}(t,x;\sigma)|_{\sigma=\sigma(t,x)} =  \delta^{\rm H}(t,x)   +  {\cal D}  \frac{  d_{_-}^2  \exp(-d_{_-}^2/2) }{x \sqrt{\tau} }  ,
  \ea
  for  ${\cal D}$ the hedging parameter which in terms of the underlying parameters  has the representation:
  \ba\label{eq:Dh}
    {\cal D}  =     \frac{  \sqrt{\eps} \rho \bar{D}  K}{  \sqrt{2\pi} \bar\sigma^2  } ,
  \ea
 and with notation: 
\begin{equation*}
   d_{_{\pm}} = \frac{\log(X_0/K)}{\sqrt{\tau}}    \pm \frac{\sqrt{\tau}}{2}  , 
  \quad \quad \tau = \bar\sigma^2(T-t) .
\end{equation*}

 --The  Hull-White (HW) delta:
\ba\label{eq:db} 
 \delta^{\rm HW}(t,x)  & = & \partial_x   P(t,x )  =  \delta^{\rm H}(t,x)   +  {\cal D}  \frac{ (d_{_-}^2-1)   \exp(-d_{_-}^2/2) }{x \sqrt{\tau} }   , 
 \ea
with  ${\cal D}$  defined as in Eq. (\ref{eq:Dh}).

The model for the  underlying and the volatility is the expfOU model  introduced in Section \ref{sec:expou}:
 \ban
&& dX_t  =  X_t   \sigma_t^\eps    dW^*_t  ,   \quad  \quad
 \sigma_t^\eps    =    \bar\sigma \exp \Big(  \frac{\omega Z_t^\eps}{\sigma_{\rm z}}  - \omega^2\Big) , 
 \ean
 for $Z_t^\eps$ a fractional Ornstein-Uhlenbeck process with rate of mean reversion $\eps$
  and Hurst parameter $H$, that is, a  scaled Gaussian process with representation
   \begin{equation*}
\label{eq:ZgenII}
Z^\eps_t = \sigma_{\rm z} \int_{-\infty}^t {\cal K}^\eps(t-s) dW_s,
\end{equation*}
 where    $W_t, W^*_t$  are standard Brownian motions   with correlation coefficient $\rho$ 
  under the  historical  measure  and where the kernel
${\cal K}^\eps$ is discussed in Section \ref{sec:svmodel}.  
  Recall  that here we assume that the drift in the price is vanishingly
small so that the price is a martingale under the historical measure.   

The hedging cost with the volatility fluctuations is: 
$$
  E^{\rm C}_T  = h(X_T) - \int_0^T \delta^{\rm C}(s,X_s)  dX_s  , \quad {\rm C}={\rm H},{\rm BS},{\rm HW} . 
$$
We simulate many independent price trajectories $(X_s)_{    0  \leq s \leq  T}$ using a spectral approach
and compute the associated hedging costs.
We  then define the  relative risk in the hedging cost by:
\ba\label{eq:Hc}
      C^{\rm C}(T,x_0) =   \frac{St.Dev [ E^{\rm C}_T ] }{ Q^{(0)}(T,x_0;\bar\sigma)   }   , 
\ea
where the standard deviation is with respect to the simulated paths.  
The approach to calibration we take here is that we assume that  historical price  paths are available 
and we choose the hedging parameter ${\cal D}$ as the one that minimizes (BS) hedging risk (to evaluate the risk of the (BS) hedging cost)
or the one  that minimizes  the (HW) hedging risk  (to evaluate the risk of the (HW) hedging cost).
   
In Figure \ref{fig_hedgecost} we show the hedging cost risk as a function of moneyness parameter  $x/K$ with $K$ the call strike. 
We use the parameters $T=1$,
$\eps=.05$, $\bar\sigma=.5$, $\omega=.5$,  and $\rho=-.5$. Thus we consider a rapidly mean  reverting  volatility factor
and a strong leverage.   Note that indeed the (BS) scheme is the optimal approach for all considered values of the
moneyness, while the (HW) scheme performs approximately as the (H) scheme.

Figure \ref{fig_hedgecostslow}  corresponds to Figure  \ref{fig_hedgecost}  only that $\eps=1$ so that we are not in the rapidly mean reverting
regime. All schemes are then associated with approximately the same risk. Thus, even though we use 
the (BS) hedging scheme outside of its regime of optimality  it  performs as well as the classic (H) scheme. 

Figures  \ref{fig_hedgecostrough}  and \ref{fig_hedgecostroughslow} correspond to Figures 
   \ref{fig_hedgecost}  and  \ref{fig_hedgecostslow}  only that here $H=.1$, which means that we consider
   a rough volatility regime. It is interesting to note that the reduction of the hedging cost uncertainty 
   by the (BS) scheme is larger than in the classic Markovian case (Figure \ref{fig_hedgecostrough}). 
   It is all the more advantageous to use the (BS) scheme as the volatility is rougher.
   This advantage is still noticeable even when $\eps=1$ (Figure \ref{fig_hedgecostroughslow}).

\begin{figure}
\begin{center}
\begin{tabular}{c}
\includegraphics[width=6.8cm]{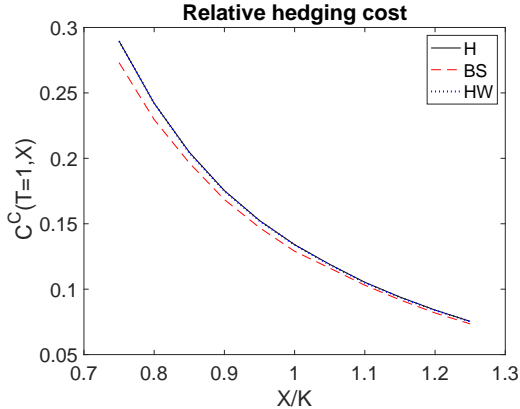}
\end{tabular}
\end{center}
\caption{      \label{fig_hedgecost}
{ The figure  shows the relative hedging cost uncertainty defined in Eq.~(\ref{eq:Hc}). The solid, dashed and dotted lines correspond to
the (H), (BS) and (HW) schemes respectively.
 For all considered values of the moneyness  the (BS) scheme has the smallest risk.  
The figure corresponds to a rapidly mean reverting and Markovian volatility factor:  $H=1/2$ and  $\eps=.05$.
The hedging parameter  is optimized so as to minimize the mean cost uncertainty over moneyness and realizations for the (BS) scheme
(${\cal D}=-.010$), respectively the (HW) scheme   (${\cal D}=-.005$). 
The theoretical parameter in   Eq.  (\ref{eq:Dh})  is  ${\cal D}=-.014$  (this theoretical value derives from the classic
delta definitions in the  fast mean reversion regime and is not optimized with respect to hedging risk). 
  } }
  
\end{figure}
\begin{figure}
\begin{center}
\begin{tabular}{c}
\includegraphics[width=6.8cm]{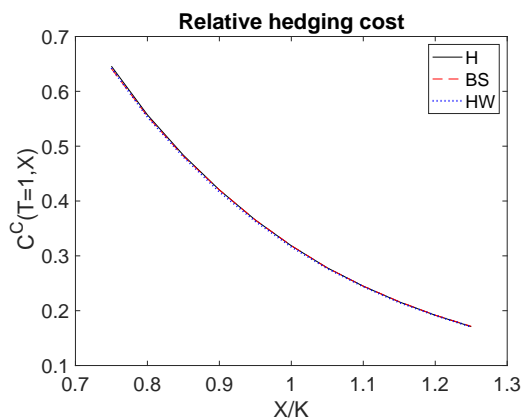}
\end{tabular}
\end{center}
\caption{      \label{fig_hedgecostslow}
{ The figure  shows the relative hedging cost uncertainty defined in Eq.~(\ref{eq:Hc}). The solid, dashed and dotted lines correspond to
the (H), (BS) and (HW) schemes respectively.   
The hedging cost is almost identical for the 3 hedging schemes.  
The figure corresponds to a slowly  mean reverting and Markovian volatility factor:  $H=1/2$ and $\eps=1$.
The hedging parameter is optimized  as in Figure \ref{fig_hedgecost}. 
 } }
 
\end{figure}
\begin{figure}
\begin{center}
\begin{tabular}{c}
\includegraphics[width=6.8cm]{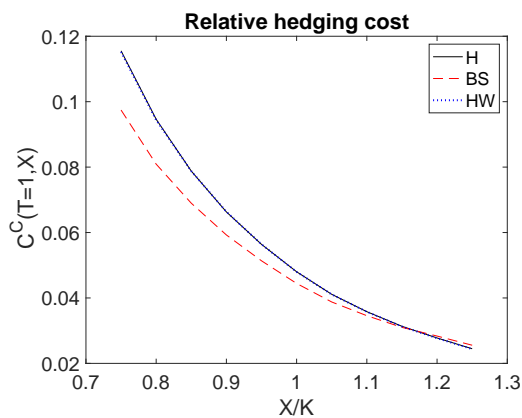}
\end{tabular}
\end{center}
\caption{      \label{fig_hedgecostrough}
{ The figure  shows the relative hedging cost uncertainty defined in Eq.~(\ref{eq:Hc}). The solid, dashed and dotted lines correspond to
the (H), (BS) and (HW) schemes respectively.   For all considered values of the moneyness  the (BS) scheme has the smallest risk
and the relative gain is larger than in the  Markovian case illustrated  in Figure \ref{fig_hedgecost}. 
The figure corresponds to a rapidly mean reverting and non-Markovian volatility factor:  $H=.1$  and $\eps=.05$.
The hedging parameter is optimized  as in Figure \ref{fig_hedgecost}. 
 } }
\end{figure}
\begin{figure}
\begin{center}
\begin{tabular}{c}
\includegraphics[width=6.8cm]{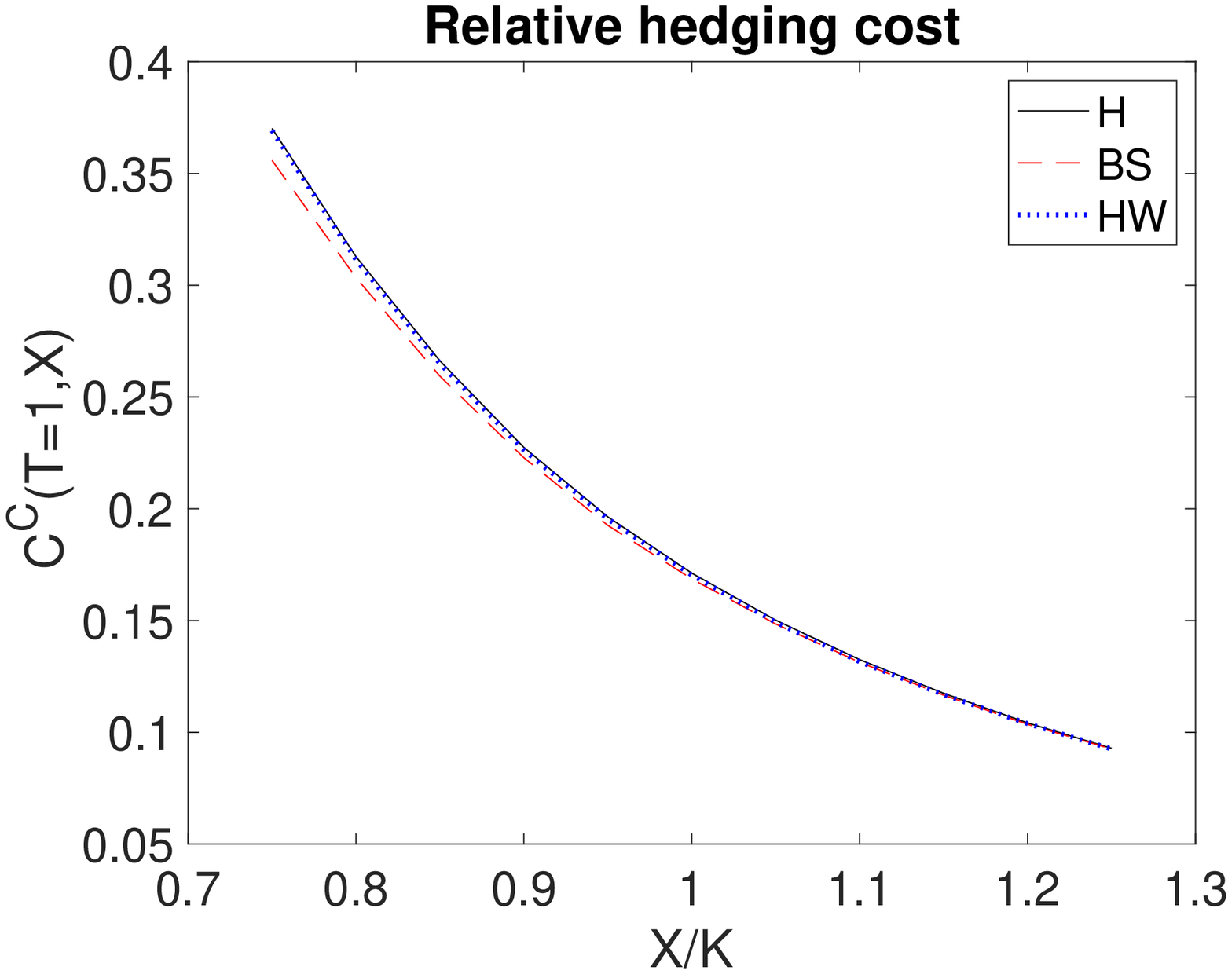}
\end{tabular}
\end{center}
\caption{      \label{fig_hedgecostroughslow}
{ The figure  shows the relative hedging cost uncertainty defined in Eq.~(\ref{eq:Hc}). The solid, dashed and dotted lines correspond to
the (H), (BS) and (HW) schemes respectively.   We can still observe that the (BS) scheme is optimal. 
The figure corresponds to a slowly  mean reverting rough  volatility factor:  $H=.1$ and  $\eps=1$.
The hedging parameter is optimized  as in Figure \ref{fig_hedgecost}. 
 } }
\end{figure}

%Finally,    there is  a small bias in the hedging cost relative to its Black-Scholes value evaluated in the historical 
%volatility. We define the relative stochastic  volatility hedging bias under the historical measure by:
%\ban
%    B(x_0,T)    &=&   \frac{ \EE[ E_T^{\rm C}]  - Q^{(0)}(T,x_0;\bar\sigma)  }{ Q^{(0)}(T,x_0;\bar\sigma)  }    =  
%       \frac{ \EE[  ( X_T - K)^+    - \int_0^T a^{\rm C}_s dX_s ]  -Q^{(0)}(T,x_0;\bar\sigma)  }
%   { Q^{(0)}(T,x_0;\bar\sigma)  }     \\ &=& 
%     \frac{ \EE[  ( X_T - K)^+  ] -Q^{(0)}(T,x_0;\bar\sigma)  }   { Q^{(0)}(T,x_0;\bar\sigma)  }  ,
%\ean 
% where we made use of the  fact that $\mu$ in Eq. (\ref{eq:Xdef}) is assumed  zero.  
%In the regime $\eps \to  0$ we  then have
%\ban
%     B(x_0,T)    &= &   \sqrt{\eps} \rho \frac{  Q^{(1)}(T,x_0;\bar\sigma)  }   { Q^{(0)}(T,x_0;\bar\sigma)  }   +o(\sqrt{\eps}).
%\ean
  
\section{Conclusions}
\label{sec:concl}

Classic price replicating delta hedging strategies are important in hedging practice.
We present here a novel analysis of the extra hedging cost associated 
with such schemes that follows from a stochastic volatility situation and
thus an incomplete market context.   
 We model the volatility as a stationary stochastic process that is rapidly 
 mean-reverting relative to the diffusion time of the underlying.
 Specifically, the volatility  is a smooth function of a Volterra type Gaussian process
(an integral  of a standard Brownian motion with respect to a deterministic integral kernel).  
 We incorporate leverage in our modeling so that the Brownian motion driving
 the volatility is correlated with the Brownian motion driving the underlying.  

In this context we identify the correction to the price that is
produced by the stochastic volatility. The two market parameters that
determine this correction are the {\it effective volatility} or root mean square 
volatility and a  {\it market pricing parameter}. The hedging cost incurred due to the 
stochastic nature of the volatility  is characterized by a  Vega risk  martingale.
The amplitude of this  martingale is proportional to a {\it market risk parameter}  that needs to be 
calibrated to the market in order to quantify the hedging cost (mean and variance). 
This market risk parameter cannot be identified from
the implied volatility skew. % to the order of the approximation that we consider here. 

   We consider specifically hedging of a European call option and then we 
   get explicit expressions for the hedging cost.  
We consider a large class of hedging schemes
 that we call dynamic asset (DA) based hedging schemes which 
are based on replicating portfolios  made of some number of underlyings 
  and some amount in the bank account,
 so that the class in particular contains all delta hedging strategies.
 We find that  in this class the optimal scheme 
    is the (BS) scheme, where the delta is the Black-Scholes delta when evaluated at the 
    implied volatility, the so-called   ``practitioners delta''. 
All the hedging schemes  that we consider can be implemented
      without knowledge of the market  risk parameter, only the quantitative evaluation of the hedging cost requires 
      the knowledge of the market  risk parameter.
           In the case of no leverage, the market pricing parameter referred to above 
      is zero, all schemes coincide, and the hedging cost is determined by 
      the  Vega risk martingale.  For general leverage and 
       for each choice of delta  we identify
      the  {\it hedging risk surface} which  characterizes
       the variance of the cost. 
        { 
Monte Carlo simulations make it possible to assess the performances of the hedging
schemes, in particular the optimal (BS) scheme. 
They reveal that the performance gain obtained when using the (BS) scheme 
is larger for rough volatility factors than with classic Markovian  volatility factors.   
A second  observation that follows from this study is that
the (BS) scheme is robust with respect to the 
assumption of rapid mean reversion. It is robust  in the sense that it performs
as good as the delta of the Black-Scholes price at the historical volatility or other (DA) strategies
when the mean reversion time is of the same order as the time to maturity.  
 }

%       It is parameterized by the exercise  time relative to maturity
%       and the $d_{_{-}}$ of the classic Black-Scholes theory at the effective 
%       volatility in the case of a European call option.     

Note that we have assumed a smooth and bounded payoff in the proofs of our results,
although the formulas can be applied with a more general payoff.
The proofs for nonsmooth payoff functions are more involved than the corresponding ones dedicated to pricing
as presented in  \cite{sv1},  they should involve a payoff regularization
scheme and they will be presented elsewhere.
 
Finally, we remark that we have considered a simplified market situation.
In order to capture a more general market context other effects, like
transaction cost, discreteness, market price of volatility risk and non-zero interest rate  and
price drift need to
be taken into account. 
 {Here, we wanted to characterize in a rigorous way 
the effect  of market  incompleteness in the simple albeit practically important
context of delta hedging schemes  leaving for future work 
more sophisticated hedging schemes incorporating in particular other derivatives  \cite{EE18}.    }

%\clearpage  
% 
%\newpage

\section*{Acknowledgements}
  This research  has been  supported in part by 
 %AFOSR grant  FA9550-18-1-0217, NSF  grant 1616954, 
 Centre Cournot, Fondation Cournot, and 
 Universit\'e Paris Saclay (chaire d'Alembert).

\appendix
 
\section{Effective Market Lemmas}
\label{app:EM}

We denote
\begin{equation}
\label{def:Gb}
G(z) = \frac{1}{2} \big( F(z)^2 - \overline{\sigma}^2\big) .
\end{equation}

The random term $\phi^\eps_{t}$ defined by (\ref{def:phit}) has the form 
\begin{equation}
\phi_{t}^\eps = 
\EE \Big[    \int_t^T G (Z_s^\eps)  ds \mid {\cal F}_t \Big] .
\end{equation}

The martingale $\psi^\eps_t$ defined by (\ref{def:Kt}) has the form
\begin{equation}
\psi_t^\eps = 
\EE \Big[   \int_0^T G(Z_s^\eps)  ds \mid {\cal F}_t\Big] .
\end{equation}

\begin{lemma}
\label{lem:0}%
For any smooth function $f$ with bounded derivative, 
we have
\begin{equation}
{\rm Var} \big( \EE \big[ f(Z_{t}^\eps) |{\cal F}_0 \big]\big)
\leq  \|f'\|_\infty^2  (\sigma_{t,\infty}^\eps)^2  ,
\label{eq:bornvarcond1}
\end{equation}
where we have defined for any $0\leq t\leq s\leq \infty$:
\begin{equation}
\label{def:sigmaepsst}
(\sigma_{t,s}^\eps)^2=\sigma_{\rm z}^2
   \int_t^{s} {\cal K}^\eps(u)^2 du .
\end{equation}
\end{lemma}
\begin{proof}
The conditional distribution of $Z_t^\eps$ given ${\cal F}_0$ is Gaussian with mean
\begin{equation}
\label{eq:meancondapp1}
\EE \big[  Z_t^\eps |{\cal F}_0 \big] = \sigma_{{\rm z}}  \int_{-\infty}^0 {\cal K}^\eps(t-u) dW_u
\end{equation}
and variance
\ba\label{eq:vvar}
 {\rm Var} \big( Z_t^\eps \mid{\cal F}_0\big) = (\sigma_{0,t}^\eps)^2 =\sigma_{{\rm z}}^2 \int_0^{t} {\cal K}^\eps(u)^2 du  .
\ea
Therefore
$$
{\rm Var} \big( \EE \big[ f(Z_{t}^\eps)  \mid{\cal F}_0 \big]\big)
=
{\rm Var} \Big( \int_\RR f\big(\EE \big[ Z_{t}^\eps  \mid{\cal F}_0 \big]  +\sigma_{0,t}^\eps z \big) p(z) dz \Big)  ,
$$
where $p(z)$ is the pdf of the standard normal distribution.
By (\ref{eq:meancondapp1}) the random variable $\EE \big[ Z_{t}^\eps  \mid{\cal F}_0 \big] $ is Gaussian with mean zero and variance
$(\sigma_{t,\infty}^\eps)^2$
so that
\begin{eqnarray*}
\nonumber
{\rm Var} \big( \EE \big[ f(Z_{t}^\eps)  \mid{\cal F}_0 \big]\big)
&=&
 \frac{1}{2} \int_\RR \int_\RR dz dz' p(z) p(z') \int_\RR \int_\RR du du' p(u) p(u') \\
\nonumber &&\times 
\Big[   f\big(\sigma_{t,\infty}^\eps u +\sigma_{0,t}^\eps z \big) - 
f\big(\sigma_{t,\infty}^\eps u' +\sigma_{0,t}^\eps z \big)\Big] \\
\nonumber&& \times
\Big[   f\big(\sigma_{t,\infty}^\eps u +\sigma_{0,t}^\eps z' \big) - 
f\big(\sigma_{t,\infty}^\eps u' +\sigma_{0,t}^\eps z' \big)\Big] \\
\nonumber&\leq & \|f'\|_\infty^2 (\sigma_{t,\infty}^\eps)^2
\frac{1}{2} \int_\RR \int_\RR du du' p(u) p(u')(u-u')^2 \\
& =  & \|f'\|_\infty^2  (\sigma_{t,\infty}^\eps)^2  ,
\end{eqnarray*}
which is the desired result.
\end{proof}

\begin{lemma}
\label{lem:A2}%
For any $t \leq T$, $\phi_{t}^\eps$ is a zero-mean random variable with standard deviation of order $\eps^{(d-\frac{1}{2})\wedge 1 }$:
\begin{equation}
\label{eq:stdev}
\sup_{\eps \in (0,1]}
\sup_{t \in [0,T]}\eps^{(2d-1)\wedge 2} 
\EE [ (\phi_{t}^\eps)^2]  < \infty ,
\end{equation}
 {where $d$ is defined in (\ref{eq:defd}).}
\end{lemma}

\begin{proof}
For $t\in [0,T]$  the second moment of $\phi_{t}^\eps$ is:
\begin{eqnarray*}
\EE \big[ (\phi_{t}^\eps )^2\big] &=&  
\EE\Big[  
\EE \Big[    \int_t^T G (Z_s^\eps)  ds \mid {\cal F}_t \Big]^2
\Big]\\
&=&
\int_0^{T-t} ds \int_{0}^{T-t} ds' 
{\rm Cov}\big( \EE \big[ G(Z_s^\eps)  \mid {\cal F}_0\big] ,
\EE \big[ G(Z_{s'}^\eps)  \mid {\cal F}_0\big] \big) .
\end{eqnarray*}
We have by Lemma \ref{lem:0}
\begin{eqnarray*}
\EE \big[ (\phi_{t}^\eps )^2\big]  &\leq &  \Big( \int_0^{T-t} ds 
\big(  {\rm Var}\big(\EE \big[ G(Z_s^\eps)  \mid {\cal F}_0\big]\big) \big)^{1/2}
\Big)^2\leq  \|G'\|_\infty^2\Big(  \int_0^{T-t} ds  \sigma_{s,\infty}^\eps  \Big)^2 .
\end{eqnarray*}
In view of Lemma \ref{lem:7} we then have
$$
\EE \big[ (\phi_{t}^\eps )^2\big]\leq C_{T} 
\big( \eps+\eps^{d-\frac{1}{2}}\big)^2 \leq 4 C_{T} \eps^{(2d-1)\wedge 2 } ,
$$
uniformly in $t  \leq T$ and $\eps \in (0,1]$ for some constant $C_{T}$.
\end{proof}

\begin{lemma}
\label{lem:1}%
Let $Y_t$  be a bounded adapted process,  we have  
\begin{equation}
\lim_{\eps \to 0}  \eps^{-1/2} \sup_{t\in [0,T]}   \EE\left[ \left| 
\int_0^t Y_s \phi_s^\eps dW_s^* 
 \right|^2 \mid {\cal F}_0  \right]^{1/2}  = 0  .
\end{equation} 
\end{lemma}
\begin{proof}
%Follows from Lemma \ref{} and the It\^o isometry.
We have by the It\^o isometry
\ban
 \EE\left[ \left| 
\int_0^t Y_s  \phi_s^\eps dW_s^* 
 \right|^2 \mid {\cal F}_0  \right]
 =
 \EE\left[  
\int_0^t \left| 
Y_s  \phi_s^\eps
\right|^2 ds \mid {\cal F}_0 \right] ,
\ean
and the result then follows from Lemma \ref{lem:A2} noting that we 
consider the case $d > 1$.  
\end{proof} 

We next present a result regarding the quadratic variation of
$\psi^\eps$. 

\begin{lemma}
\label{lem:1b}%
$(\psi_t^\eps)_{t\in [0,T]}$ is a square-integrable martingale and
\ba
\label{def:varthetaepsb}
d \left< \psi^\eps, W\right>_t  =  \vartheta^\eps_{t} dt ,
\quad \quad  
d\left< \psi^\eps,\psi^\eps  \right>_t =  (\vartheta^\eps_t)^2 dt   ,
\ea
with
\ba
\label{express:thetaeps1}
\vartheta^\eps_{t} = \sigma_{{\rm z}} \int_t^T 
\EE \big[ G'(Z_s^\eps) \mid{\cal F}_t \big]{\cal K}^\eps(s-t) ds   
.
\ea

\end{lemma}
  An alternative expression of  $\vartheta^\eps_{t}$ is given in (\ref{eq:crocKW2}).
\begin{proof}
This follows from \cite[Lemma B.1]{sv2} and its proof. 
For $t < s$, the conditional distribution of $Z_s^\eps$ given ${\cal F}_t$ is Gaussian with mean
$$
\EE \big[  Z_s^\eps  \mid{\cal F}_t \big] =\sigma_{{\rm z}}
 \int_{-\infty}^t {\cal K}^\eps(s-u) dW_u 
$$
and deterministic variance given by
$$
 {\rm Var} \big( Z_s^\eps  \mid{\cal F}_t\big) = (\sigma_{0,s-t}^\eps)^2,
$$
where $\sigma_{s,t}^\eps$ is defined by (\ref{def:sigmaepsst}).  
%We thus  have  that  the distribution of
%$$
%\frac{1}{\sigma_{0,s-t}^\eps} \Big( \big( Z_s^\eps  - \sigma_{\rm z} \int_{-\infty}^t {\cal K}^\eps(s-u) dW_u  \big)   \big| {\cal F}_t 
% \Big) 
%  $$
%is standard normal.
Therefore we have
$$
\EE \big[ {G}( Z_s^\eps )  \mid{\cal F}_t \big] 
= \int_\RR G \Big( \sigma_{{\rm z}} \int_{-\infty}^t {\cal K}^\eps(s-u) dW_u
+\sigma_{0,s-t}^\eps z\Big) p(z) dz ,
$$
where $p(z)$ is the pdf of the standard normal distribution. 
As a random process in $t$ it is a continuous martingale.
By It\^o's formula, for any $t <  s$:
\begin{eqnarray*}
\EE \big[ {G}( Z_s^\eps )  \mid{\cal F}_t \big] &=&
 \int_\RR {G} \Big( \sigma_{{\rm z}}\int_{-\infty}^0 {\cal K}^\eps(s-v) dW_v
+\sigma_{0,s}^\eps z\Big) p(z) dz  \\
&&
+ \int_0^t  \int_\RR {G}' \Big( \sigma_{{\rm z}}\int_{-\infty}^u {\cal K}^\eps(s-v) dW_v
+\sigma_{0,s-u}^\eps z\Big) z p(z) dz \partial_u \sigma_{0,s-u}^\eps du   \\
&&
+\sigma_{{\rm z}} \int_0^t  \int_\RR {G}' \Big( \sigma_{{\rm z}}\int_{-\infty}^u {\cal K}^\eps(s-v) dW_v
+\sigma_{0,s-u}^\eps z\Big)p(z) dz  {\cal K}^\eps(s-u) dW_u   \\
&&
+\frac{\sigma_{{\rm z}}^2}{2}  \int_0^t  \int_\RR {G}'' \Big( \sigma_{{\rm z}}\int_{-\infty}^u {\cal K}^\eps(s-v) dW_v
+\sigma_{0,s-u}^\eps z\Big) p(z) dz {\cal K}^\eps(s-u)^2 du   .
\end{eqnarray*}
Note that we have from Eq. (\ref{def:sigmaepsst}) that 
\ban
  2   \sigma_{0,s-u}^\eps  \partial_u \sigma_{0,s-u}^\eps   =  -  \partial_s (\sigma_{0,s-u}^\eps )^2=
 -  \sigma_{{\rm z}}^2  {\cal K}^\eps(s-u)^2  .
\ean 
The martingale representation then follows explicitly via integration by parts (with respect to $z$, using $zp(z) = -\partial_z p(z)$):
\begin{eqnarray*}
\EE \big[ {G}( Z_s^\eps )  \mid{\cal F}_t \big] &=&
 \int_\RR {G} \Big( \sigma_{{\rm z}}\int_{-\infty}^0 {\cal K}^\eps(s-v) dW_v
+\sigma_{0,s}^\eps z\Big) p(z) dz  \\
 &&
+\sigma_{{\rm z}} \int_0^t  \int_\RR {G}' \Big( \sigma_{{\rm z}}\int_{-\infty}^u {\cal K}^\eps(s-v) dW_v
+\sigma_{0,s-u}^\eps z\Big)p(z) dz  {\cal K}^\eps(s-u) dW_u     .
\end{eqnarray*}
 We also have
 \begin{eqnarray*}
{G}( Z_s^\eps )  &=&{G} \Big( \sigma_{{\rm z}}\int_{-\infty}^s {\cal K}^\eps(s-v) dW_v\Big)  \\
&=&  
 \int_\RR {G} \Big(\sigma_{{\rm z}} \int_{-\infty}^u {\cal K}^\eps(s-v) dW_v
+\sigma_{0,s-u}^\eps z\Big) p(z) dz  \mid_{u=s} \\
% &=&  
% \int_\RR {G} \Big(\sigma_{{\rm z}} \int_{-\infty}^s {\cal K}^\eps(s-v) dW_v
%+\sigma_{0,0}^\eps z\Big) p(z) dz  \\
%&=&  \int_\RR {G} \Big(\sigma_{{\rm z}} \int_{-\infty}^0 {\cal K}^\eps(s-v) dW_v
%+\sigma_{0,s}^\eps z\Big) p(z) dz  \\
&=&  \int_\RR {G} \Big(\sigma_{{\rm z}} \int_{-\infty}^0 {\cal K}^\eps(s-v) dW_v
+\sigma_{0,s}^\eps z\Big) p(z) dz  \\
&&
+ \int_0^s \int_\RR {G}' \Big( \sigma_{{\rm z}}\int_{-\infty}^u {\cal K}^\eps(s-v) dW_v
+\sigma_{0,s-u}^\eps z\Big) z p(z) dz \partial_u \sigma_{0,s-u}^\eps du   \\
&&
+\sigma_{{\rm z}}\int_0^s  \int_\RR {G}' \Big( \sigma_{{\rm z}}\int_{-\infty}^u {\cal K}^\eps(s-v) dW_v
+\sigma_{0,s-u}^\eps z\Big)p(z) dz  {\cal K}^\eps(s-u) dW_u   \\
&&
+\frac{\sigma_{{\rm z}}^2}{2}  \int_0^s  \int_\RR {G}'' \Big(\sigma_{{\rm z}} \int_{-\infty}^u {\cal K}^\eps(s-v) dW_v
+\sigma_{0,s-u}^\eps z\Big) p(z) dz {\cal K}^\eps(s-u)^2 du  \\
&=&  \int_\RR {G} \Big(\sigma_{{\rm z}} \int_{-\infty}^0 {\cal K}^\eps(s-v) dW_v
+\sigma_{0,s}^\eps z\Big) p(z) dz  \\
 &&
+\sigma_{{\rm z}}\int_0^s  \int_\RR {G}' \Big( \sigma_{{\rm z}}\int_{-\infty}^u {\cal K}^\eps(s-v) dW_v
+\sigma_{0,s-u}^\eps z\Big)p(z) dz  {\cal K}^\eps(s-u) dW_u    
   .
\end{eqnarray*}
Therefore
\begin{eqnarray*}
\psi^\eps_t &=& \int_0^t  {G}( Z_s^\eps) ds + \int_t^T \EE \big[ {G}( Z_s^\eps )  \mid{\cal F}_t \big]  ds \\
&=& 
\Big[ \int_\RR \int_0^T {G} \Big( \sigma_{{\rm z}} \int_{-\infty}^0 {\cal K}^\eps(s-v) dW_v
+\sigma_{0,s}^\eps z\Big)  ds p(z) dz\Big] \\
%&& +\int_0^t \Big[ \int_u^T  \int_\RR {G}' \Big(\sigma_{{\rm z}} \int_{-\infty}^u {\cal K}^\eps(s-v) dW_v
%+\sigma_{0,s-u}^\eps z\Big) z p(z) dz \partial_u \sigma_{0,s-u}^\eps ds \Big] du \\
&&+
\sigma_{{\rm z}}\int_0^t \Big[ \int_u^T \int_\RR {G}'\Big( \sigma_{{\rm z}}\int_{-\infty}^u {\cal K}^\eps(s-v) dW_v
+\sigma_{0,s-u}^\eps z\Big)p(z) dz  {\cal K}^\eps(s-u) ds \Big]dW_u 
%\\
%&& +\frac{\sigma_{{\rm z}}^2}{2} \int_0^t \Big[ \int_u^T \int_\RR {G}'' \Big(\sigma_{{\rm z}} \int_{-\infty}^u {\cal K}^\eps(s-v) dW_v
%+\sigma_{0,s-u}^\eps z\Big) p(z) dz {\cal K}^\eps(s-u)^2 ds \Big] du
 .
\end{eqnarray*}
This gives (\ref{def:varthetaepsb})  with 
\begin{eqnarray}
\vartheta^\eps_t  &=& \sigma_{{\rm z}}
 \int_t^T \int_\RR {G}'\Big( \sigma_{{\rm z}}\int_{-\infty}^t {\cal K}^\eps(s-v) dW_v
+\sigma_{0,s-t}^\eps z\Big)p(z) dz  {\cal K}^\eps(s-t) ds  ,
\label{eq:crocKW2}
\end{eqnarray}
which can also be written 
as stated in the Lemma.
\end{proof}

\begin{lemma}
\label{lem:K1}%
Let $Y_t$  be a bounded adapted process. Then we have  
\ban
\sup_{\eps \in (0,1]}  \eps^{-1/2} \sup_{t\in [0,T]}   \EE\left[ \Big|
\int_0^t Y_s   d\psi^\eps_s 
 \Big|^2 \mid {\cal F}_0  \right]^{1/2}   &<&    \infty 
  . \ean 
\end{lemma}
\begin{proof}
There exists $\tilde{K} < \infty$ such that, for $t \in (0,T)$,
\ban
 \EE\left[ \left| \int_0^t Y_s   d\psi^\eps_s  \right|^2  \mid {\cal F}_0 \right] 
      \leq \tilde{K} \EE \big[ \left< \psi^\eps,\psi^\eps  \right>_T-\left< \psi^\eps,\psi^\eps  \right>_0  \mid {\cal F}_0 \big]  ,
\ean
 and the result follows from (\ref{def:varthetaepsb}) and (\ref{eq:estimvarthetaeps}).
   \end{proof}
 
\begin{lemma}
\label{lem:K2} 
Let $f(t,x)$ be smooth bounded and with bounded derivatives and let $X_t$  
be defined by Eq. (\ref{eq:Xdef}). 
Then for any $t \in [0,T]$ we have 
\ba
\lim_{\eps \to 0}   \EE\left[ \Big( \int_0^t f(s,X_s)  \big( \eps^{-1/2}   \sigma^\eps_s \vartheta^\eps_s  -   \overline{D}\big)
  ds \Big)^2 \right]     = 0 ,  \label{a31}     \\
  \lim_{\eps \to 0}   \EE\left[ \Big( \int_0^t f(s,X_s)  \left( \eps^{-1}    
  \big( \vartheta^\eps_s \right)^2   -   \overline{\Gamma}^2\big)
  ds \Big)^2 \right]     = 0  .  \label{a32} 
\ea
\end{lemma} 
 \begin{proof}
 The result in Eq. (\ref{a31})  follows via an argument as in the proof of Eq. (\ref{eq:estimeRj}) for $j=3$
 as given in \cite{sv3} (note that, by (\ref{eq:estimvarthetaeps}), $\eps^{-1/2}\vartheta^\eps_t$ is uniformly bounded almost surely).  
The result in Eq.  (\ref{a32}) follows via an argument as in the proof of Eq. (\ref{eq:estimeRj}) for $j=2$
 as given in \cite{sv3}. To complete that proof it remains to show that   
   \ban
 \lim_{\eps \to 0} \sup_{t \in [0,T]}  \EE\left[   \left( \kappa_t^\eps \right)^2    \right]   = 0 ,
 \quad \hbox{for} \quad  \kappa_t^\eps  = 
     \int_0^t \left( 
       \eps^{-1}  \left(\vartheta^\eps_s\right)^2 - \overline{\Gamma} \right) ds .
 \ean
 We show this in  Lemma \ref{lem:kest}. 
 \end{proof}

 \begin{lemma}
\label{lem:kest} 
  Let
  \ban
     \kappa_t^\eps  = 
     \int_0^t \left( 
       \eps^{-1}  \left(\vartheta^\eps_s\right)^2 - \overline{\Gamma} \right) ds    ,
  \ean
 then 
 \ban
 \lim_{\eps \to 0} \sup_{t \in [0,T]}  \EE\left[   \left( \kappa_t^\eps \right)^2    \right]   = 0 .
 \ean
 \end{lemma}
  \begin{proof} 
  As $(a+b)^2\leq 2a^2+2b^2$ we have
  \ban
  \EE\left[   \left( \kappa_t^\eps \right)^2    \right]  \leq
  2 \eps^{-2}  \int_0^tds \int_0^tds' {\rm Cov}\left( \left( \vartheta^\eps_s \right)^2,  \left(\vartheta^\eps_{s'}\right)^2 \right) 
  +2 \left( \int_0^t \left( \eps^{-1} \EE\left[  \vartheta^\eps_s \right]^2 
  - \overline{\Gamma}  \right) ds \right)^2 . 
  \ean
  The results then follows from Lemmas 
  \ref{lem:kest1} and \ref{lem:kest2} and the bound
  in Eq. (\ref{eq:estimvarthetaeps}) using dominated convergence theorem.
  \end{proof} 
  
 \begin{lemma}
 \label{lem:kest1}
 Let $\vartheta^\eps_t$ be defined by (\ref{express:thetaeps1}).
  We have  for any $t \in [0,T)$:
$$
\lim_{\eps \to 0}  \eps^{-1} \EE\left[ \left(\vartheta^\eps_t\right)^2 \right] =  \overline{\Gamma}^2 ,
$$
where $\overline{\Gamma}$ is defined by (\ref{def:barGamma}).   
 \end{lemma}
 \begin{proof}
We 
consider
 \begin{eqnarray*}
\EE\left[ \left(\vartheta^\eps_t\right)^2 \right]  \sigma_{{\rm z}}^{-2} &=&  
\EE\Big[  
\EE \Big[    \int_t^T G' (Z_s^\eps)  {\cal K}^\eps(s-t) ds  \mid {\cal F}_t \Big] 
\EE \Big[    \int_{t}^{T} G' (Z_s^\eps)  {\cal K}^\eps(s-t) ds  \mid {\cal F}_{t} \Big] 
\Big]\\
 &=& 2
\int_0^{T-t} ds \int_{s}^{T-t} ds' 
\EE\left[ 
\EE \big[ G'(Z_s^\eps)   \mid{\cal F}_0\big] 
\EE \big[ G'(Z_{s'}^\eps)  \mid{\cal F}_0\big] 
\right] {\cal K}^\eps(s) {\cal K}^\eps(s')
.
\end{eqnarray*}
We can then write
\begin{eqnarray*}
 \EE\left[ \left(\vartheta^\eps_t\right)^2 \right] \sigma_{{\rm z}}^{-2}  &=&  
2
\int_0^{ T-t } ds \int_{s}^{T-t} ds'  
\\ &&   \hbox{} \times
\EE\left[ 
\big[  \int_\RR G'(  \sigma_{\rm z} \int_{-\infty}^0 {\cal K}^\eps(s-v) 
dW_v+\sigma_{0,s}^\eps z )  
p(z) dz \big]  
\right.
\\ &&   \hbox{} \times
\left.   \big[  \int_\RR G'(\sigma_{\rm z} \int_{-\infty}^0 {\cal K}^\eps(s'-v) 
dW_v+\sigma_{0,s'}^\eps z')  p(z') dz' \big] 
\right]   {\cal K}^\eps(s)  {\cal K}^\eps(s')   \\   
&=&
2
\int_0^{ T-t } ds \int_{s}^{T-t} ds'  
 \int_{\RR^2} du  du'
  \big[  \int_\RR G'(  \sigma_{s,\infty}^\eps u +\sigma_{0,s}^\eps z )  
p(z) dz \big]  
\\ &&   \hbox{} \times
 \big[  \int_\RR G'(  \sigma_{s',\infty}^\eps u' +\sigma_{0,s'}^\eps z' )p(z') dz' \big] 
p_{\tilde{\cal C}_{\cal K}^\eps(s,s')}(u,u')   {\cal K}^\eps(s)  {\cal K}^\eps(s')  ,
 \end{eqnarray*}
where $p(z)$ is the pdf of the standard normal distribution,
$p_C$ is the pdf of the (standardized) bivariate normal distribution with 
mean zero and covariance matrix as in Lemma \ref{prop:main}, and
 \ban
\tilde{\cal C}_{\cal K}^\eps(s,s') = \frac{\sigma^2_{\rm z}  \int_{-\infty}^0 {\cal K}^\eps(s'-v) {\cal K}^\eps(s-v)  dv }
     {\sigma_{s,\infty}^\eps \sigma_{s',\infty}^\eps}  .
\ean
%It then follows from  ${\cal K} \in L^1(0,\infty)$ that 
%\ban
%&&
%\lim_{\eps \to 0}  \eps^{-1} \EE\left[ \left(\vartheta^\eps_t\right)^2 \right] \sigma_{{\rm z}}^{-2}  = \overline{\Gamma}^2 ,\\
%&&  \overline{\Gamma}^2  =2   
%\int_0^{\infty} ds \int_{s}^{\infty} ds' 
% \int_\RR du \int_\RR du' 
% \left[ 
%  \big[  \int_\RR FF'( \sigma_{s,\infty} u +\sigma_{0,s} z ) {\cal K}(s) 
%p(z) dz \big]  
%\right.
%\\ &&   \quad \quad  \times
%\left.  \big[  \int_\RR FF'(\sigma_{s',\infty} u' +\sigma_{0,s'} z') {\cal K}(s') p(z') dz' \big] 
%p_{\tilde{\cal C}_{\cal K}(s,s')}(u,u') 
%\right]      . 
%\ean
By remarking that $(\sigma_{s,\infty}^\eps)^2+(\sigma_{0,s}^\eps)^2 =\sigma_{\rm z}^2 $ and $
\sigma_{s,\infty}^\eps \sigma_{s',\infty}^\eps \tilde{\cal C}_{\cal K}^\eps(s,s') = \sigma_{\rm z}^2 {\cal C}_{\cal K}^\eps(s,s')$,
with ${\cal C}_{\cal K}^\eps(s,s') = {\cal C}_{\cal K}(s/\eps,s'/\eps)$,
we can see that, if $(Z,Z',U,U')$ is a four-dimensional Gaussian vector with pdf $p(z)p(z')p_{\tilde{\cal C}_{\cal K}^\eps(s,s')}(u,u')$,
then $(\sigma_{s,\infty}^\eps U +\sigma_{0,s}^\eps Z, \sigma_{s',\infty}^\eps U' +\sigma_{0,s'}^\eps Z') =
( \sigma_{\rm z}  Y,\sigma_{\rm z}  Y')$ where $(Y,Y')$ 
is a two-dimensional Gaussian vector with pdf $p_{{\cal C}_{\cal K}^\eps(s,s')}(y,y')$.
This gives
\begin{eqnarray*}
 \EE\left[ \left(\vartheta^\eps_t\right)^2 \right]  &=&  
2\sigma_{{\rm z}}^2  
\int_0^{ T-t } ds \int_{s}^{T-t} ds'  
 \int_{\RR^2} dy dy'
 G'(  \sigma_{\rm z} y )  
 G'(  \sigma_{\rm z} y' )  
p_{ {\cal C}_{\cal K}^\eps(s,s')}(y,y')  {\cal K}^\eps(s)  {\cal K}^\eps(s') 
  ,
 \end{eqnarray*}
or
\begin{eqnarray*}
 \EE\left[ \left(\vartheta^\eps_t\right)^2 \right]  &=&  
2 \eps \sigma_{{\rm z}}^2  
\int_0^{ \frac{T-t}{\eps} } ds \int_{s}^{\frac{T-t}{\eps} } ds'  
 \int_{\RR^2} dy dy'
 G'(  \sigma_{\rm z} y ) 
 G'(  \sigma_{\rm z} y' )  
p_{ {\cal C}_{\cal K} (s,s')}(y,y') {\cal K} (s) {\cal K} (s') 
.
 \end{eqnarray*}
By using the fact that ${\cal K}\in L^1(0,\infty)$ we finally
get
$$
\lim_{\eps \to 0} \eps^{-1}  \EE\left[ \left(\vartheta^\eps_t\right)^2 \right]  = \overline{\Gamma}^2 ,
$$
with the expression (\ref{def:barGamma}) of $\overline{\Gamma}$, which completes the proof of the Lemma.
\end{proof} 
  
\begin{lemma}
 \label{lem:kest2}
 For any $0 \leq t < t' < T$ we have
 \ban
   \lim_{\eps \to 0} \eps^{-2} \left| {\rm Cov}\left( \left(\vartheta^\eps_t\right)^2,  
   \left(\vartheta^\eps_{t'}\right)^2\right) \right|= 0.
 \ean
 \end{lemma}
 \begin{proof}
 Let us consider $0 \leq t' <  t\leq T$. We have
\begin{eqnarray*}
 && \EE\big[ \left(\vartheta^\eps_{t} \right)^2 \left(\vartheta^\eps_{t'}\right)^2   \big] 
= \sigma_{{\rm z}}^4 
\int_t^{T} ds   {\cal K}^\eps(s-t) \int_t^{T} ds'   {\cal K}^\eps(s'-t) % \\ &&   \hbox{} \times
 \int_{t'}^T du  {\cal K}^\eps(u -t') \\ && \times  \int_{t'}^T du'  {\cal K}^\eps(u' -t')    \EE\Big[ 
\EE\big[  G'(Z_s^\eps)|{\cal F}_t \big]   \EE\big[  G'(Z_{s'}^\eps)|{\cal F}_t \big] 
% \\ && \hbox{} \times 
 \EE\big[   G'(Z_{u}^\eps)|{\cal F}_{t'} \big]   \EE\big[   G'(Z_{u'}^\eps)|{\cal F}_{t'} \big]   \Big]  ,
\end{eqnarray*}
so we can write
\begin{eqnarray*}
&& {\rm Cov} \big( (\vartheta^\eps_{t})^2 ,  (\vartheta^\eps_{t'})^2  \big)
 =  
 \sigma_{{\rm z}}^4  \int_t^{T} ds   {\cal K}^\eps(s-t)  \int_{t}^T ds'  {\cal K}^\eps(s' - t)  \\
&&\times
\int_{t'}^{T} du   {\cal K}^\eps(u-t')  \int_{t'}^T du'  {\cal K}^\eps(u' - t')
\bigg\{
\EE\Big[ 
\EE\big[  G'(Z_s^\eps)|{\cal F}_t \big]  \EE\big[  G'(Z_{s'}^\eps)|{\cal F}_{t} \big]
\\ && \hbox{} \times 
\EE\big[  G'(Z_u^\eps)|{\cal F}_{t'} \big]  \EE\big[  G'(Z_{u'}^\eps)|{\cal F}_{t'} \big] 
 \Big]    - 
\EE\Big[ 
\EE\big[ G'(Z_s^\eps)|{\cal F}_t \big]  \EE\big[   G'(Z_{s'}^\eps) |{\cal F}_t   \big]  \Big]
\\ && \hbox{} \times \EE\Big[
 \EE\big[ G'(Z_u^\eps)|{\cal F}_{t'} \big]  \EE\big[   G'(Z_{u'}^\eps) |{\cal F}_{t'}   \big] 
\Big]  \bigg\}
\\ && \hspace*{1.05in} =
 \sigma_{{\rm z}}^4  \int_t^{T} ds   {\cal K}^\eps(s-t)  \int_{t}^T ds'  {\cal K}^\eps(s' - t)  \\
&&\times
\int_{t'}^{T} du   {\cal K}^\eps(u-t')  \int_{t'}^T du'  {\cal K}^\eps(u' - t')
\bigg\{ \EE\Big[
 \EE\big[ 
\EE\big[  G'(Z_s^\eps)|{\cal F}_t \big]  \EE\big[  G'(Z_{s'}^\eps)|{\cal F}_{t} \big]
|{\cal F}_{t'} \big]
\\ && \hbox{} \times 
\EE\big[  G'(Z_u^\eps)|{\cal F}_{t'} \big]  \EE\big[  G'(Z_{u'}^\eps)|{\cal F}_{t'} \big] 
 \Big]    - 
\EE\Big[  \EE\big[
\EE\big[ G'(Z_s^\eps)|{\cal F}_t \big]  \EE\big[   G'(Z_{s'}^\eps) |{\cal F}_t   \big]  \big]
\\ && \hbox{} \times 
 \EE\big[ G'(Z_u^\eps)|{\cal F}_{t'} \big]  \EE\big[   G'(Z_{u'}^\eps) |{\cal F}_{t'}   \big] 
\Big]  \bigg\} ,
\end{eqnarray*}
and therefore
\begin{eqnarray*}
&& \big| {\rm Cov} \big( (\vartheta^\eps_{t})^2  , (\vartheta^\eps_{t'})^2  \big) \big|
\leq   \sigma_{{\rm z}}^4    \|G'\|_\infty^2  \int_t^{T} ds   |{\cal K}^\eps(s-t)|  
 \int_t^{T} ds'   |{\cal K}^\eps(s'-t)|
 \\ && \hbox{} \times 
 \int_{t'}^T du  |{\cal K}^\eps(u-t')|
\int_{t'}^T du'  |{\cal K}^\eps(u'-t')|
\EE\Big[ \big(\EE\big[ \EE\big[   G'(Z_{s}^\eps)|{\cal F}_{t} \big]  \EE\big[   G'(Z_{s'}^\eps)|{\cal F}_{t} \big] 
|{\cal F}_{t'} \big] 
\\ && \hbox{}  \quad 
-   \EE\big[     \EE\big[   G'(Z_{s}^\eps)|{\cal F}_{t} \big]  \EE\big[   G'(Z_{s'}^\eps)|{\cal F}_{t} \big]      \big] \big)^2 \Big]^{1/2} .
\end{eqnarray*}
We can write for any $\tau>t>t'$:
\begin{eqnarray*}
&&Z^\eps_{\tau} = A^\eps_{t',\tau} +  B^\eps_{t',t,\tau} + C^\eps_{t,\tau},\quad   \quad 
 A^\eps_{t',\tau} = \sigma_{\rm z} \int_{-\infty}^{t'} {\cal K}^\eps(\tau-u)dW_u,
  \\ && 
 B^\eps_{t',t,\tau} =  \sigma_{\rm z}\int_{t'}^{t} {\cal K}^\eps(\tau-u)dW_u ,
  \quad   \quad 
 C^\eps_{t,\tau} =  \sigma_{\rm z}\int_{t}^{\tau} {\cal K}^\eps(\tau-u)dW_u ,
\end{eqnarray*}
with  $A^\eps_{t',\tau}, B^\eps_{t',t,\tau}, C^\eps_{t,\tau}$ being independent
and in particular $A^\eps_{t',\tau}$  is ${\cal F}_{t'}$ adapted.
Therefore, with $s'\geq s \geq t > t'$, we have 
\begin{eqnarray*}
&& 
\EE\Big[ \big(
\EE\big[
\EE\big[   G'(Z_{s}^\eps)|{\cal F}_{t} \big]  \EE\big[   G'(Z_{s'}^\eps)|{\cal F}_{t} \big] 
|{\cal F}_{t'} \big] 
  -   \EE\big[     \EE\big[   G'(Z_{s}^\eps)|{\cal F}_{t} \big] 
   \EE\big[   G'(Z_{s'}^\eps)|{\cal F}_{t} \big]      \big] \big)^2 \Big] \\
 &&= 
  \EE\Big[   
\EE\big[
\EE\big[   G'(Z_{s}^\eps)|{\cal F}_{t} \big]  \EE\big[   G'(Z_{s'}^\eps)|{\cal F}_{t} \big] 
|{\cal F}_{t'} \big]^2 \Big] 
  -   \EE\big[     \EE\big[   G'(Z_{s}^\eps)|{\cal F}_{t} \big] 
   \EE\big[   G'(Z_{s'}^\eps)|{\cal F}_{t} \big]      \big]^2  
   \\ &&
=
\EE\Big[
G'( A^\eps_{t',s} +  B^\eps_{t',t,s} + C^\eps_{t,s}) 
G'( A^\eps_{t',s'} +  B^\eps_{t',t,s'} + \tilde{C}^\eps_{t,s'})
G'( A^\eps_{t',s} +  \tilde{B}^\eps_{t',t,s} + \tilde{\tilde{C}}^\eps_{t,s}) 
\\ && 
\quad \hbox{} \times
G'( A^\eps_{t',s'} +  \tilde{B}^\eps_{t',t,s'} + \tilde{\tilde{\tilde{C}}}^\eps_{t,s'})
-
G'( A^\eps_{t',s} +  B^\eps_{t',t,s} + C^\eps_{t,s}) G'( A^\eps_{t',s'} +  B^\eps_{t',t,s'} + \tilde{C}^\eps_{t,s'} )
\\ && \quad \hbox{} \times 
G'( \tilde{A}^\eps_{t',s} +  \tilde{B}^\eps_{t',t,s} + \tilde{\tilde{C}}^\eps_{t,s}) 
G'( \tilde{A}^\eps_{t',s'} +  \tilde{B}^\eps_{t',t,s'} + \tilde{\tilde{\tilde{C}}}^\eps_{t,s'})
\Big] ,
\end{eqnarray*}
where each additional ``tilde'' refers to a new independent copy of
$A^\eps_{t',s} ,  B^\eps_{t',t,s}, C^\eps_{t,s}$.
 We can then write
\begin{eqnarray*}
&&  \EE\Big[ \big(
\EE\big[
\EE\big[   G'(Z_{s}^\eps)|{\cal F}_{t} \big]  \EE\big[   G'(Z_{s'}^\eps)|{\cal F}_{t} \big] 
|{\cal F}_{t'} \big]   
  -   \EE\big[     \EE\big[   G'(Z_{s}^\eps)|{\cal F}_{t} \big] 
   \EE\big[   G'(Z_{s'}^\eps)|{\cal F}_{t} \big]      \big] \big)^2 \Big]    \\
&& \leq   \|G'\|_\infty^2
\EE\Big[
\big( G'( A^\eps_{t',s} +  \tilde{B}^\eps_{t',t,s} + \tilde{\tilde{C}}^\eps_{t,s}) 
G'( A^\eps_{t',s'} +  \tilde{B}^\eps_{t',t,s'} + \tilde{\tilde{\tilde{C}}}^\eps_{t,s'})  \\ && \hbox{} 
- G'( \tilde{A}^\eps_{t',s} +  \tilde{B}^\eps_{t',t,s} + \tilde{\tilde{C}}^\eps_{t,s}) 
G'( \tilde{A}^\eps_{t',s'} +  \tilde{B}^\eps_{t',t,s'} + \tilde{\tilde{\tilde{C}}}^\eps_{t,s'})  \big)^2
\Big]^{1/2}  \\
&& \leq  2 \|G'\|_\infty^2  \|G' G'' \|_\infty    \Big( 
   \EE \big[ (A^\eps_{t',s}  - \tilde{A}^\eps_{t',s})^2 \big]^{1/2}+
\EE \big[ (A^\eps_{t',s'}  - \tilde{A}^\eps_{t',s'})^2 \big]^{1/2}\Big) \\
&& \leq 2 \sqrt{2} \|G'\|_\infty^2  \|G' G'' \|_\infty   \Big( \EE \big[ (A^\eps_{t',s}  )^2 \big]^{1/2}+
\EE \big[ (A^\eps_{t',s'} )^2 \big]^{1/2}\Big) \\ 
&& 
\leq 2 \sqrt{2}  \|G'\|_\infty^2  \|G' G'' \|_\infty  
 \Big[  \Big(  \sigma_{\rm z}^2 \int_{-\infty}^{t'} {\cal K}^\eps(s-u)^2 du \Big)^{1/2} 
+
 \Big(  \sigma_{\rm z}^2 \int_{-\infty}^{t'} {\cal K}^\eps(s'-u)^2 du \Big)^{1/2} \Big]   \\ && 
 \leq 4\sqrt{2} \|G'\|_\infty^2  \|G' G'' \|_\infty    \sigma^\eps_{t'-t,\infty} 
  \leq  {C_1}   \big( 1 \wedge (\eps/(t'-t))^{d-\frac{1}{2}}\big)  , 
\end{eqnarray*}  
where we used Lemma  \ref{lem:7} in the last inequality. 
 Then, using the fact that ${\cal K}\in L^1$, this gives 
\begin{eqnarray*}
\big|{\rm Cov} \big( (\vartheta^\eps_{t})^2 ,  (\vartheta^\eps_{t'})^2  \big)\big|
&\leq & C_2
 \left(\int_t^{T} ds   |{\cal K}^\eps(s-t)| \int_{t'}^T du  |{\cal K}^\eps(u-t)|  \right)^2
  \big( 1 \wedge (\eps/(t'-t)))^{\frac{d}{2}-\frac{1}{4}}\big)\\
&\leq & C_3 \eps^2  \big( 1 \wedge (\eps/(t'-t)))^{\frac{d}{2}-\frac{1}{4} }\big) ,
\end{eqnarray*} 
from which the lemma follows. 
 \end{proof}   
 
\begin{lemma}
\label{lem:7}%
Let $\sigma^\eps_{t,\infty}$ be defined by (\ref{def:sigmaepsst}).
Then there exists $C>0$ such that
\begin{eqnarray}\label{eq:sigi2}
\sigma^\eps_{t,\infty} \leq C \big( 1 \wedge (\eps/t)^{d-\frac{1}{2}}\big) .
\end{eqnarray}
\end{lemma}
\begin{proof}
By assumption there exists 
%$K>0$, $t_0\geq 1$ 
$K,t_0>0$ 
so that $|{\cal K}(t)| \leq K t^{-d}$ 
for $t \geq t_0$ with $d>1$.  Therefore, for $t \geq \eps t_0$: 
\ban
 \int_t^\infty  {\cal K}^\eps(s)^2 ds  \leq  
 \eps^{2d-1} \int_t^\infty  K^2 s^{-2d} ds  =
  \frac{K^2}{2d-1} \left( \frac{\eps}{t} \right)^{2d-1}   .
  % \leq \frac{K^2 t_0^{2d-1}}{2d-1} \left( \frac{\eps}{t} \right)^{2d-1}   .
\ean
For $t < \eps t_0$ we have  $\int_t^\infty  {\cal K}^\eps(s)^2 ds  \leq 1$ since ${\cal K}\in L^2(0,\infty)$ with a $L^2$-norm equal to one.
This gives the desired result.
\end{proof}  
  
%\begin{lemma}
%\label{lem:2}%
%For $t \in (0,T), k \in (0,1,\ldots,k_0)$
% $\exists K < \infty$ so that: 
% \ban
%   |(x\partial_x)^k Q_t^{(0)}(x)| < K . 
%\ean
%\end{lemma}
%\begin{proof}
%\end{proof}

Let $\tilde{X}_t $  be defined by (\ref{eq:tXdef0}).
%\begin{equation}\label{eq:tXdef}
%d\tilde{X}_t = \bar\sigma \tilde{X}_t dW^*_t , \quad \quad  \tilde{X}_0 =  {X}_0  .
%\end{equation}
%    Then it follows  
%  \ban
%     (x\partial_x)^k h(x) = ( \partial_y)^k h(y)
%  \ean  
%  for $y=\ln(x)$ and   
Then we have
  \begin{equation}
     Q_t^{(0)}(\tilde{X}_t) = \EE\left[ h(\tilde{X}_T) \mid {\cal F}_t \right]. 
     \label{eq:relietildeXQ}
  \end{equation}

 We finish this appendix with three   {\it effective market lemmas}:

\begin{lemma}
\label{lem:EMn}
Let $f(t,x)$ be smooth bounded and with bounded derivatives.
Let $X_t$  
be defined by Eq. (\ref{eq:Xdef}) and $\tilde{X}_t$  be defined by Eq. (\ref{eq:tXdef0}).
For $t,t'\in [0,T]$ we have
\begin{eqnarray}
&&
 \label{eq:lemEMn1}
\sup_{\eps \in (0,1]}    \eps^{-1/2} 
\EE\left[ \Big|  \EE[f(t,X_t)  - f(t,\tilde{X}_t)\mid {\cal F}_0 ] \Big|^2 \right]^{1/2}  <  \infty ,\\
&&
\sup_{\eps \in (0,1]}    \eps^{-1/2}  
\EE\left[ \Big|  \EE[f(t,X_t) f(t',X_{t'})   -  f(t,\tilde{X}_t) f(t',\tilde{X}_{t'})
\mid {\cal F}_0 ] \Big|^2 \right]^{1/2}  <  \infty .
\label{eq:lemEMn2}
\end{eqnarray} 
 \end{lemma}
 
\begin{proof}
 {First, we prove (\ref{eq:lemEMn1}). 
We apply Proposition \ref{prop:main} with $h(x)=f(t,x)$ and $T=t$ and we look at $M_0=\EE[h(X_T) \mid {\cal F}_0 ]=\EE[f(t,{X}_t) \mid {\cal F}_0 ]$. 
Then Proposition \ref{prop:main} gives that $\lim_{\eps\to 0} \eps^{-1/2}\EE[ |M_0 -P(0,X_0)|^2]^{1/2}=0$
with $P(0,x)=Q_0^{(0)}(x)+\eps^{1/2} \rho Q_0^{(1)}(x)$. Therefore we have
$\EE[f(t,{X}_t) \mid {\cal F}_0 ] = Q_0^{(0)}(X_0)+O(\sqrt{\eps})$.
Moreover, (\ref{eq:relietildeXQ}) gives $Q_0^{(0)}(\tilde{X}_0) =\EE[f(t,\tilde{X}_t) \mid {\cal F}_0 ]$.
This gives (\ref{eq:lemEMn1}) because $\tilde{X}_0=X_0$.\\
Second, we prove  (\ref{eq:lemEMn2}) for $t<t'$. 
We write
\begin{eqnarray*}
&&
\EE[f(t,X_t) f(t',X_{t'})   -  f(t,\tilde{X}_t) f(t',\tilde{X}_{t'}) \mid {\cal F}_0 ]\\
&&=
\EE\big[f(t,X_t) \EE[ f(t',X_{t'})|{\cal F}_t]   -  f(t,\tilde{X}_t) \EE[ f(t',\tilde{X}_{t'})|{\cal F}_t]  
\mid {\cal F}_0 \big] .
\end{eqnarray*}
We apply Proposition \ref{prop:main} with $h(x)=f(t',x)$ and $T=t'$ and we get
$$
\EE \big[ f(t',X_{t'})   \mid {\cal F}_t ] = {\cal Q}_t^{(0)} (X_t) +O(\sqrt{\eps}) ,
$$
where ${\cal Q}_{x}^{(0)}(x)$ satisfies ${\cal L}_{\rm BS}(\bar{\sigma}) {\cal Q}_s^{(0)}(x) = 0$ for $s\in [t,t')$ with ${\cal Q}^{(0)}_{t'}(x)=f(t',x)$.
We also have $\EE \big[ f(t',\tilde{X}_{t'})   \mid {\cal F}_t ] = {\cal Q}_t^{(0)} (\tilde{X}_t) $. Therefore
\begin{eqnarray*}
&&\EE[f(t,X_t) f(t',X_{t'})   -  f(t,\tilde{X}_t) f(t',\tilde{X}_{t'})
\mid {\cal F}_0 ] \\
&&=
\EE[f(t,X_t) {\cal Q}_t^{(0)} (X_t)   -  f(t,\tilde{X}_t) {\cal Q}_t^{(0)} (\tilde{X}_t) 
\mid {\cal F}_0 ]
+O(\sqrt{\eps}) ,
\end{eqnarray*}
and we can apply (\ref{eq:lemEMn1}) with the function $\tilde{f}(t,x)= f(t,x) {\cal Q}_t^{(0)}(x)$ to get the desired result (\ref{eq:lemEMn2}).} 
%The second result follows from 
%the first part noting that
%$(ab-\tilde{a}\tilde{b}) = b (a-\tilde{a})  +  \tilde{a}(b-\tilde{b})$.
\end{proof}

\begin{lemma}
\label{lem:EMn0}
Let $f(t,x)$ be smooth bounded and with bounded derivatives. Then  we have
for $0 \leq t \leq T$:
\begin{equation}
 \lim_{\eps \to 0}      \EE\left[
 \left(  \int_0^t  f(s,X_s)  \left(  (\sigma_s^\eps)^2  - \bar\sigma^2 \right)  ds \right)^2 \right]   = 0. 
 \label{eq:comparesigmasigmaeps}
\end{equation}
\end{lemma}
\begin{proof}
 This follows via an argument as in the proof of Eq. (\ref{eq:estimeRj}) for $j=2$
 as given in \cite{sv3}. 
\end{proof}

\begin{lemma}\label{lem:EMn2}
Let $f_j(t,x), \quad j=1,2$ be smooth bounded functions and with bounded derivatives and satisfying:
\ba\label{eq:BSf}
 {\cal L}_{\rm BS} (\bar\sigma)  f_j(t,x)  = 0 , \quad j=1,2 .
\ea
Let  $X_t$ be
defined by Eq.~(\ref{eq:Xdef}) and  $\tilde{X}_t$ be defined by
Eq.~(\ref{eq:tXdef0}). 
  Then   we have for $t \in (0,T)$:
 \begin{equation}\label{eq:EM1}
\lim_{\eps \to 0}   
\EE\left[ \left| \EE\Big[  \Big|  \int_0^t  f_1(s,X_s)   ds  \Big|^2   \mid {\cal F}_0\Big] 
 -
 2
\EE\Big[   \int_0^t  f_1(s, \tilde{X}_s)^2 (t-s)   ds   \mid {\cal F}_0  \Big] 
\right|
\right]   = 0  ,
\end{equation} 
\ba
\nn
&& \lim_{\eps \to 0}   \EE\left[ \left|
  \eps^{-1/2} 
\EE\Big[   \int_0^t  f_1(s,X_s)   ds   \int_0^t  f_2(s,X_s)   d\psi^\eps_s   \mid {\cal F}_0  \Big]   
\right.\right. 
\\ &&   ~~~~~~~~
\left. \left. - 
\rho  \overline{D}
\EE\Big[   \int_0^t (t-s) 
\left( (x \partial_x ) f_1(s, \tilde{X}_s) \right) f_2(s, \tilde{X}_s)  ds     \mid {\cal F}_0   \Big]
\right|
\right]  =0, 
\label{eq:EM2}
 \ea
 \begin{equation}\label{eq:EM3}
 \lim_{\eps \to 0}   \EE\left[ \left|  \eps^{-1}
\EE\Big[ \Big(  \int_0^t  f_1(s,X_s)   d\psi^\eps_s     \Big)^2  \mid {\cal F}_0 \Big] 
-
\overline{\Gamma}^2
\EE\Big[   \int_0^t  f_1(s, \tilde{X}_s)  ^2   ds  \mid {\cal F}_0   \Big] \right|
\right] =0 .
\end{equation} 
 \end{lemma}
 
\begin{proof}
{\it Proof of (\ref{eq:EM1}):}
Note first that in view of Lemma \ref{lem:EMn}, Eq.~(\ref{eq:lemEMn2}), we have
\ban
  \lim_{\eps \to 0}   
\EE\left[ \left| \EE\Big[\Big| \int_0^t  f_1(s,X_s)   ds  \Big|^2  \mid {\cal F}_0 \Big] 
-   \EE\Big[ \Big|  \int_0^t  f_1(s,\tilde{X}_s)   ds  \Big|^2  \mid {\cal F}_0 \Big]\right|\right]   =0.
 \ean
 Note next that in view of Eq.~(\ref{eq:BSf}) $f_1(s, \tilde{X}_s)$ is a martingale so that
 \ban
 &&   \EE\left[ \Big|  \int_0^t  f_1(s,\tilde{X}_s)   ds  \Big|^2  \mid {\cal F}_0 \right]   =
   \EE\left[ 2  \int_0^t  f_1(s,\tilde{X}_s) \int_s^t   f_1(u,\tilde{X}_u) du  ds   \mid {\cal F}_0   \right] 
  \\ && 
       =  \EE\left[ 2  \int_0^t f_1(s, \tilde{X}_s)^2 (t-s)  ds   \mid {\cal F}_0\right] ,
 \ean
 which gives  Eq.~(\ref{eq:EM1}).  
 
{\it Proof of (\ref{eq:EM2}): }
It follows from the fact that $\int_0^t f_2(u,X_u) d\psi^\eps_u$ is a martingale and It\^o's Lemma that 
  \ba
  \nonumber
  && \EE\left[   \int_0^t  f_1(s,X_s)   ds   \int_0^t  f_2(s,X_s)   d\psi^\eps_s  \mid {\cal F}_0   \right]\\
&&  =    
  \nonumber    \int_0^t     \EE\left[    
 \EE\big[  \int_0^t  f_2(u,X_u)   d\psi^\eps_u \mid {\cal F}_s \big]
      f_1(s,X_s)      \mid {\cal F}_0   \right]      ds
    \\ 
  \nonumber  
  &&=       \int_0^t      \EE\left[   
  \int_0^s  f_2(u,X_u)   d\psi^\eps_u 
      f_1(s,X_s)     \mid {\cal F}_0   \right]       ds
    \\ 
      \nonumber
&&   
     =      \int_0^t         \EE\left[  
  \int_0^s  f_2(u,X_u)   d\psi^\eps_u 
     \int_0^s  (x^2 \partial^2_x )f_1(u,X_u)  
       \frac{1}{2} \left(  (\sigma_u^\eps)^2  - \bar\sigma^2 \right)  du    \mid {\cal F}_0 \right]  ds  
            \\
            \nonumber
             && \quad 
    \hbox{} +  \int_0^t     \EE\left[  
  \int_0^s  f_2(u,X_u)   d\psi^\eps_u 
     \int_0^s  (x \partial_x )f_1(u,X_u)   \sigma_u^\eps dW^*_u   \mid {\cal F}_0 \right]   ds  \\
      && \quad 
    \hbox{} +  \int_0^t     \EE\left[  
  \int_0^s  f_2(u,X_u)   d\psi^\eps_u 
    f_1(0,X_0)    \mid {\cal F}_0 \right]   ds            .
\label{eq:EM2:pr1}
  \ea
The last term of Eq.~(\ref{eq:EM2:pr1}) is zero because $\int_0^t f_2(u,X_u) d\psi^\eps_u$ is a zero-mean martingale.
It follows from    Lemmas \ref{lem:K1} and  \ref{lem:EMn0}    that 
\ba
\nonumber
  && \EE\left[  \left|  \EE\Big[    \eps^{-1/2}   
  \int_0^s  f_2(u,X_u)   d\psi^\eps_u 
     \int_0^s (x^2 \partial^2_x) f_1(u,X_u)  
       \frac{1}{2} \big(  (\sigma_u^\eps)^2  - \bar\sigma^2 \big)  du   \mid {\cal F}_0\Big]\right|
    \right]  
      \\ 
      \nonumber &&    \leq
    \EE\left[ \eps^{-1}  \left(\int_0^s  f_2(u,X_u)   d\psi^\eps_u  \right)^2 \right]^{1/2}
    \EE\left[   \left( \int_0^s  ( x^2 \partial^2_x ) f_1(u,X_u) 
       \frac{1}{2} \left(  (\sigma_u^\eps)^2  - \bar\sigma^2 \right)  du    \right)^2
    \right]^{1/2}\\
    \nonumber
    &&    =
    \EE\left[ \eps^{-1}  \int_0^s  f_2(u,X_u)   (\vartheta^\eps_u)^2 du \right]^{1/2}
    \EE\left[   \left( \int_0^s  ( x^2 \partial^2_x ) f_1(u,X_u) 
       \frac{1}{2} \left(  (\sigma_u^\eps)^2  - \bar\sigma^2 \right)  du    \right)^2
    \right]^{1/2}
     \\  &&    \stackrel{\eps \to  0}{\longrightarrow}   0    .
                 \label{eq:EM2:pr2}
    \ea 
By Lemma  \ref{lem:1b} we have
 \ban
   && \EE\Big[ \eps^{-1/2}  \int_0^s  f_2(u,X_u)   d\psi^\eps_u 
     \int_0^s (x \partial_x) f_1(u,X_u)   \sigma_u^\eps dW^*_u    \mid {\cal F}_0  \Big]
     \\ && \mbox{} 
   \quad   - \rho  \EE\Big[   \int_0^s  f_2(u,X_u)( x \partial_x) f_1(u,X_u)    \overline{D}
     du   \mid {\cal F}_0\Big]
     \\ &&  
      =   \rho \EE\Big[  \int_0^s  f_2(u,X_u) ( x \partial_x) f_1(u,X_u) 
      \big( \eps^{-1/2}  \sigma_u^\eps  \vartheta^\eps_{u}   -  \overline{D}\big)
     du  \mid {\cal F}_0 \Big] .
         \ean   
By Lemma \ref{lem:K2}, Eq.~(\ref{a31}), we get
 \ban
 \nonumber
 &&
 \EE\left[  \left| \EE\Big[  \int_0^s  f_2(u,X_u) ( x \partial_x) f_1(u,X_u) 
      \big( \eps^{-1/2}  \sigma_u^\eps  \vartheta^\eps_{u}   -  \overline{D}\big)
     du  \mid {\cal F}_0 \Big]  \right| \right] \\
&&     \leq
    \EE\left[ \Big(  \int_0^s  f_2(u,X_u) ( x \partial_x) f_1(u,X_u) 
      \big( \eps^{-1/2}  \sigma_u^\eps  \vartheta^\eps_{u}   -  \overline{D}\big)
     du \Big)^2\right]^{1/2} 
            \stackrel{\eps \to  0}{\longrightarrow}   0  .
            \ean
By Lemma \ref{lem:EMn}, Eq.~(\ref{eq:lemEMn1}), we find
\ban
&&
\lim_{\eps\to 0}
 \EE\left[  \left|
 \EE\Big[   \int_0^s  f_2(u,X_u)( x \partial_x) f_1(u,X_u)  
     du   \mid {\cal F}_0\Big] \right.\right. \\
     && \hspace*{0.5in}
   \left.\left.  -
 \EE\Big[   \int_0^s  f_2(u,\tilde{X}_u)( x \partial_x) f_1(u,\tilde{X}_u)         
  du   \mid {\cal F}_0\Big] \right|\right] =0     .
\ean          
Therefore
\ba
\nonumber
  && \EE\left[\left| 
  \EE\Big[ \eps^{-1/2}  \int_0^s  f_2(u,X_u)   d\psi^\eps_u 
     \int_0^s (x \partial_x) f_1(u,X_u)   \sigma_u^\eps dW^*_u    \mid {\cal F}_0  \Big] \right.\right.
     \\ && \mbox{} 
   \quad    \left.\left.  - \rho  \overline{D}
 \EE\Big[   \int_0^s  f_2(u,\tilde{X}_u)( x \partial_x) f_1(u,\tilde{X}_u)  
     du   \mid {\cal F}_0\Big] \right|\right] =0     .
     \label{eq:EM2:pr3}
         \ea
We substitute (\ref{eq:EM2:pr2}-\ref{eq:EM2:pr3}) into (\ref{eq:EM2:pr1})
and we  invoke Lebesgue's dominated convergence theorem (because $\eps^{-1/2} \vartheta^\eps_u$ is uniformly bounded)
to prove (\ref{eq:EM2}).

{\it Proof of (\ref{eq:EM3}):} 
        The result follows from Lemmas  \ref{lem:1b}, \ref{lem:K2}
        and \ref{lem:EMn}. 
   \end{proof}

\section{The fOU Volatility Factor}\label{app:fOU}

We use a rapid fractional Ornstein-Uhlenbeck (fOU) process as the volatility factor and
describe here how this process can be represented in terms of a fractional Brownian motion.
Since fractional Brownian motion can be expressed in terms of ordinary Brownian motion
we also arrive at an expression for the rapid fOU process as a filtered version
of Brownian motion.

A fractional Brownian motion (fBM) is a zero-mean Gaussian process $(W^H_t)_{t\in \RR}$  
with the covariance
\begin{equation}
\label{eq:covfBM}
\EE[ W^H_t W^H_s ] = \frac{\sigma^2_H}{2} \big( |t|^{2H} + |s|^{2H} - |t-s|^{2H} \big) ,
\end{equation}
where $\sigma_H$ is a positive constant.
We use the following moving-average stochastic integral representation of the
fBM \cite{mandelbrot}:
\begin{equation}\label{mandelbrot}
W^H_t = \frac{1}{\Gamma(H+\frac{1}{2})} 
\int_{\RR} (t-s)_+^{H - \frac{1}{2}} -(-s)_+^{H - \frac{1}{2}} dW_s ,
\end{equation}
where $(W_t)_{t \in \RR}$ is a standard Brownian motion over $\RR$.
Then  indeed  $(W^H_t)_{t\in \RR}$ is a zero-mean Gaussian process with
the covariance (\ref{eq:covfBM}) and we have
\begin{eqnarray}
%\nonumber
\sigma^2_H
% &=& \frac{1}{\Gamma(H+\frac{1}{2})^2}  \Big[ \int_0^\infty \big( (1+s)^{H - \frac{1}{2}} -s^{H - \frac{1}{2}} \big)^ 2 ds +\frac{1}{2H}\Big] \\
&=&  \frac{1}{\Gamma(2H+1) \sin(\pi H)}.
\end{eqnarray}

We introduce the $\eps$-scaled fractional Ornstein-Uhlenbeck process (fOU) as 
\begin{equation}\label{eq:fOU}
Z^\eps_t = \sqrt{2\sin (\pi H)} \sigma_{{\rm z}}  \eps^{-H} \int_{-\infty}^t e^{-\frac{t-s}{\eps}} dW^H_s  .
\end{equation}
The fractional OU process can be seen as a fractional Brownian motion with a restoring force towards zero.
It is a zero-mean, stationary Gaussian process,
with variance
\begin{equation}
\label{eq:sou}
\EE [ (Z^\eps_t)^2 ] = \sigma^2_{{\rm z}},  
\end{equation}
that is independent of $\eps$,
and covariance:
\begin{eqnarray}\label{eq:sou2}
\EE [ Z^\eps_t Z^\eps_{t+s}  ] &=&\sigma^2_{{\rm z}} {\cal C}_Z\Big(\frac{s}{\eps}\Big) ,
\end{eqnarray}
that is a function of $s/\eps$ only, with
\begin{eqnarray}
\nn
{\cal C}_Z(s) &=& 
 \frac{1}{\Gamma (2H+1)}
 \Big[ \frac{1}{2} \int_\RR e^{- |v|}| s+v|^{2H} dv - |s|^{2H}\Big]
\\
&=&
\frac{2\sin (\pi H)}{\pi}
\int_0^\infty \cos (s x  ) \frac{x^{1-2H}}{1+x^2} dx .
\label{def:calCZ}
\end{eqnarray}
This shows that $\eps$ is the natural scale of variation of the fOU $Z^\eps_t$.
Note that the random process $Z^\eps_t$ is not  a martingale, neither a Markov process.
For $H \in  (0,1/2)$ it possesses short-range correlation properties in the sense that its correlation
function is rough at zero as seen in (\ref{eq:OUa1}) 
while it is integrable and it decays as $s^{2H-2}$ at infinity
as seen in (\ref{eq:OUa2}).

Using Eqs. (\ref{mandelbrot}) and (\ref{eq:fOU}) we arrive at
the moving-average integral representation of the scaled fOU as:
\begin{equation}
\label{eq:fOU2}
Z^\eps_t = \sigma_{{\rm z}} \int_{-\infty}^t {\cal K}^\eps(t-s) dW_s,
\end{equation}
where ${\cal K}^\eps$ is of the form (\ref{def:Keps})-(\ref{def:Keps0}).
The   kernel ${\cal K}$ 
satisfies  the assumptions set forth in Section \ref{sec:svmodel} and the main properties
 are the following ones (valid for any $H \in (0,1/2)$):
\begin{enumerate}
\item[(i)]
${\cal K} \in L^2(0,\infty)$ with $\int_0^\infty {\cal K}^2(u) du = 1$
and ${\cal K} \in L^1(0,\infty)$.
\item[(ii)] For small times $t \ll 1$:
\begin{equation}
{\cal K} (t) = \frac{\sqrt{2 \sin(\pi H)}}{\Gamma(H+\frac{1}{2})  } 
\Big( t^{H - \frac{1}{2}}  + O\big(  t^{H +\frac{1}{2}} \big) \Big) .
\end{equation}
\item[(iii)] For large times $t \gg 1$:
\begin{equation}
{\cal K} (t) = \frac{\sqrt{2 \sin(\pi H)}}{\Gamma(H-\frac{1}{2})}  
\Big( t^{H - \frac{3}{2}}  +O\big(  t^{H - \frac{5}{2}} \big) \Big) .
\end{equation}
\end{enumerate}

The volatility process $\sigma_t^\eps$ defined by (\ref{def:stochmodel}) inherits 
the short-range correlation properties of  the volatility driving process  $Z_t^\eps$.
This follows from the following lemma proved in \cite{sv3}: 
  \begin{lemma}
\label{prop:process}
We denote, for $j=1,2$:
\begin{equation}
\label{def:meanF}
\left< F^j \right> =  \int_{\RR}  F(\sigma_{{\rm z}}z)^j p(z) dz,
\quad \quad 
\left< {F'}^j \right> =  \int_{\RR}  F'(\sigma_{{\rm z}}z)^j p(z) dz,
\end{equation}
where $p(z)$ is the pdf of the standard normal distribution.

\begin{enumerate}
\item
The process $\sigma_t^\eps$ is a
stationary random process with mean $\EE[\sigma_t^\eps] = \left<F\right>$ and variance $
{\rm Var}(\sigma_t^\eps) = \left<F^2\right>-\left<F\right>^2$, independently of $\eps$.

\item
The covariance function of the process $\sigma_t^\eps$ is of the form
\begin{eqnarray}
\label{eq:corrY1}
{\rm Cov}\big( \sigma_t^\eps , \sigma_{t+s}^\eps \big) &=& \big( \left<F^2\right>-\left<F\right>^2\big) {\cal C}_\sigma \Big(\frac{s}{\eps}\Big),
\end{eqnarray}
where the correlation function ${\cal C}_\sigma$ satisfies ${\cal C}_\sigma(0)=1$ and 
\begin{eqnarray}
{\cal C}_\sigma(s) &=&
1 -  \frac{1   }{  \Gamma(2H+1)}
\frac{ \sigma_{{\rm z}}^2\left< {F' }^2\right> }{ \left< F^2\right>-\left<F\right>^2}
s^{2H} +o\big( s^{2H}  \big)  ,
\quad \mbox{ for  } s \ll 1 ,
\label{eq:corrY12}\\
{\cal C}_\sigma(s) &=&
 \frac{1   }{  \Gamma(2H-1)}
\frac{ \sigma_{{\rm z}}^2\left< F' \right>^2 }{ \left< F^2\right>-\left<F\right>^2}
s^{2H-2} +o\big( s^{2H-2}  \big)  ,
\quad \mbox{ for  } s \gg 1 .
\label{eq:corrY12b} 
\end{eqnarray}
\end{enumerate}
\end{lemma}
Consequently,  the process $\sigma_t^\eps$ has short-range correlation properties
and its covariance function is integrable.

 \section{Call Option Hedging Cost and Risk}\label{app:mom}

  We use below the following ``Greek'' identities for the European  call case:
\ba\label{eq:greek1}
\frac{ x^2 \partial^2_x Q^{(0)}_t(x) }{K} &=&  
\frac{  \partial_\sigma Q^{(0)}_t(x) }{K \bar\sigma (T-t)}
=
 \frac{ (x/K)  e^{-d_+^2(x,t)/2} }{\sqrt{2\pi   \tau_t}}   =
      \frac{  e^{-d_{_{-}}^2(x,t)/2} }{\sqrt{2\pi   \tau_t}}  ,\\ 
\frac{x\partial_x x^2 \partial^2_x Q^{(0)}_t(x) }{K}  &=&
   -    \frac{   d_{_{-}}(x,t) e^{-d_{_{-}}^2(x,t)/2} }{\sqrt{2\pi  } \tau_t} ,  \label{eq:greek} \\
   \frac{(x\partial_x)^2  x^2 \partial^2_x Q^{(0)}_t(x) }{K}  &=&
      \frac{ \big( d^2_{_{-}}(x,t)   -1  \big) e^{-d_{_{-}}^2(x,t)/2} }{\sqrt{2\pi  } \tau_t^{3/2}}   \label{eq:greek2}  ,
\ea
with $\tau_t  =  {(T-t)}{\bar\sigma^2}  $ and
$$
   d_{\pm}(x,t)  =  \frac{ \log(x/K)  }{ \sqrt{\tau_{t}} }   
   \pm \frac{ \sqrt{\tau_{t}} }{2}  .
$$
We will also use the following lemma:   
\begin{lemma}
\label{lem:gauss}
\ba
&& \EE\left[    
e^{-d_{_{-}}^2(\tilde{X}_{Ts},Ts)}   \mid {\cal F}_0  \right]   = 
\exp\Big(-\frac{d_{_{-}}^2(X_0,0)}{1+s}\Big)  f_0(s)  , \\
&& f_0(s)   =  \sqrt{ \frac{1-s}{1+s}  }   , \\
&&
  \EE\left[  d_{_{-}}^2(\tilde{X}_{Ts},Ts) 
e^{ -d_{_{-}}^2(\tilde{X}_{Ts},Ts) }    \mid {\cal F}_0  \right]     
= \exp\Big(-\frac{d_{_{-}}^2(X_0,0)}{1+s}\Big) f_2\big(s, d_{_{-}}(X_0,0)\big) , \\  
  &&  f_2(s, d) =     
  d^2
\sqrt{\frac{(1-s)^3}{(1+s)^5}  }  +
  s \sqrt{ \frac{1-s}{(1+s)^3}  }  
   ,  \\
&&
  \EE\left[  d_{_{-}}^4(\tilde{X}_{Ts},Ts) 
e^{ -d_{_{-}}^2(\tilde{X}_{Ts},Ts) }   \mid {\cal F}_0   \right]   
  =  
\exp\Big(-\frac{d_{_{-}}^2(X_0,0)}{1+s}\Big)  f_4\big(s, d_{_{-}}(X_0,0)\big)  ,  \label{eq:4th}
\\ &&  f_4(s, d) =  
d^4
\sqrt{\frac{(1-s)^{5}}{(1+s)^{9}}  }   +
 6 d^2 s 
 \sqrt{\frac{(1-s)^{3}}{(1+s)^{7}}  }   +
  3s^2  \sqrt{ \frac{  1-s }{(1+s)^{5}}  }  
   .  
\ea
\end{lemma}
\begin{proof}
By (\ref{eq:tXdef0}), for any $t$, we have  in distribution
$$
\tilde{X}_t = X_0 \exp\Big(\bar{\sigma} W_t^* -\frac{1}{2}\bar{\sigma}^2 t\Big)
=X_0 \exp\Big(\bar{\sigma} \sqrt{t} Z -\frac{1}{2} \bar{\sigma}^2 t\Big)
,
$$
   with $Z$ having the standard normal distribution,
   and therefore
\ban
   d_{_{-}}(\tilde{X}_{t},t)  =  \frac{ \log(\tilde{X}_{t}/K)  }{ \sqrt{\tau_{t}} }   
   - \frac{ \sqrt{\tau_{t}} }{2}
   =  \frac{ d_{_{-}}(X_0,0) +Z \sqrt{t/T}}{  \sqrt{1-t/T} } .
     \ean
One  can then carry  out the resulting Gaussian integral upon a completion
of the square in the exponential.
\end{proof}

\section{{Vega Positivity}}\label{app:pos}

The Black-Scholes price at volatility $\sigma$ can be written as $Q^{(0)}(t,x)=u(T-t,x)$ where $u$ is the solution of
\ban
&& \partial_\tau u =    \frac{1}{2} \sigma^2  x^2   \partial_x^2 u      , \quad  x >0,  \tau  > 0 , \\
 && u(0,x)  = h(x) , \quad u(\tau,0) = h(0) ,
\ean
with  $h$ the payoff function.  
Let $v(\tau,y)=u(\tau,e^y)$. The function $v$ solves
\ban
&& \partial_\tau  v  =    \frac{1}{2}  \sigma^2  \left(  \partial_{y}^2  -  {\partial_{y}}\right)   v      , \quad   y \in \RR,  \tau   > 0 , \\
 && v(0,y)  = h(e^y) , \quad \lim_{y \to -\infty} v(\tau,y) = h(0) .
\ean
By Fourier transform, it can be expressed as
\ban
    v (\tau,y) = \frac{1}{2\pi}   \int_\RR e^{\frac{ \sigma^2}{2}(-\omega^2+i \omega) \tau } \ch (\omega) e^{-i \omega y} d\omega  ,
 \mbox{ where }   \ch (\omega) = \int_\RR  h(e^y)  e^{i \omega y} dy  . 
\ean
Thus
$$
 \partial_\sigma   Q^{(0)}(T-\tau,e^y) =  \partial_\sigma  v(\tau,y) =
      \frac{ \sigma\tau}{2\pi}   \int_\RR (-\omega^2+i\omega) e^{\frac{ \sigma^2}{2}(-\omega^2+i \omega) \tau } \ch (\omega) e^{-i \omega y} d\omega .
 $$     
We have
\begin{align*}
    \int_\RR (-\omega^2+i\omega) e^{- \frac{ \sigma^2}{2}\omega^2 \tau } \ch (\omega) e^{-i \omega y} d\omega 
&=   (\partial_{y}^2 - \partial_{y} )  \iint_{\RR^2} h(e^{y'})  e^{- \frac{ \sigma^2}{2}\omega^2 \tau }  e^{i \omega (y'-y)} d\omega dy' \\
&=  2\pi (\partial_{y}^2 - \partial_{y} )   \EE \left[  h(e^{ y+ \sigma\sqrt{\tau} X}) \right]  ,
\end{align*}
where $X$  is a standard Gaussian random variable. 
By combining the last two identities with $x=e^y$, we find
\ban
 \partial_\sigma   Q^{(0)}(T-\tau,x) = g( x  )  , \mbox{ where }g(x) =   \sigma\tau x^2 \partial_x^2 \EE\left[  h(xe^{ \sigma\sqrt{\tau} X}) \right] .
\ean 
Since $h$ is convex and not affine $x\mapsto  \EE\left[  h(xe^{ \sigma \sqrt{\tau}X}) \right] $
is smooth and strictly convex and hence the Vega is positive.

\end{document}